\documentclass[reqno,11pt]{amsart}
\usepackage{amssymb,amsmath}
\usepackage{color} 
\usepackage{accents} 
\usepackage{texdraw}
\usepackage{eucal}
\usepackage{epic,epsfig}
\usepackage{graphicx}
\usepackage[all]{xy}
\newtheorem{theorem}{Theorem}[section]
\newtheorem{lemma}[theorem]{Lemma}
\newtheorem{cor}[theorem]{Corollary}
\newtheorem{prop}[theorem]{Proposition}
\numberwithin{equation}{section}
\DeclareMathAccent{\wtilde}{\mathord}{largesymbols}{"65}
\DeclareMathAccent{\what}{\mathord}{largesymbols}{"62}

\def\m@th{\mathsurround=0pt}
\mathchardef\bracell="0365 
\def\upbrall{$\m@th\bracell$}
\def\undertilde#1{\mathop{\vtop{\ialign{##\crcr
    $\hfil\displaystyle{#1}\hfil$\crcr
     \noalign
     {\kern1.5pt\nointerlineskip}
     \upbrall\crcr\noalign{\kern1pt
   }}}}\limits}

\newcommand{\wb}[1]{\overline{#1}}

\newcommand{\wh}{\widehat}
\newcommand{\wt}{\widetilde}

\newcommand{\Equiv}{:=}

\newcommand{\cQ}{\mathcal{Q}} 
\newcommand{\cH}{\mathcal{H}}
\newcommand{\cA}{\mathcal{A}}

\newcommand{\DONE}{\\$\blacksquare$\par\smallskip\smallskip\noindent\ignorespacesafterend}

\newcommand{\ssp}{\mathfrak{p}}

\newcommand{\ssk}{\mathfrak{k}}

\newcommand{\ssl}{\mathfrak{l}}

\newcommand{\ssq}{\mathfrak{q}}

\newcommand{\ssr}{\mathfrak{r}}

\newcommand{\veta}{t}

\newcommand{\bun}{\boldsymbol{1}}

\newcommand{\tbu}{\,^{t\!}{\bu}}
\newcommand{\tbs}{\,^{t\!}{\boldsymbol{s}}}
\newcommand{\bblu}{\begin{color}{blue}}
\newcommand{\bred}{\begin{color}{red}}
\newcommand{\ecl}{\end{color}}

\newcommand{\bA}{\boldsymbol{A}}

\newcommand{\bK}{\boldsymbol{K}} 
\newcommand{\bL}{\boldsymbol{L}} 
\newcommand{\bM}{\boldsymbol{M}}

\newcommand{\aR}{\alpha}
\newcommand{\bb}{\beta}
\newcommand{\gm}{\gamma}

\newcommand{\ven}{\varepsilon}

\newcommand{\dd}{\delta}
\newcommand{\sg}{\sigma}

\newcommand{\kp}{\kappa}
\newcommand{\ld}{\lambda}

\newcommand{\oa}{\omega}
\newcommand{\be}{\begin{equation}}
\newcommand{\ee}{\end{equation}}
\newcommand{\bea}{\begin{eqnarray}}
\newcommand{\eea}{\end{eqnarray}}
\newcommand{\bse}{\begin{subequations}}
\newcommand{\ese}{\end{subequations}}
\newcommand{\nn}{\nonumber}

\newcommand{\ol}{\overline}

\newcommand{\po}{\accentset{o}{p}}
\newcommand{\qo}{\accentset{o}{q}}
\newcommand{\Po}{\accentset{o}{P}}
\newcommand{\Qo}{\accentset{o}{Q}}
\newcommand{\Do}{\accentset{o}{\delta}}

\newcommand{\bu}{\boldsymbol{u}}

\newcommand{\ba}{{\boldsymbol a}}
\newcommand{\bbb}{\boldsymbol{b}} 
\newcommand{\bc}{\boldsymbol{c}}

\newcommand{\bk}{\boldsymbol{k}} 
 
\newcommand{\brr}{\boldsymbol{r}} 

\newcommand{\mbx}{{\boldsymbol x}}
\newcommand{\mby}{{\boldsymbol y}}

\newcommand{\bs}{{\boldsymbol s}}

\newcommand{\bphi}{{\boldsymbol \phi}}

\newcommand{\tbphi}{\,^{t\!}\boldsymbol{\bphi}}

\begin{document}
\title[ Elliptic solutions of ABS lattice equations]
{Elliptic solutions of ABS lattice equations}

\author{Frank W Nijhoff }
\address{
School of Mathematics\\
University of Leeds\\ 
Leeds LS2 9JT\\ 
United Kingdom}
\email{nijhoff@maths.leeds.ac.uk}

\author{James Atkinson}
\address{
School of Mathematics and Statistics\\ 
University of Sydney\\ 
NSW 2006\\
Australia} 
\email{jamesa@maths.usyd.edu.au}

%\thanks{The authors are grateful for the hospitality of the Isaac Newton Institute for Mathematical Sciences, Cambridge, where the 
%present work was commenced during the programme Discrete Integrable Systems (DIS)}

\begin{abstract}
Elliptic $N$-soliton-type solutions, i.e. solutions emerging from the application of $N$ consecutive 
B\"acklund transformations to an elliptic {\it seed} solution, are constructed for all equations
in the ABS list of quadrilateral lattice equations, except for 
the case of the Q4 equation which is treated elsewhere. The 
main construction, which is based on an elliptic Cauchy matrix, 
is performed for the equation Q3, and by coalescence on certain auxiliary parameters, 
the corresponding solutions of the remaining equations in the list are obtained. Furthermore, the 
underlying linear structure of the equations is exhibited, leading, in 
particular, to a novel Lax representation of the Q3 equation. 

\end{abstract}

\maketitle

\section{Introduction}
 
In a series of recent papers, soliton type solutions of two-dimensional integrable quadrilateral lattice equations 
were explored, \cite{AHN,AHN2,AN,NAH}. Integrability here is understood in the sense of the multidimensional 
consistency property, \cite{NW,BS}, which states that such equations can be consistently embedded in a multidimensional lattice, and 
which has been accepted as a key integrability characteristic.
Several examples of equations exhibiting this property were known for more than two decades, cf. e.g. 
\cite{WE,Hir,DJM,NQC,QNCL}, but it was only recently that a full list of \textit{scalar} multidimensionally consistent equations was obtained, 
cf. \cite{ABS}, and also \cite{ABS2}. Apart from lattice equations \textit{of KdV type} which had been established early on, a number 
of novel equations arose from this classification, for which no further structures (such as Lax pairs, inverse scattering scheme or 
direct linearization treatments) were a priori known. As was demonstrated in \cite{BS,Nij} some of these structures, in fact, follow 
from the multidimensional consistency of the equation itself. However, the construction of explicit solutions remained an open problem, and 
this was systematically undertaken in the series of papers mentioned above, and they reveal some surprising new features. In \cite{NAH} 
we revealed some of the underlying structures of the equations in the ABS list, showing that they are deeply interrelated not only through 
degeneration, but through Miura type relations as well. In particular, closed-form $N$ solitons for the whole ABS list were obtained, with the 
exception of the ``top'' equation in the list, the so-called Q4 equation, which was first discovered by V. Adler, \cite{Adler}, as the 
permutability condition for the B\"acklund transformations of the famous Krichever-Novikov equation.
Recently we have given explicit $N$-soliton solutions for Adler's equation \cite{AN2} using a new constructive approach in which the solutions emerge in {\it Hirota form}.
With respect to the Cauchy-matrix approach considered here, the equation Q4 has special features which warrant a separate treatment. 
The main focus in \cite{NAH}, cf. also \cite{AHN2}, was on the construction of ``rational soliton solutions'' for the Q3 equation, which turns out to be 
the equation in which the features of the entire remaining list culminate. Its underlying structure incorporates ingredients from all the 
other equations in the ABS list, and consequently from its explicit $N$-soliton solutions, the corresponding solutions of 
all the ``lower'' equations in the ABS list are obtained through degenerations (i.e. via limits on some fixed parameters). Thus, we 
were able to present explicitly all such rational N-soliton solutions for those equations in closed form, (cf. also \cite{HZ} for 
bilinear form and Casorati determinant expressions of the same solutions).

In this paper we will generalise the Cauchy-matrix approach of \cite{NAH} to the elliptic case, i.e. to solutions which we 
could call \textit{elliptic N-soliton solutions}. The problem of finding such elliptic solutions should not be 
confused with the one of finding soliton solutions of the elliptic equation Q4 (in which the lattice parameters take values as points on an elliptic 
curve) considered in \cite{AN2}. The treatment here follows 
closely the derivations in \cite{NAH}, but the lifting of the rational solutions to the elliptic case introduces some new interesting features, 
such as the emergence of a non-autonomous Cauchy kernel that incorporates the core strucuture behind these solutions. Furthermore, some aspects 
become actually more transparant in the elliptic soliton case, as the elliptic addition formulae at the heart of the development in some sense 
drive the construction. 
Moreover, we unravel in this paper some of the underlying linear structures behind these solutions, which as a byproduct leads to the construction of 
a novel Lax pair for Q3, which seems more natural than the Lax pairs that are obtained from multidimensional consistency following the recipe in 
\cite{Nij,BS}. Thus, we expect that this 
new Lax representation can prove useful in the construction of wider classes of solutions, such as finite-gap solutions, and inverse scattering solutions 
with radiation. We also envisage that the structures revealed in this paper may prove important 
in understanding the general nature of elliptic solutions of integrable discrete equations and the corresponding many-body 
systems of Calogero-Moser and Ruijsenaars type, cf. also \cite{NRK,KrichWieg}. 
  
\section{Preliminaries: the ABS list}
First, we need to establish some notations that we will employ throughout the paper. 
Equations within the class of \textit{quadrilateral} partial difference equations (P$\Delta$Es) have the following canonical form: 
\[ \cQ_{\ssp,\ssq~}(u,\tilde{u},\hat{u},\hat{\tilde{u}}) = 0\  ,  \] 
where we adopt the short-hand notation of vertices along an 
elementary plaquette on a rectangular lattice: 
$$ u=u_{n,m}\quad,\quad\wt{u}=u_{n+1,m}\quad,\quad
\wh{u}=u_{n,m+1}\quad,\quad\wh{\wt{u}}=u_{n+1,m+1} $$
Schematically these are indicated in Figure \ref{quad}.
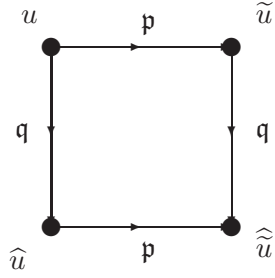
\begin{figure}[ht]
\begin{center}  

\setlength{\unitlength}{0.8mm}
\begin{picture}(40,40)(10,-5) 

\put(0,0){\circle*{3}}
\put(0,30){\circle*{3}}
\put(30,30){\circle*{3}}
\put(30,0){\circle*{3}}

\put(0,0){\vector(1,0){15}}
\put(15,0){\vector(1,0){15}}
\put(0,30){\vector(0,-1){15}}
\put(0,15){\vector(0,-1){15}}
\put(0,30){\vector(1,0){15}}
\put(15,30){\vector(1,0){15}}
\put(30,30){\vector(0,-1){15}}
\put(30,15){\vector(0,-1){15}}

\put(-7,-7){$\wh{u}$}
\put(34,-5){$\wh{\wt{u}}$}
\put(34,34){$\wt{u}$}
\put(-5,34){$ u$}

\put(15,-5){$\ssp$}
\put(15,33){$\ssp$}
\put(-6,15){$\ssq$}
\put(34,15){$\ssq$}
\end{picture} 
\end{center} 
\caption{Arrangement of the shifted dependent variable on the vertices of a quadrilateral and association of the lattice parameters to the edges.}
\label{quad}
\end{figure}
Here $\ssp$, $\ssq$ are parameters of the equation, related to the lattice 
spacing, so each parameter is identified with a direction in the lattice. It is useful, and actually crucial for the classification 
problem as solved in \cite{ABS}, to consider the lattice parameters 
$\ssp,\ssq$ as points on an algebraic curve, i.e. given by  ~$\ssp=(p,P)$~~,~
$\ssq=(q,Q)$~~, where the coordinates $p,P$ and $q,Q$ respectively are related through some polynomial equation. 

ABS, in \cite{ABS}, considered the case where $u$ is a \textit{scalar} (i.e. single-field) quantity, and restricted themselves to the 
affine-linear case, i.e. the case in which the function $\cQ$ is assumed to be affine linear in each of its four arguments.
They required the {\it multidimensional consistency} property, cf. \cite{NW,BS}.
That is, in each pair of directions of the multidimensional lattice a copy of the equation (chosing values of the lattice parameters 
$\ssp$, $\ssq$, etc., associated with those directions) can be imposed such that the iteration of well-chosen initial values leads to a unique 
determination of the solution in each lattice point, avoiding possible multivaluedness that would generically occur in the iteration of 
the solution by evaluating the values of $u$ on lattice points following different configurations of quadrialterals. In the integrable 
case, such multivaluedness does not occur by the precise structure of the equation and the combinatorics of the evaluation process. 
More concretely, this means that any three copies of the equation with different parameters, namely 
$$ \cQ_{\ssp,\ssq}(u,\wt{u},\wh{u},\wh{\wt{u}})=0\quad,\quad\cQ_{\ssp,\ssr}(u,\wt{u},\ol{u},\wt{\ol{u}})=0\quad,\quad
\cQ_{\ssq,\ssr}(u,\wh{u},\ol{u},\wh{\ol{u}})=0\  $$ 
can be simultaneously solved, and lead to a unique and single-valued determination of the triple-shifted point $\wh{\wt{\ol{u}}}$ at the vertex of an 
elementary cube on which initial conditions are given by $u$, $\wt{u}$, $\wh{u}$, $\ol{u}$, where the $\ol{\phantom{a}}$ denotes the elementary 
shift in an additional lattice direction associated with the lattice parameter $\ssr$. 
 
The classification in \cite{ABS} yielded, among others, the following list of equations (presented in a slightly different notation)
\vspace{.2cm}

\noindent
\framebox[2cm][l]{\bf Q-list:} 
\bse\label{eq:Qeqs}\begin{eqnarray}
%&&{\rm Q1}: \quad p^2(u-\wt{u})(\wh{u}-\wh{\wt{u}})-q^2(u-\wh{u})(\wt{u}-\wh{\wt{u}})+\Delta\left(p^{-2}-q^{-2}\right)=0   
%\label{eq:Qeqsa} \\ 
&&{\rm Q1}: \quad \po(u-\wh{u})(\wt{u}-\wh{\wt{u}})-q^2(u-\wt{u})(\wh{u}-\wh{\wt{u}})=\Do\po\qo\left(\po-\qo\right)  
\label{eq:Qeqsa} \\ 
&&{\rm Q2}: \quad  \po(u-\wh{u})(\wt{u}-\wh{\wt{u}})-\qo(u-\wt{u})(\wh{u}-\wh{\wt{u}})+
\po\qo(\po-\qo)(u+\wt{u}+\wh{u}+\wh{\wt{u}})=\nn \\ 
&& \hspace{3cm}  =\po\qo(\po-\qo)(\po^2-\po\qo+\qo^2) \label{eq:Qeqsb}\\ 
%&&{\rm Q2}: \quad  p^2(u-\wt{u})(\wh{u}-\wh{\wt{u}})-q^2(u-\wh{u})(\wt{u}-\wh{\wt{u}})+
%\po\qo(\po-\qo)(u+\wt{u}+\wh{u}+\wh{\wt{u}})=\nn \\ 
%&& \hspace{3cm}  =\po\qo(\po-\qo)(\po^2-\po\qo+\qo^2)  \label{eq:Qeqsb}\\ 
&&{\rm Q3}: \quad \po(1-\qo^2)(u\wh{u}+\wt{u}\wh{\wt{u}})-\qo(1-\po^2)(u\wt{u}+\wh{u}\wh{\wt{u}})= \nn \\ 
&& \hspace{3cm} =(\po^2-\qo^2)\left((\wh{u}\wt{u}+u\wh{\wt{u}})+\Do\frac{(1-\po^2)(1-\qo^2)}{4\po\qo}\right) 
\label{eq:Qeqsc} \\
%&&{\rm Q3}: \quad P(u\wh{u}+\wt{u}\wh{\wt{u}})-Q(u\wt{u}+\wh{u}\wh{\wt{u}})= 
%(p^2-q^2)\left[\left(\wh{u}\wt{u}+u\wh{\wt{u}}\right)+\frac{\Delta}{PQ}\right] 
%\label{eq:Qeqsc} \\
&&{\rm Q4}: \quad \po(u\wt{u}+\wh{u}\wh{\wt{u}})-\qo(u\wh{u}+\wt{u}\wh{\wt{u}})= \nn \\ 
&& \hspace{3cm} =\frac{\po\Qo-\qo\Po}{1-\po^2\qo^2}\left((\wh{u}\wt{u}+u\wh{\wt{u}})-
\po\qo(1+u\wt{u}\wh{u}\wh{\wt{u}})\right) \label{eq:Qeqsd}
%&&{\rm Q4}: \quad P(u\wh{u}+\wt{u}\wh{\wt{u}})-Q(u\wt{u}+\wh{u}\wh{\wt{u}})+ \nn \\ 
%&& \hspace{3cm} +\frac{P^2-Q^2}{p+q}\left[\left(\wh{u}\wt{u}+u\wh{\wt{u}}\right)-
%\frac{1}{PQ}\left(g+u\wt{u}\wh{u}\wh{\wt{u}}\right)\right]=0\  .  \label{eq:Qeqsd}
\end{eqnarray}\ese 
with $\Do$ being a fixed parameter, and 
where the form of Q4 as given in \eqref{eq:Qeqsd} was actually due to \cite{Hiet}, and involves additional parameters $\Po$, $\Qo$ related to 
$\po$, $\qo$ respectively via the Jacobi elliptic curve relations: ~$\Po^2=\po^4-\gm \po^2+1$~,~$\Qo^2=\qo^4-\gm \qo^2+1$~,~ 
in which $\gm$ denotes the modulus of this curve.

In what follows we will focus mostly on Q3 and its degenerations, and for the purpose of our treatment we will rewrite \eqref{eq:Qeqsc} 
in a different parametrisation which is more adapted to the actual structure of the solutions, namely as follows:
\begin{equation}
{\rm Q3}:\quad P(u\wh{u}+\wt{u}\wh{\wt{u}})-Q(u\wt{u}+\wh{u}\wh{\wt{u}})= 
(p^2-q^2)\left[\left(\wh{u}\wt{u}+u\wh{\wt{u}}\right)+\frac{\Delta}{PQ}\right] 
\label{eq:Q3new} 
\end{equation}
in which, following \cite{AHN2,NAH}, the lattice parameters $p$, $q$ are related to the parameters $\po$, $\qo$ by the relations 
\be\label{eq:parmtransf}
\po=\frac{P}{p^2-a^2}\quad,\quad \qo=\frac{Q}{q^2-a^2}\  ,  
\ee  
and where $(p,P)$ and $(q,Q)$ are points on a Jacobi elliptic curve (different from the one mentioned above associated with the Q4 equation) 
with branch points $\pm a$, $\pm b$, i.e. ~$P^2=(p^2-a^2)(p^2-b^2)$~, and ~$Q^2=(q^2-a^2)(q^2-b^2)$~. 
The fixed parameter $\Delta$ will have a special significance in terms of certain arbitrary ceofficients in the solutions, as we shall see. 
Even though Q3, as given in \eqref{eq:Qeqsc}, has a rational parametrisation, it turns out that for the solution structure, both here 
(where we are dealing with elliptic solutions) as in the case of rational solitons as in \cite{NAH}, 
this new Jacobi elliptic curve plays a crucial role. Thus, throughout this paper we will consider the parameters $p$ and $q$, and the associated 
parameters $P$, $Q$ on the elliptic curve, to be the natural parameters for the solutions, and consequently we prefer to express the solutions of 
the degenerate cases, such as Q2 and Q1 in terms of these parameters. Other degenerate cases involve the so-called H-equations from \cite{ABS}, and 
they are given by 

\vspace{.2cm}

\noindent
\framebox[2cm][l]{\bf H-list:} 
\bse\label{eq:Heqs}\begin{eqnarray}
&&{\rm H1}: \quad (w-\wh{\wt{w}})(\wh{w}-\wt{w})=p^2-q^2  \label{eq:Heqsa}\\ 
&&{\rm H2}: \quad (w-\wh{\wt{w}})(\wh{w}-\wt{w})=(p^2-q^2)(w+\wt{w}+\wh{w}+\wh{\wt{w}})-p^4+q^4  
\label{eq:Heqsb}\\ 
% &&{\rm H3}: \quad \po(w\wt{w}+\wh{w}\wh{\wt{w}})-\qo(w\wh{w}+\wt{w}\wh{\wt{w}})=\dd(\qo^2-\po^2) 
&&{\rm H3}: \quad P(w\wh{w}+\wt{w}\wh{\wt{w}})-Q(w\wt{w}+\wh{w}\wh{\wt{w}})=2\Delta \frac{P^2-Q^2}{PQ} 
\label{eq:Heqsc}
\end{eqnarray}\ese 
%\vspace{.2cm}
where we have now expressed them in terms of the new parameters $p$ and $q$, whereas in \eqref{eq:Heqsc} we have introduced associated 
parameters $P$, $Q$ defined by the relations ~$P^2+p^2=a^2$~, ~$Q^2+q^2=a^2$~. 
%\noindent
%\framebox[2cm][l]{\bf A-list:}  
%\bse\label{eq:Aeqs}\begin{eqnarray}
%&&{\rm A1}: \quad \po(u+\wh{u})(\wt{u}+\wh{\wt{u}})-\qo(u+\wt{u})(\wh{u}+\wh{\wt{u}})=\dd^2\po\qo(\po^2-\qo^2) \label{eq:Aeqsa}\\ 
%&&{\rm A2}: \quad \po(1-\qo^2)(u\wh{u}+\wt{u}\wh{\wt{u}})-\qo(1-\po^2)(u\wt{u}+\wh{u}\wh{\wt{u}}) +
%(\po^2-\qo^2)\left(1+u\wt{u}\wh{u}\wh{\wt{u}}\right)=0 \label{eq:Aeqsb}
%\end{eqnarray}\ese 
%\vspace{.2cm}
We have omitted from the list the equations denoted by A1 and A2 in \cite{ABS}, which are equivalent to Q1 and Q3 respectively upon gauge transformations.
In a more recent paper, \cite{ABS2}, ABS achieved a somewhat stronger classification result, still in the scalar and affine-linear case, but otherwise 
under less stringent assumptions. 
The various equations are connected to each other through degenerations of the parameter curve, from Q4 down to all other equations. Concentrating 
on the Q-list alone, the corresponding coalescence diagram is given in Figure \ref{Qdegen}.
\begin{figure}[h]
\xymatrix{
&&& \framebox{Q3$_{\delta=1}$} \ar[rr] \ar[dd] && \framebox{Q3$_{\delta=0}$} \ar[dd]\\
&\framebox{\ Q4\ } \ar[urr] \ar[drr]\\
&&& \framebox{\ Q2\ } \ar[rr] && \framebox{Q1$_{\delta=1}$} \ar[rr]&& \framebox{Q1$_{\delta=0}$}\\
}
\caption{Coalescence diagram for equations in the ABS Q list (this was first given by Adler and Suris in \cite{Q4}).}
\label{Qdegen}
\end{figure}
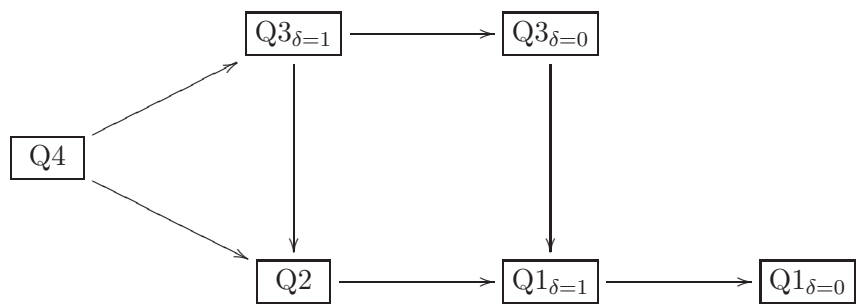
The question as to what is known about explicit solutions of these equations is an interesting one. In fact, most results so far are obtained for the equations 
of KdV type, and go already back several years.  Thus, identifying this subclass of equations, we note that H1 is the so-called \textit{lattice potential KdV equation}, 
which also appeared under the guise of the permutability condition of the B\"acklund transformations of the (continuous) KdV equation, \cite{WE}, or as the so-called $\ven$-algorithm in 
numerical analysis, due to Wynn, cf. \cite{Wynn}. The equation $({\rm Q1})_0$, i.e. the Q1 equation with fixed parameter $\dd=0$, is the lattice Schwarzian KdV 
equation, or \textit{cross-ratio} equation, first presented as an integrable lattice equation in \cite{NC}. The equation $({\rm H3})_0$ can, up to a point transformation, 
be identified with the lattice potential modified KdV (MKdV) equation, also referred to as the Hirota equation, and $({\rm Q3})_0$ can be obtained from an interpolating equation 
that was first presented in \cite{NQC}, and which hereafter we refer to as the NQC equation following \cite{RasinNQC}. For all these equations multi-soliton solutions 
are implicit in the \textit{direct linearisation approach} developed for lattice equations in \cite{NQC,QNCL}. In fact, inverse scattering type solutions including radiation are 
in principle (i.e. on a formal level) included in that scheme as well. Finite-gap solutions of those equations are the ones that are connected to the 
periodic reductions of those lattice equations, i.e. the reduction to finite-dimensional integrable (in the sense of Liouville) dynamical mappings, 
cf. \cite{PNC}. Those were developed 
for the lattice KdV in \cite{EN}, whereas finite-gap solutions for the Hirota equation and the lattice Schwarzian KdV were constructed from a geometric perspective 
in \cite{BP,MP}. Finally, scaling symmetric solutions, leading to reductions to discrete Painlev\'e equations were constructed in \cite{NijPap,NRGO}. To our knowledge this 
exhausts the list of more or less explicit solutions to the lattice equations of KdV type. For any of the new equations in the ABS list, such as the deformed equations 
depending on the parameter $\dd$, as well as H2, Q2 and Q4, no explicit solutions were presented until more recently. The case of Q4 was considered in \cite{AHN}, whilst with the exception of Q4, all of the other new ABS equations were shown to be B\"acklund related to the older equations \cite{Atkinson}, which percipitated the generalisation of the Cauchy-matrix approach to these new systems and hence the results of \cite{AHN2,NAH,HZ}.
The aim of the present paper is to add a new class of solutions, namely elliptic soliton type solutions, of the equations in the ABS list to those found in recent years. 

The outline of the remainder of the paper is as follows. In section 3 we show that additiona formulae for elliptic functions (in the Weierstrass class) have a natural 
interpretation in terms of partial difference equations. Thus, elliptic functions are shown to arise naturally as seed solutions of a number of such equations. In section 
4 we will set out the ground-work for our construction of elliptic $N$-soliton solutions, defining the basic objects and relations in terms of elliptic Cauchy matrices. 
These results will subsequently be used in section 5 to arrive at the main statement (Theorem 5.1) presenting the elliptic $N$-soliton solution for Q3. Furthermore,we 
present a realisation of those solutions in terms of a novel Lax pair for Q3, identifying explicitly the eigenfunctions. In section 6, we demonstrate that the closed form 
solution of Q3 of the previous section arises also from a B\"acklund chain. Finally, in section 7 we show how the solutions of Q3 degenerate to analogous solutions of the 
lower members of the ABS list. Some direct verifications of the resulting solutions of the Q- and H-equations are also presented. 

%%%%%%%%%%%%%%%%%%%%%%%%%%%%%%%% 

\section{Elliptic addition formulae as solutions of P$\Delta$Es }

In this section we will show that many of the equations in the ABS list admit elementary solutions in 
terms of elliptic functions. In fact, one can consider the lattice equations as models for basic addition formulae 
of elliptic functions. We will then, in subsequent sections, generalise these elementary solutions to more general elliptic solutions, 
which we will call elliptic solitons, and we will show that the elementary solutions play the role of the \textit{seed solutions}
for those elliptic solitons as viewed as arising from a B\"acklund chain.  

\subsection{Elliptic addition formulae}

In this paper we will work throughout exclusively in terms of Weierstrass functions, and we employ primarily  
the basic addition formulae for these functions. In what follows, 
$\sg(x)=\sg(x|2\oa,2\oa')$, $\zeta(x)=\zeta(x|2\oa,2\oa')$, $\zeta(x)=\sg(x|2\oa,2\oa')$, denote the 
standard Weierstrass $\sg$-, $\zeta$- and $\wp$-functions with simple periods $2\oa,2\oa'\in\mathbb{C}$, 
see e.g. \cite{Akh}. Let us now recall the following elliptic identities:

\begin{description} 
\item[$\sg$-function] The celebrated three-term identity for Weierstrass $\sg$-function
\begin{eqnarray}\label{eq:3term}
&& \sg(x+a)\sg(x-a)\sg(y+b)\sg(y-b)-\sg(x+b)\sg(x-b)\sg(y+a)\sg(y-a) \nn \\ 
&& =\sg(x+y)\sg(x-y)\sg(a+b)\sg(a-b) 
\end{eqnarray}
can be rewritten in terms of 
\be\label{eq:Phi} 
\Phi_\kp(x)\Equiv  \frac{\sg(x+\kp)}{\sg(\kp)\,\sg(x)} ~\ee 
as:
\be\label{eq:PhiPhi} 
\Phi_\kp(x)\Phi_\ld(y)=\Phi_{\kp+\ld}(x)\Phi_\ld(y-x)+\Phi_\kp(x-y)\Phi_{\kp+\ld}(y)\  . \ee  
\item[$\zeta$-function] The Weierstrass $\zeta$-function ~$\zeta(x)=\sg'(x)/\sg(x)$~~ obeys 
\be\label{eq:zetasg} 
\zeta(x)+\zeta(y)+\zeta(z)-\zeta(x+y+z)=\frac{\sg(x+y)\,\sg(y+z)\,\sg(x+z)}{\sg(x)\,\sg(y)\,\sg(z)\,\sg(x+y+z)} \ee  
or:
\be\label{eq:Phizeta}
 \Phi_\kp(x)\Phi_\kp(y)=\Phi_\kp(x+y)\left[\zeta(\kp)+\zeta(x)+\zeta(y)-\zeta(\kp+x+y) \right] \ee 

\item[$\wp$-function] The Weierstrass $\wp$-function ~$\wp(x)=-\zeta'(x)$~~ obeys: 
\be\label{eq:wpsg} 
\wp(x)-\wp(y)=\frac{\sg(x+y)\,\sg(y-x)}{\sg^2(x)\,\sg^2(y)}\ee 
or:
\be\label{eq:Phiwp} 
\Phi_\kp(x)\Phi_{-\kp}(x)=\wp(x)-\wp(\kp) \ee  
\end{description} 
Furthermore, we recall the following basic addition rules 
\be\label{eq:zetawp}
\zeta(x+y)-\zeta(x)-\zeta(y)=\frac{1}{2} \frac{\wp'(x)-\wp'(y)}{\wp(x)-\wp(y)}\  , 
\ee 
and 
\be\label{eq:zetawpwp}
\wp(x)+\wp(y)+\wp(x+y)= \left(\zeta(x+y)-\zeta(x)-\zeta(y)\right)^2 = \frac{1}{4}\left( \frac{\wp'(x)-\wp'(y)}{\wp(x)-\wp(y)}\right)^2\  . 
\ee
All identities needed in what follows rely on these the above basic addition formulae for the Weierstrass family. 

\subsection{From addition formulae to P$\Delta$Es}

We will now show that the identites given above lead directly to an interpretation in terms of basic elliptic solutions of 
certain lattice equations. Setting
\be\label{eq:tausg}
\tau=\tau_{n,m,h}=\sg(\xi)\quad {\rm with}\quad \xi=\xi_0+n\dd+m\ven+h\gm\  , \ee
and using the abbreviations 
$$ \wt{\tau}=\tau_{n+1,m,h}\quad,\quad \wh{\tau}=\tau_{n,m+1,h}\quad,\quad \ol{\tau}=\tau_{n,m,h+1}\  ,  $$  
we can write the three-term relation as:
\be\label{eq:HM} 
\sg(\dd-\ven)\sg(\gm)\ol{\tau}\,\wh{\wt{\tau}}+\sg(\ven-\gm)\sg(\dd)\wt{\tau}\,\wh{\ol{\tau}}+ \sg(\gm-\dd)\sg(\ven)\wh{\tau}\,\wt{\ol{\tau}}=0\  ,  \ee  
which is the famous \textit{discrete analogue of generalised Toda equation}, or Hirota-Miwa equation, cf. \cite{HM}. It is probably significant that the fundamental 
addition formula for the Weierstrass $\sigma$-function obeys naturally one of the fundamental partial difference equations (P$\Delta$Es) in the 
theory of integrable systems. We will demonstrate now that other P$\Delta$Es, closely related to cases of the ABS list, emerge naturally from 
the addition formulae as well.   

Thus, setting  ~$w=w_{n,m}=\zeta(\xi)-n\zeta(\dd)-m\zeta(\ven)$~  
we get from \eqref{eq:zetasg} 
\be\label{eq:H1} 
(\wh{w}-\wt{w})(w-\wh{\wt{w}})=\wp(\dd)-\wp(\ven)\  , \ee 
which is the lattice potential KdV equation, cf. \cite{Hir,NQC}, or the H1 equation in the ABS list, by the 
identification ~$\wp(\dd)-e\Equiv p^2$~~,~$\wp(\ven)-e\Equiv q^2$~~ 
where $e=\wp(\oa)$ being a branch point of the Weierstrass curve: 
$$ \Gamma=\{(x,y)\,|\,y^2=4(x-e)(x-e')(x-e'')\,\}  $$ 

Setting ~$v(\aR)=v_{n,m}(\aR)= \Phi_\aR(\xi)$~ and identifying parameters
 ~$\po_\aR=\Phi_\aR(\dd)$~~,~$\qo_\aR=\Phi_\aR(\ven)$~ we also have 
\bse\begin{eqnarray}
&&\frac{\po_\aR\wh{v}(\aR)-\qo_\aR\wt{v}(\aR)}{\wh{\wt{v}}(\aR)}=\wh{w}-\wt{w}\\ 
&& 
\frac{\po_\aR v(\aR)+\qo_{-\aR}\wh{\wt{v}}(\aR)}{\wt{v}(\aR)}=w-\wh{\wt{w}}  
\end{eqnarray}\ese
from which we get that $v(\aR)$ obeys a quadrilateral equation which generalises the lattice potential 
MKdV equation, namely 
\be\label{eq:potMKdV}
\po_\aR v(\aR)\wh{v}(\aR)+\qo_{-\aR}\wh{v}(\aR)\wh{\wt{v}}(\aR)=\qo_\aR v(\aR)\wt{v}(\aR)+
\po_{-\aR}\wt{v}(\aR)\wh{\wt{v}}(\aR) \ee  
Furthermore, setting ~$s(\aR,\bb)=\Phi_{\aR+\bb}(\xi)$~~, we have 
\be\label{eq:vv} 
\wt{v}(\aR)\,v(\bb)=\po_\aR s(\aR,\bb)-\po_{-\bb}\wt{s}(\aR,\bb)\quad,\quad 
\wh{v}(\aR)\,v(\bb)=\qo_\aR s(\aR,\bb)-\qo_{-\bb}\wh{s}(\aR,\bb) \ee 
from which we can deduce:
\be\label{eq:Q30}
(\po_\aR s-\po_{-\bb}\wt{s})(\po_\bb \wh{s}-\po_{-\aR}\wh{\wt{s}})= 
(\qo_\aR s-\qo_{-\bb}\wh{s})(\qo_\bb \wt{s}-\qo_{-\aR}\wh{\wt{s}}) \ee  
which by a point transformation can be shown to be equivalent to $({\rm Q3})_0$, i.e. the Q3 
equation with parameter $\Delta=0$. 

\subsection{Elliptic (non-germinating) seed solution for Q4}
We have observed that the addition formulae for elliptic functions, and combinations thereof, can be naturally 
interpreted as quadrilateral lattice equations of KdV type. Thus, the elliptic functions constitute elementary 
solutions of these lattice equations. Such solutions can be viewed as \textit{seed solutions} for B\"acklund 
chains yielding more complicated solutions which could be considered to be elliptic analogues of the soliton 
solutions. In the remainder of the paper we will construct entire families of such solutions, not only for lattice equations 
of KdV type, but for all equations in the ABS list of \cite{ABS} with the exception of Q4 which is treated 
elsewhere, \cite{AN2}. Here, however we will restrict ourselves by showing that also Q4 admits elementary 
solutions in terms of elliptic functions, and we will do that by presenting the so-called three-leg formula 
for the equation in the form \eqref{eq:Qeqsd}. This particular form which essentially is the Jacobi form of Q4, first derived in 
\cite{Hiet}, and for which the 3-leg form was given in \cite{BS2}, admits 
a natural parametrization in terms of Weierstrass functions as follows. 

Noting the periodicity properties of the Weierstrass functions, i.e., 
\be\label{eq:period}
\sg(\xi+2\oa)=-e^{2\eta(\xi+\oa)}\sg(\xi)\quad,\quad \zeta(\xi+2\oa)=\zeta(\xi)+2\eta\quad,\quad \wp(\xi+2\oa)=\wp(\xi)\  , 
\ee
in which $\eta=\zeta(\oa)$, and $\oa$ one of the halfperiods, we can introduce the function 
\be\label{eq:W}
W(x)=\Phi_{\oa}(x) e^{-\eta x}\  , 
\ee   
which obeys the relations
\be\label{eq:Wrel}
W(\xi)W(\xi')=W(\xi+\xi')\left[\zeta(\xi)+\zeta(\xi')+\eta-\zeta(\xi+\xi'+\oa)\right]\quad,\quad W(-\xi)=-W(\xi)\  .  
\ee 
Furthermore, we have 
\be\label{eq:WWrel}
W(\xi)W(\xi+\oa)=-\frac{e^{\eta\oa}}{\sg^2(\oa)}\quad,\quad W^2(\xi)=\wp(\xi)-e\   , 
\ee 
and hence as a consequence
\be\label{eq:wPrel}
(\wp(\xi)-e)(\wp(\xi+\oa)-e)=(e'-e)(e''-e)=g\   , 
\ee
where $e'=\wp(\oa')$, $e''=\wp(\oa+\oa')=-e-e'$ are the other branch points of the standard Weierstrass curve ~$y^2=4(x-e)(x-e')(x-e'')$~. 
The crux of the matter is the following statement

\begin{prop}\label{th:prop1}
For arbitrary variables $X$, $Y$, $Z$ we have the following identity 
\begin{eqnarray}\label{eq:3leg}
&& \left(X-W(\xi+\dd)\right)\left(Y-W(\xi-\ven)\right) \left(Z-W(\xi-\dd+\ven)\right) \nn \\ 
&& \quad -t\left(X-W(\xi-\dd)\right)\left(Y-W(\xi+\ven)\right) \left(Z-W(\xi+\dd-\ven)\right)= \nn \\ 
&& = s\left[ \phantom{\frac{a}{b}} W(\dd)\left(Y\,W(\xi)+XZ\right)-W(\ven)\left(X\,W(\xi)+YZ\right) \right. \nn \\ 
&& \qquad \left. +\frac{W(\dd)W(\ven)}{W(\dd-\ven)} 
\left(Z\,W(\xi)+XY\right)-\frac{1}{W(\dd-\ven)}\left( g+XYZ\,W(\xi)\right)\right] \  ,
\end{eqnarray} 
in which 
\be\label{eq:st} t=\frac{\sg(\xi-\dd)\,\sg(\xi+\ven)\,\sg(\xi+\dd-\ven)}{\sg(\xi+\dd)\,\sg(\xi-\ven)\,\sg(\xi-\dd+\ven)}\quad,
\quad s= (t-1)\frac{W(\dd-\ven)}{W(\xi)}\   . \ee
\end{prop}
\begin{proof}
Eq. \eqref{eq:3leg} can be established by direct computation through identities that hold for the coefficients of each of the 
monomials $XYZ$, $XY$, $XZ$, $YZ$, $X$, $Y$, $Z$ and 1. The corresponding coefficients boil down to the following 
identities: 
\begin{eqnarray*}
t&=& \frac{W(\xi+\ven-\dd)\,W(\xi)-W(\dd)\,W(\ven)}{W(\xi+\dd-\ven)\,W(\xi)-W(\dd)\,W(\ven)}  \\ 
 &=& \frac{W(\xi-\ven)\,W(\xi)-W(\dd-\ven)\,W(\dd)}{W(\xi+\ven)\,W(\xi)-W(\dd-\ven)\,W(\dd)}  \\ 
 &=& \frac{W(\xi+\ven-\dd)\,W(\xi-\ven)-W(\dd-\ven)\,W(\ven)}{W(\xi+\dd-\ven)\,W(\xi+\ven)-W(\dd-\ven)\,W(\ven)}  \\ 
 &=& \frac{W(\xi+\dd)\,W(\xi-\ven)+W(\dd)\,W(\ven)}{W(\xi-\dd)\,W(\xi+\ven)+W(\dd)\,W(\ven)}  \\ 
 &=& \frac{g+W(\xi+\dd)\,W(\xi-\ven)W(\xi+\ven-\dd)\,W(\xi)}{g+W(\xi-\dd)\,W(\xi+\ven)W(\xi-\ven+\dd)\,W(\xi)}\  , 
\end{eqnarray*}
and which are all equal to the form given in \eqref{eq:st}. All these identities can be proven by using the basic addition formula 
\eqref{eq:Wrel}, the only exception being the last equality which follows from the following version of the three-term relation 
for the $\sg$-function:
$$
\sg(\xi+\dd)\,\sg(\xi-\ven)\,\sg(\xi+\ven-\dd)\,\sg(\xi) 
-\sg(\xi-\dd)\,\sg(\xi+\ven)\,\sg(\xi-\ven+\dd)\,\sg(\xi)= \sg(2\xi)\,\sg(\dd-\ven)\,\sg(\dd)\,\sg(\ven)\   ,  
$$
and a similar formula with $\xi$ replaced by $\xi+\oa$, together with the identification ~$g=e^{2\eta\oa}/\sg^4(\oa)$~ and the periodicity 
property \eqref{eq:period} of the $\sg$-function. 
\end{proof}

Identifying in \eqref{eq:3leg} the parameters 
$$ P=W(\dd)\quad,\quad Q=W(\ven)\quad,\quad p=\zeta(\dd+\oa)-\zeta(\dd)-\eta\quad,\quad q=\zeta(\ven+\oa)-\zeta(\ven)-\eta\  , $$ 
which implies that the points $\ssp=(P,p)$, $\ssq=(Q,q)$ lie on the elliptic curve 
$$ \bar{\Gamma}=\{ (P,p)\in\mathbb{C}^2\,|\, p^2=P^2+3e+gP^{-2}\,\}\  , $$
cf. \cite{NP}, and identifying $X=W(\xi)$, we see that the expression in brackets on the right-hand side of \eqref{eq:3leg} can be written in terms of  the 
following quadrilateral expression
\begin{equation}
\cQ_{\ssp,\ssq}(u,\wt{u},\wh{u},\wh{\wt{u}})\Equiv P(u\wh{u}+\wt{u}\wh{\wt{u}})-Q(u\wt{u}+\wh{u}\wh{\wt{u}})
+\frac{P^2-Q^2}{p+q}\left[\left(\wh{u}\wt{u}+u\wh{\wt{u}}\right)-
\frac{1}{PQ}\left(g+u\wt{u}\wh{u}\wh{\wt{u}}\right)\right]\  .  \label{eq:newQ4}
\end{equation} 
It is not hard to see that the equation ~$\cQ_{\ssp,\ssq}(u,\wt{u},\wh{u},\wh{\wt{u}})=0$~ is, up to some simple scaling transformations, equivalent to the Q4 equation in the form \eqref{eq:Qeqsd}. 
As a direct corollary of Proposition \ref{th:prop1} we have now that ~$u=W(\xi)=W(\xi_0+n\dd+m\ven)$~, ~$X=\wt{u}=W(\wt{\xi})$~, ~$Y=\wh{u}=W(\wh{\xi})$~, 
~$Z=\wh{\wt{u}}=W(\wh{\wt{\xi}})$~, is an elliptic solution of the Q4 equation, 
albeit a trivial one in the sense that as a seed solution of a B\"acklund chain it is \textit{non-germinating}, in the sense of the discussion of \cite{AHN}. 
In the latter paper we have constructed germinating seed solutions for Q4, in a slightly dfferent but equivalent parametrisation, as well 
as 1-soliton solutions. In a recent paper \cite{AN2} we have extended these to $N$-soliton solutions for Q4. 
 
%The form of Q4 given in \eqref{eq:Qeqsd} is close to the Jacobi form of Adler's equation due to Hietarinta in \cite{Hiet}, which was explored in \cite{AHN} 
%to find the first seed and soliton solutions. However, in \eqref{eq:Qeqsd} $(P,p)$, $(Q,q)$ are points on what is essentially a Weierstrass elliptic curve written 
%in the form ~$p^2=P^2+3e+gP^{-2}$~,~$q^2=Q^2+3e+gQ^{-2}$~ with  branch points $e$, $e'$, $e''$ and infinity, and where $g=(e-e')(e-e'')$, which was the form of the 
%curve exploited in \cite{NP} where an elliptic extension of the lattice potential KdV (the H1 equation given below) was constructed and studied. 

\section{Cauchy Matrix Scheme} 

We will develop now a scheme along the lines of the paper \cite{NAH} for elliptic soliton solutions, based 
on elliptic Cauchy matrices. In this section we derive the basic relations, and in the next section we 
will use these relations to find a general elliptic $N$-soliton solution for Q3. 

\subsection{Basic ingredients}

At this point let us introduce the Lam\'e function 
\be\label{eq:Psi}
\Psi_\xi(\kp)\Equiv  \Phi_\xi(\kp)\,e^{-\zeta(\xi)\kp}\   ,  
\ee
which coincides with \eqref{eq:Phi} up to an exponential factor, breaking the symmetry between 
the argument of the function and the suffix\footnote{Although most of the results of this paper can be obtained in terms of the 
$\Phi$-function alone, working with the function $\Psi$ has certain advantages as the latter is the more natural function 
in connection with associated continuum equations. We note, however, that the inclusion of the exponential factor amounts to a specific gauge 
transformation on the quantities defined later on in the constructions, and hence could be removed without affecting the main results.}. 
The basic identities for the $\Psi$ function are the following:
\bse\label{eq:Psiids}\bea 
  \Psi_\xi(\kp)\Psi_\dd (\ld)&=&e^{\eta_\dd\kp} \Psi_{\xi+\dd}(\kp)\Psi_\dd(\ld-\kp)+ 
e^{\eta_\dd\ld} \Psi_\xi(\kp-\ld)\Psi_{\xi+\dd}(\ld)\   , \label{eq:Psiidsa} \\ 
 \Psi_\xi(\kp)\Psi_\dd(\kp) &=& e^{\eta_\dd\kp}\Psi_{\xi+\dd}(\kp)\,\left[\zeta(\xi)+\zeta(\dd)+\zeta(\kp)-\zeta(\xi+\dd+\kp) \right]\   , \label{eq:Psiidsb}\\ 
\Psi_\xi(\kp)\Psi_\xi(\ld) &=& \Psi_{\xi}(\kp+\ld)\,\left[\zeta(\xi)+\zeta(\kp)+\zeta(\ld)-\zeta(\xi+\kp+\ld) \right]\   , 
\label{eq:Psiidsc}
\eea\ese 
in which we have introduced 
\be\label{eq:etadd}
\eta_\dd=\eta_\dd(\xi)=\zeta(\xi+\dd)-\zeta(\xi)-\zeta(\dd)=\frac{1}{2}\,\frac{\wp'(\xi)-\wp'(\dd)}{\wp(\xi)-\wp(\dd)}\  . 
\ee 
Furthermore, we have the symmetry: ~$\Psi_\dd(-\kp)=-\Psi_{-\dd}(\kp)$~. 

The starting point for our construction is the ``bare'' non-autonomous Cauchy matrix~ 
\be\label{eq:Cauchy} 
\boldsymbol{M^0}=\left(M^0_{i,j}\right)_{i,j=1,\dots,N}\quad ,\quad  M^0_{i,j}(\xi)\Equiv  \Psi_\xi(\kp_i+\kp_j)\  , 
\ee 
depending on a variable $\xi$ which is linear in the independent variables $n$, $m$, namely ~$\xi=\xi_0+n\dd+m\ven$~~, with 
$\dd$, $\ven$ being the corresponding \textit{lattice parameters}. We will assume that the set rapidity parameters $\{\kp_i, i=1,\dots,N\}$
is such that $\kp_i+\kp_j\neq 0$ (modulo the period lattice of the Weierstrass functions) .

Furthermore, we redefine henceforth the lattice parameters $\po_\kp$, $\qo_\kp$ of section 3, to include an exponential factor, and thus 
we define 
\be\label{eq:pq} p_\kp=\Psi_\dd(\kp)\quad,\quad q_\kp=\Psi_\ven(\kp)  \ee
Setting $\kp=\pm\kp_i$, $\kp=\pm\kp_j$, we can derive from the basic addition formula \eqref{eq:Psiidsa} the following dynamical 
properties of the elliptic Cauchy matrix:  
\begin{eqnarray*}
p_{\kp_i}M_{i,j}^0 &=& \Psi_\dd(\kp_i)\Psi_\xi(\kp_i+\kp_j) \\ 
&=& e^{\eta_\dd(\kp_i+\kp_j)}\Psi_{\xi+\dd}(\kp_i+\kp_j)\Psi_\dd(-\kp_j) + 
e^{\eta_\dd\kp_i}\Psi_{\xi+\dd}(\kp_i)\Psi_\xi(\kp_j) \\ 
&=& \wt{M}_{i,j}^0p_{-\kp_j}e^{\eta_\dd(\kp_i+\kp_j)}+e^{\eta_\dd\kp_i}\Psi_{\xi+\dd}(\kp_i)\Psi_\xi(\kp_j)\  . 
\end{eqnarray*}
Similarly, we have 
$$ p_{-\kp_i}\wt{M}_{i,j}^0=M_{i,j}^0p_{\kp_j}e^{-\eta_\dd(\kp_i+\kp_j)}-e^{-\eta_\dd\kp_i}\Psi_\xi(\kp_i)\Psi_{\xi+\dd}(\kp_j)\  .  $$  

We introduce now the plane-wave factors (i.e., discrete exponential functions)
\bea\label{eq:rho}
 \rho(\kp)=\rho_{n,m}(\kp)&=&\left(e^{-2\zeta(\dd)\kp}\,\frac{p_{-\kp}}{p_\kp}\right)^{~n~} 
\left(e^{-2\zeta(\ven)\kp}\,\frac{q_{-\kp}}{q_\kp}\right)^{~m~}e^{2\zeta(\xi)\kp}\rho_{0,0}(\kp) \nn \\ 
&=&\left(\frac{\sg(\kp-\dd)}{\sg(\kp+\dd)}\right)^{~ n~}
\left(\frac{\sg(\kp-\ven)}{\sg(\kp+\ven)}\right)^{~ m~}e^{2\zeta(\xi)\kp}\rho_{0,0}(\kp)  \  ,  \eea 
and for the specific values $\kp_i$ of $\kp$ the quantities ~$\rho_i\Equiv \rho(\kp_i)$~ obeying the shift relations 
\be\label{eq:rhoshifts}
\frac{\wt{\rho}_i}{\rho_i} = e^{2\eta_\dd\kp_i}\frac{p_{-\kp_i}}{p_{\kp_i}}\quad, \quad 
\frac{\wh{\rho}_i}{\rho_i} = e^{2\eta_\ven\kp_i}\frac{q_{-\kp_i}}{q_{\kp_i}}\  , 
\ee 
where the superscripts $\wt{\phantom{a}}$, $\wh{\phantom{a}}$ as before denote the elementary lattice shifts with 
regard to the independent variables $n$ and $m$ respectively. Use has been made of \eqref{eq:etadd}, noting that ~$\wt{\xi}=\xi+\dd$~, 
and a similar relation for $\eta_\ven$ using ~$\wh{\xi}=\xi+\ven$~.

Next we introduce the $N$-component vectors 
\be\label{eq:rs} 
\brr=\left(\rho_i\Psi_\xi(\kp_i)\right)_{i=1,\dots,N}\quad,\quad \bs=\left(c_j\Psi_\xi(\kp_j)\right)_{j=1,\dots,N} \ee 
where $c_j$ are $N$ \textit{constants} w.r.t. the discrete variables $n$, $m$, and ~$\rho_i=\rho_{n,m}(\kp_i)$~~, in terms of which 
we define now the ``dressed'' Cauchy matrix: 
\be\label{eq:MCauchy} \bM=\left(\bM_{i,j}\right)_{i,j=1,\dots,N}\quad, \quad \bM_{i,j}=\rho_i\bM_{i,j}^0 c_j\  . \ee  
As a consequence of the relations given earlier, and employing the definitions of the plane-wave factors, 
we can now describe the discrete dynamics as follows:

\begin{lemma}
The dressed Cauchy matrix $\bM$, defined in \eqref{eq:MCauchy}, obeys the following linear relations under elementary shifts 
of the independent variables $n$ 
\bse\label{eq:qM}\begin{eqnarray}
&& e^{\eta_\dd\bK}p_{-\bK}\bM-\wt{\bM}p_{-\bK}e^{\eta_\dd\bK} = \wt{\brr}\,\bs^T \label{eq:pMa}\\ 
&& \bM p_{\bK}e^{-\eta_\dd\bK}-e^{-\eta_\dd\bK}p_{\bK}\wt{\bM} = \brr\,\wt{\bs}^T\  , \label{eq:pMb}
\end{eqnarray}\ese 
and under shifts of the variable $m$ the similar relations:
\bse\label{eq:pM}\begin{eqnarray}
&& e^{\eta_\ven\bK}q_{-\bK}\bM-\wh{\bM}q_{-\bK}e^{\eta_\ven\bK} = \wh{\brr}\,\bs^T \label{eq:qMa} \\ 
&& \bM q_{\bK}e^{-\eta_\ven\bK}-e^{-\eta_\ven\bK}q_{\bK}\wh{\bM} = \brr\,\wh{\bs}^T\  . \label{eq:qMb} 
\end{eqnarray}\ese
%where, as before, $\wt{\bM}$ and $\wh{\bM}$ denote the shifted Cauchy matrices. 
\end{lemma}
Here we have adopted the matrix notation, with ~$\bK={\rm diag}(\kp_1,\dots,\kp_N)$~~ 
denoting the diagonal matrix of the $\kp_i$ parameters. Similarly, the symbols $p_{\pm\bK}$, $q_{\pm\bK}$ denote the 
diagonal matrices with entries $p_{\pm\kp_i}$, $q_{\pm\kp_i}$. 

In what follows we will employ the relations \eqref{eq:qM} and \eqref{eq:pM} to obtain nonlinear discrete equations for 
specific objects that we will subsequently define in terms of the Cauchy matrix $\bM$.

\subsection{The $\tau$-function and related basic objects}  
 
Introduce now the ~$\tau$-function~~: 
\be\label{eq:tau} f=f_{n,m}=\det\left(\boldsymbol{1}+\boldsymbol{M}\right) \  , \ee  
where $\boldsymbol{1}$ is the $N\times N$ unit matrix. 
{}From the relations \eqref{eq:qM} and \eqref{eq:pM} we can derive the discrete dynamics of the $\tau$-function by performing 
the following computation  
\begin{eqnarray*} 
\wt{f} &=& \det\left( \bun+\wt{\bM}\right)=
\det\left\{ \bun+\left[ e^{\eta_\dd\bK}p_{-\bK}\bM -\wt{\brr}\bs^T\right]e^{-\eta_\dd\bK}(p_{-\bK})^{-1} \right\} \\  
&=& \det\left\{ e^{\eta_\dd\bK}p_{-\bK} \left[ \bun + \bM - e^{-\eta_\dd\bK}(p_{-\bK})^{-1}\wt{\brr} \bs^T \right]
e^{-\eta_\dd\bK} (p_{-\bK})^{-1} \right\} \\ 
&=& \det\left\{ (\bun + \bM) \left[ \bun - (\bun+\bM)^{-1} e^{-\eta_\dd\bK}(p_{-\bK})^{-1}\wt{\brr} \bs^T \right]  \right\} \\ 
&=& f\,\det\left\{ \bun - (\bun+ \bM)^{-1} e^{-\eta_\dd\bK}(p_{-\bK})^{-1}\wt{\brr} \bs^T  \right\} \  , 
\end{eqnarray*} 
and similar computations involving the other lattice shift. Noting that 
from \eqref{eq:rhoshifts} it follows that  
\be 
e^{-\eta_\dd\bK}(p_{-\bK})^{-1}\wt{\brr}=e^{\eta_\dd\bK}p_{\bK}^{-1}\frac{\Psi_{\wt{\xi}}(\bK)}{\Psi_\xi(\bK)}\,\brr\  , 
\ee 
we can use \eqref{eq:Psiidsc} and the Weinstein-Aronszajn formula 
$$  \det\left( \bun+\ba\,\bbb^T\right) = 1+\bbb^T\cdot\ba $$  
to write 
\be\label{eq:ff}
\frac{\wt{f}}{f}= 1-\bs^T(\bun+\bM)^{-1}\,\left[\zeta(\xi)+\zeta(\bK)+\zeta(\dd)-\zeta(\xi+\bK+\dd)\right]^{-1}\brr 
= V_\dd~.  \ee 
Here and in what follows the notation $\zeta(\bK)$ denotes the diagonal matrix with entries $\zeta(\kp_i)$, ($i=1,\dots,N$), 
whilst all scalar terms are supposed to be considered as coefficients of the unit matrix. 
In \eqref{eq:ff} we have introduce the quantity: 
\bse\label{eq:UV}
\begin{equation}\label{eq:V}
V_\aR = 1-\bs^T\,(\bun+\bM)^{-1}(\chi_{\aR,~\bK~})^{-1}\brr= 1-\bs^T\,(\chi_{\aR,~\bK~})^{-1}(\bun+\bM)^{-1}\brr \  ,  
\end{equation}
for arbitrary parameter $\aR$. In addition to $V_\aR$ we will also need, for later purpose, the objects  
\begin{equation}
U_{\aR,\bb} = \bs^T\,(\chi_{\bb,~\bK~})^{-1}(\bun+\bM)^{-1}(\chi_{\aR,~\bK~})^{-1}\brr=U_{\bb,\aR} \  , 
\label{eq:U}
\end{equation}\ese  
where both $\aR,\bb~$ are arbitrary (complex) parameters. In \eqref{eq:U} and \eqref{eq:V} we have abbreviated 
\be\label{eq:chi} 
\chi_{\aR,\bb}=\chi_{\aR,\bb}(\xi)\Equiv  \zeta(\aR)+\zeta(\bb)+\zeta(\xi)-\zeta(\xi+\aR+\bb)\  ,  \ee
and where $\chi_{\aR,\bK}$ denotes the diagonal matrix with entries $\chi_{\aR,\kp_i}$.   
%%
%suppressing in the notation the dependence on $\xi$.  

\subsection{Basic linear relations}

In order to derive relations for the objects $V_\aR$ and $U_{\aR,\bb}$ we need to introduce the $N$-component 
column- resp. row vectors: 
\bse\label{eq:butbu}\begin{eqnarray}
\bu_\aR &=& (\bun+\bM)^{-1}(\chi_{\aR,\bK})^{-1}\brr \  , \label{eq:bu} \\ 
\tbu_\bb &=& \bs^T\,(\chi_{\bb,\bK})^{-1}(\bun+\bM)^{-1} \  ,   \label{eq:tbu}
\end{eqnarray}\ese 
Performing the following calculation:
\begin{eqnarray*}
&&\wt{\bu}_\aR= (\bun+\wt{\bM})^{-1}(\wt{\chi}_{\aR,\bK})^{-1}\wt{\brr}=(\bun+\wt{\bM})^{-1}(\wt{\chi}_{\aR,\bK})^{-1}
e^{2\eta_\dd\bK}\frac{\Psi_{\wt{\xi}}(\bK)p_{-\bK}}{\Psi_\xi(\bK)p_{\bK}}\brr \\ 
&& \Rightarrow\quad e^{-\eta_\dd\bK}p_{\bK}(\bun+\wt{\bM})\wt{\bu}_\aR= \frac{\Phi_{\wt{\xi}+\aR}(\bK)}
{\Phi_{\wt{\xi}}(\bK)\,\Phi_{\aR}(\bK)}\,
e^{\eta_\dd\bK}\frac{\Psi_\dd(-\bK)\,\Psi_{\wt{\xi}}(\bK)}{\Psi_\xi(\bK)}\,\brr \\ 
 && \Rightarrow\quad \left[(\bun+\bM)p_{\bK}e^{-\eta_\dd\bK}-\brr\,\wt{\bs}^T\right]\,\wt{\bu}_\aR = 
-\frac{\zeta(\bK)+\zeta(\wt{\xi}+\aR)-\zeta(\dd)-\zeta(\bK+\xi+\aR)}{\zeta(\bK)+\zeta(\xi)+\zeta(\aR)-\zeta(\bK+\xi+\aR)}\,\brr \\ 
&& \qquad\qquad = \left(-1+\frac{\zeta(\aR)+\zeta(\xi)-\zeta(\wt{\xi}+\aR)+\zeta(\dd)}
{\zeta(\bK)+\zeta(\xi)+\zeta(\aR)-\zeta(\bK+\xi+\aR)}\right)\,\brr
\end{eqnarray*}
multiplying both sides by ~$(\bun+\bM)^{-1}$~ and introducing the vector
\be\label{eq:bu0} 
\bu^0\Equiv (\bun+\bM)^{-1}\brr\  , 
\ee 
we get the relation
\bse\be\label{eq:burel} e^{-\eta_\dd\bK}p_{\bK}\wt{\bu}_\aR=(\wt{\bs}^T\,\wt{\bu}_\aR)\bu^0+
\chi_{\aR,\dd}\bu_\aR-\bu^0\quad\Rightarrow\quad 
 e^{-\eta_\dd\bK}{p}_{\bK}\wt{\bu}_\aR=-\wt{V}_\aR\,\bu^0+\chi_{\aR,\dd}\bu_\aR\  . \ee  
In a similar way one can derive the relation:\\ 
\be  e^{\eta_\dd\bK}p_{-\bK}\bu_\aR=V_\aR\,\wt{\bu}^0-\wt{\chi}_{\aR,-\dd}\wt{\bu}_\aR\  . \ee \ese  
A similar set of relations can be derived for the adjoint vectors \eqref{eq:tbu} which involves the 
adjoint vector to \eqref{eq:bu0}, namely 
\be\label{eq:tbu0} 
\tbu^0\Equiv \bs^T\,(\bun+\bM)^{-1}\  , 
\ee 
and obviously these relations all have their counterparts involving the other lattice shift related to shifts in 
the discrete independent variable $m$ instead of $n$.   

Summarising the results of these derivations, we have the following statement:
\begin{lemma}
The $N$-component vectors given in \eqref{eq:butbu}, together with the ones defined in \eqref{eq:bu0} and \eqref{eq:tbu0} obey 
the following set of linear difference equations 
\bse\label{eq:shiftrels}\begin{eqnarray}
 && e^{-\eta_\dd\bK}{p}_{\bK}\wt{\bu}_\aR=-\wt{V}_\aR\,\bu_0+\chi_{\aR,\dd}\bu_\aR\  ,  \\ 
 && e^{\eta_\dd\bK}p_{-\bK}\bu_\aR=V_\aR\,\wt{\bu}_0-\wt{\chi}_{\aR,-\dd}\wt{\bu}_\aR\  , \\ 
 && \tbu_\bb\, p_{\bK} e^{-\eta_\dd\bK} = V_{\bb}\,^{t\!}\wt{\bu}_0 - \wt{\chi}_{\bb,-\dd}\,^{t\!}\wt{\bu}_\bb  , \\ 
 && \,^{t\!}\wt{\bu}_\bb\, p_{-\bK} e^{\eta_\dd\bK} = -\wt{V}_\bb \tbu_0 + \chi_{\bb,\dd}\tbu_\bb\   , 
    \end{eqnarray}\ese 
and a similar set of relations involving the shifts in the variable $m$ obtained by replacing $p_{\pm\bK}$ by $q_{\pm\bK}$, 
$\dd$ by $\ven$ and $\wt{\phantom{a}}$ by $\wh{\phantom{a}}$.  
\end{lemma}

\subsection{Basic nonlinear relations} 

From the basic relations \eqref{eq:shiftrels} we observe that  
\begin{eqnarray*}
&& \bs^T\,(\chi_{\bb,\bK})^{-1}\left(\chi_{\aR,\dd}\bu_\aR-\wt{V}_\aR \bu_0 \right)=
\wt{\bs}^T\,\frac{\Psi_\xi(\bK)}{\Psi_{\wt{\xi}}(\bK)}\,\frac{\Phi_{\xi+\bb}(\bK)}{\Phi_\xi(\bK)\,\Phi_\bb(\bK)}\,e^{-\eta_\dd\bK}p_{\bK}\wt{\bu}_\aR \\ 
&& =\wt{\bs}^T\,\frac{\Phi_{\xi+\bb}(\bK)\Phi_\dd(\bK)}{\Phi_{\wt{\xi}}(\bK)\,\Phi_\bb(\bK)}\,\wt{\bu}_\aR = 
\wt{\bs}^T\,\frac{\zeta(\bK)+\zeta(\dd)+\zeta(\xi+\bb)-\zeta(\wt{\xi}+\bK+\bb)}
{\zeta(\bK)+\zeta(\bb)+\zeta(\wt{\xi})-\zeta(\wt{\xi}+\bK+\bb)}\,\wt{\bu}_\aR \nn \\ 
&& = \wt{\bs}^T\,\left(1+\frac{\zeta(\dd)-\zeta(\bb)+\zeta(\xi+\bb)-\zeta(\xi+\dd)}
{\zeta(\bK)+\zeta(\bb)+\zeta(\wt{\xi})-\zeta(\wt{\xi}+\bK+\bb)}\right)\,\wt{\bu}_\aR=(1-\wt{V}_\aR)-\wt{\chi}_{\bb,-\dd}\wt{U}_{\aR,\bb} \\
&& = \chi_{\aR,\dd}U_{\aR,\bb}-\wt{V}_\aR(1-V_\bb)\  , 
\end{eqnarray*}
from which we get the following relation:
\bse\label{eq:VV} 
\be\label{eq:pVV} \wt{V}_\aR\,V_\bb=1-\wt{\chi}_{\bb,-\dd}\wt{U}_{\aR,\bb}-\chi_{\aR,\dd}U_{\aR,\bb}\   . \ee  
Similarly we have
\be\label{eq:qVV} \wh{V}_\aR\,V_\bb=1-\wh{\chi}_{\bb,-\ven}\wh{U}_{\aR,\bb}-\chi_{\aR,\ven}U_{\aR,\bb}\   . \ee\ese  

These relations can also be rewritten as follows
\be\label{eq:PsiVV} 
\left(\Psi_{\wt{\xi}}(\aR)\,\wt{V}_\aR\right)\left(\Psi_\xi(\bb)\,V_\bb\right)=
p_\aR e^{-\eta_\dd\aR}\left[\Psi_\xi(\aR+\bb)(1-\chi_{\aR,\bb}U_{\aR,\bb})\right]
-p_{-\bb}e^{\eta_\dd\bb}\left[\Psi_{\wt{\xi}}(\aR+\bb)(1-\wt{\chi}_{\aR,\bb}\wt{U}_{\aR,\bb})\right]
\ee   
provided $\aR+\bb\neq 0$ (modulo the period lattice of the Weierstrass functions). 

We now intend to remove the $\xi$-dependence in the coefficients of \eqref{eq:PsiVV} by performing yet another change of variables, namely 
by defining 
\be\label{eq:WS} 
 W_\aR\Equiv \Psi_\xi(\aR)\,V_\aR\quad,\quad S_{\aR,\bb}\Equiv \Psi_\xi(\aR+\bb)(1-\chi_{\aR,\bb}U_{\aR,\bb})=S_{\bb,\aR}\quad,\quad \aR+\bb\neq 0\ 
{\rm (mod.\ \ root\ \ lattice},   
\ee 
%\bred where the latter object is symmetric w.r.t. its labels, i.e. ~$S_{\aR,\bb}=S_{\bb,\aR}$~. 
We now have the following result:

\begin{lemma}
The following difference relations hold between the quantities defined in \eqref{eq:WS}  
\bse\label{eq:WW}\begin{eqnarray}
&& \wt{W}_\aR\,W_\bb= p_\aR e^{-\eta_\dd\aR}S_{\aR,\bb}-p_{-\bb}e^{\eta_\dd\bb}\wt{S}_{\aR,\bb}\  , \\ 
&& \wh{W}_\aR\,W_\bb= q_\aR e^{-\eta_\ven\aR}S_{\aR,\bb}-q_{-\bb}e^{\eta_\ven\bb}\wh{S}_{\aR,\bb}\  , 
\end{eqnarray}\ese 
where ~~ and $\aR+\bb\neq 0$ (mod. period lattice). Furthermore, in the latter case the relations \eqref{eq:WW} are 
replaced by the following ones:  
\bse\label{eq:WWs}\begin{eqnarray}
&& \wt{W}_\aR W_{-\aR}=p_\aR e^{-\eta_\dd\aR}\left\{ [\zeta(\xi)-\zeta(\wt{\xi})-\zeta(\aR)+\zeta(\aR+\dd)] 
+[\wp(\wt{\xi})-\wp(\aR)]\wt{U}_{\aR,-\aR}-[\wp(\xi)-\wp(\aR)]U_{\aR,-\aR}\right\}\  , \nn \\ 
&& \label{eq:pWW} \\ 
&& \wh{W}_\aR W_{-\aR}=q_\aR e^{-\eta_\ven\aR}\left\{ [\zeta(\xi)-\zeta(\wh{\xi})-\zeta(\aR)+\zeta(\aR+\ven)] 
+[\wp(\wh{\xi})-\wp(\aR)]\wh{U}_{\aR,-\aR}-[\wp(\xi)-\wp(\aR)]U_{\aR,-\aR}\right\}\  . \nn \\ 
&& \label{eq:qWW} 
\end{eqnarray}\ese 
\end{lemma} 
As a consequence of the relation \eqref{eq:WWs}, taking the limit $\aR\rightarrow-\dd$ and $\aR\rightarrow-\dd$ respectively, 
%recalling that $p_\aR=\Psi_\dd(\aR)$, 
we obtain the inversion relations 
\be\label{eq:WWinv} 
\wt{W}_{-\dd} W_{\dd}=-\frac{e^{[\zeta(\wt{\xi})-\zeta(\xi)]\dd}}{\sg^2(\dd)}\quad,\quad 
\wh{W}_{-\ven} W_{\ven}=-\frac{e^{[\zeta(\wh{\xi})-\zeta(\xi)]\ven}}{\sg^2(\ven)}\  .  
 \ee 

\subsection{NQC equation} 
We now present the nonlinear P$\Delta$E for the variables $S_{\aR,\bb}$ which follow directly from the relations 
\eqref{eq:WWs}, in fact by simple elimination of the  $W$ variables. 

\begin{prop}\label{eq:prop2} 
The variable $S_{\aR,\bb}$, defined in \eqref{eq:WS}, for fixed parameters $\aR$, $\bb$ obeys the following nonlinear partial difference 
equation 
\bea\label{eq:NQC} 
&& ( p_\aR e^{-\eta_\dd\aR}S_{\aR,\bb}-p_{-\bb}e^{\eta_\dd\bb}\wt{S}_{\aR,\bb})
\,(p_\bb e^{-\wh{\eta}_\dd\bb}\wh{S}_{\aR,\bb}-p_{-\aR}e^{\wh{\eta}_\dd\aR}\wh{\wt{S}}_{\aR,\bb} ) \nn \\  
&& =( q_\aR e^{-\eta_\ven\aR}S_{\aR,\bb}-q_{-\bb}e^{\eta_\ven\bb}\wh{S}_{\aR,\bb})
\,( q_\bb e^{-\wt{\eta}_\ven\bb}\wt{S}_{\aR,\bb}-q_{-\aR}e^{\wt{\eta}_\ven\aR}\wh{\wt{S}}_{\aR,\bb})\  .  \eea  
Furthermore, the following relation holds 
\bea\label{eq:modulNQC} 
&& ( p_\aR e^{-\eta_\dd\aR} S_{\aR,\bb}-p_{-\bb}e^{\eta_\dd\bb}\wt{S}_{\aR,\bb})
\,(p_{\aR'} e^{-\wh{\eta}_\dd\aR'} \wh{S}_{\aR',\bb'}-p_{-\bb'}e^{\wh{\eta}_\dd\bb'}\wh{\wt{S}}_{\aR',\bb'} ) \nn \\  
&& =( q_{\bb'}e^{-\eta_\ven\bb'} S_{\bb,\bb'}-q_{-\bb}e^{\eta_\ven\bb}\wh{S}_{\bb,\bb'})
\,( q_{\aR'}e^{-\wt{\eta}_\ven\aR'} \wt{S}_{\aR,\aR'}-q_{-\aR}e^{\wt{\eta}_\ven\aR}\wh{\wt{S}}_{\aR,\aR'}) \eea  
between these solutions of the partial difference equation \eqref{eq:NQC} with different fixed parameters $\aR,\bb,\aR',\bb'$. 
\end{prop} 
Eq. \eqref{eq:NQC} first appeared in a slightly different form in \cite{NQC}. It should be noted that 
the latter parameters $\aR,\bb,\aR',\bb'$ are to be distinguished from the lattice parameters $\dd$ and $\epsilon$ which are 
associated with the lattice shifts. 

Eq. \eqref{eq:NQC} can be cast in a more universal form by setting ~$u_{\aR,\bb}=\:\rho^{1/2}(\aR)\rho^{1/2}(\bb)S_{\aR,\bb}$~, 
which leads to an equation of the form: 
\be\label{eq:Q30parm}  
P(u_{\aR,\bb}\wh{u}_{\aR,\bb}+\wt{u}_{\aR,\bb}\wh{\wt{u}}_{\aR,\bb})-
Q(u_{\aR,\bb}\wt{u}_{\aR,\bb}+\wh{u}_{\aR,\bb}\wh{\wt{u}}_{\aR,\bb})=
(p^2-q^2)(\wh{u}_{\aR,\bb}\wt{u}_{\aR,\bb}+u_{\aR,\bb}\wh{\wt{u}}_{\aR,\bb})\  , \ee 
i.e. the Q3 equation for $\Delta=0$. Here we have used the identities 
~$p_\aR p_{-\aR}-q_\aR q_{-\aR}=p_\bb p_{-\bb}-q_\bb q_{-\bb}=p^2-q^2=\wp(\dd)-\wp(\ven)$~ and we have introduced 
the parameters $P$, $Q$ obeying 
\bse\label{eq:ellcurves} \bea\label{eq:Ppcurve} 
&&P^2=p_\aR p_\bb p_{-\aR}p_{-\bb}=(p^2-a^2)(p^2-b^2) \  , \\   
&&Q^2=q_\aR q_\bb q_{-\aR}q_{-\bb}=(q^2-a^2)(q^2-b^2)\  ,  \eea\ese 
i.e. $\ssp=(p,P)$, $\ssq=(q,Q)$ are points on a (Jacobi) elliptic curve, with moduli ~$\pm a$, $\pm b$~ which 
are given by ~$a^2=\wp(\aR)-e$~, ~$b^2=\wp(\bb)-e$~. From \eqref{eq:modulNQC} it follows also that 
any four solutions $u_{\aR,\bb}$, $u_{\aR',\bb'}$, $u_{\aR,\aR'}$, $u_{\bb,\bb'}$ are connected via the 
relation 
\bea\label{eq:modulQ30}
&& P_{\aR,\aR'} \wh{u}_{\aR,\bb}u_{\aR',\bb'}-P_{\aR',\bb}\wh{\wt{u}}_{\aR,\bb}u_{\aR',\bb'}
-P_{\aR,\bb'}\wh{u}_{\aR,\bb}\wt{u}_{\aR',\bb'}+P_{\bb,\bb'}\wh{\wt{u}}_{\aR,\bb}\wt{u}_{\aR',\bb'} = \nn \\
&& = Q_{\aR,\bb}\wt{u}_{\aR,\aR'}u_{\bb,\bb'}-Q_{\aR',\bb}\wh{\wt{u}}_{\aR,\aR'}u_{\bb,\bb'}
-Q_{\aR,\bb'} \wt{u}_{\aR,\aR'}\wh{u}_{\bb,\bb'}+Q_{\aR',\bb'}\wh{\wt{u}}_{\aR,\aR'}\wh{u}_{\bb,\bb'}\   , 
\eea
in which the parameters $P_{\aR_1,\aR_2}$, etc. are lattice parameters associated with different elliptic curves, i.e. 
they obey relations of the type 
$$ P_{\aR_1,\aR_2}^2=(p^2-a_1^2)(p^2-a_2^2)\quad,\quad Q_{\aR_1,\aR_2}^2=(q^2-a_1^2)(q^2-a_2^2)  $$ 
with branch points $\pm a_1$, $\pm a_2$ which are associated with uniformising parameters $\aR_1$, $\aR_2$ (taken in 
the set ~$\{ \aR,\aR',\bb,\bb'\}$~). The various lattice parameters are not independent, but are related through 
the quadric relations 
$$ P_{\aR,\aR'}^2-P_{\aR',\bb}^2-P_{\aR,\bb'}^2+P_{\bb,\bb'}^2=(a^2-b^2)(a'^2-b'^2)\  , $$ 
and a similar relation for the parameters $Q$.

\subsection{Miura relations}

In addition to the relations between the objects $W_\aR$ and $S_{\aR,\bb}$
we need another set of relations involving the object 
\be\label{eq:wnew}  w\Equiv  \zeta(\xi)-n\zeta(\dd)-m\zeta(\ven)-\bs^T\bu_0~~,  \ee
which, as we shall see, solves the H1 equation.    

\begin{lemma}
%{\bf Lemma:} {\it 
Between the objects $W_\aR$ and $w$ the following relations hold:
\bse\label{eq:wWMiura}\begin{eqnarray}
\wh{w}-\wt{w} &=& \frac{p_\aR e^{-\wh{\eta}_\dd\aR}\wh{W}_\aR-q_\aR e^{-\wt{\eta}_\ven\aR}\wt{W}_\aR}{\wh{\wt{W}}_\aR}=
\frac{p_{-\aR} e^{\eta_\dd\aR}\wt{W}_\aR-q_{-\aR}e^{\eta_\ven\aR}\wh{W}_\aR}{W_\aR} \label{eq:wWMiuraa} \\ 
w-\wh{\wt{w}} &=& \frac{p_\aR e^{-\eta_\dd\aR} W_\aR+q_{-\aR} e^{\wt{\eta}_\ven\aR}\wh{\wt{W}}_\aR}{\wt{W}_\aR}=
\frac{p_{-\aR} e^{\wh{\eta}_\dd\aR}\wh{\wt{W}}_\aR+q_{\aR} e^{-\eta_\ven\aR} W_\aR}{\wh{W}_\aR}\  , \label{eq:wWMiurab} 
\end{eqnarray} \ese 
and, furthermore, $w$ obeys the equation H1 in the following form:
\be\label{eq:H1new} (\wh{w}-\wt{w})(w-\wh{\wt{w}})=p^2-q^2 \quad,\quad  p^2=\wp(\dd)-e, \quad q^2=\wp(\ven)-e~\  .  \ee 
\end{lemma} 
\begin{proof}
%{\bf Proof:} 
From the relation 
\begin{eqnarray*}
 e^{-\eta_\dd\bK}p_{\bK}\wt{\bu}_\aR=\chi_{\aR,\dd}\bu_\aR-\wt{V}_\aR\bu_0 &\Rightarrow&  
\bs^T\,e^{-\eta_\dd\bK}p_{\bK}\wt{\bu}_\aR=\wt{\bs}^T\,\frac{\Psi_\xi(\bK)}{\Psi_{\wt{\xi}}(\bK)}\Psi_\dd(\bK)e^{-\eta_\dd\bK}\wt{\bu}_\aR \\ 
&& = \chi_{\aR,\dd}\bs^T\bu_{\aR}-\wt{V}_\aR \bs^T\bu_0 
\end{eqnarray*}
we get 
$$ \wt{\bs}^T\left[ \zeta(\bK)+\zeta(\xi)+\zeta(\dd)-\zeta(\bK+\wt{\xi})\right]\wt{\bu}_\aR = (1-V_\aR) \chi_{\aR,\dd} 
-w_0\wt{V}_\aR \  , $$ 
where ~$w_0=\bs^T\bu_0$~~. Similarly, from
~$p_{-\bK}\bu_\aR=V_\aR\wt{\bu}_0-\wt{\chi}_{\aR,-\dd}\wt{\bu}_\aR$~~ we can derive the relation:
$$ -\bs^T\left[ \zeta(\bK)+\zeta(\wt{\xi})-\zeta(\dd)-\zeta(\bK+\xi)\right]\bu_\aR = V_\aR\wt{w}_0 -(1-\wt{V}_\aR)
\wt{\chi}_{\aR,-\dd}\  . $$  
Combining these with the analogous relations involving the other lattice shift, i.e. 
\begin{eqnarray*}
&& \wh{\bs}^T\left[ \zeta(\bK)+\zeta(\xi)+\zeta(\ven)-\zeta(\bK+\wh{\xi})\right]\wh{\bu}_\aR = (1-V_\aR) \chi_{\aR,\ven} -w_0\wh{V}_\aR \  , \\
&& -\bs^T\left[ \zeta(\bK)+\zeta(\wh{\xi})-\zeta(\ven)-\zeta(\bK+\xi)\right]\bu_\aR = V_\aR\wh{w}_0 -(1-\wh{V}_\aR)
\wh{\chi}_{\aR,-\dd}\  , 
\end{eqnarray*} 
to eliminate the terms involving $\bK$ on the left-hand sides, we obtain the following set of relations:
\bse\label{eq:wVMiura}\begin{eqnarray}
&& \wh{w}-\wt{w}=\frac{\wh{\chi}_{\aR,\dd}\wh{V}_\aR-\wt{\chi}_{\aR,\ven}\wt{V}_\aR}{\wh{\wt{V}}_\aR} 
=\frac{\wh{\chi}_{\aR,-\ven}\wh{V}_\aR-\wt{\chi}_{\aR,-\dd}\wt{V}_\aR}{V_\aR} \label{eq:wVMiuraa} \\ 
&& w-\wh{\wt{w}}=\frac{\chi_{\aR,\dd}V_\aR-\wh{\wt{\chi}}_{\aR,-\ven}\wh{\wt{V}}_\aR}{\wt{V}_\aR} 
=\frac{\chi_{\aR,\ven}V_\aR-\wh{\wt{\chi}}_{\aR,-\dd}\wh{\wt{V}}_\aR}{\wh{V}_\aR}\  . \label{eq:wVMiurab} 
\end{eqnarray}\ese 
%where ~$p_0=\zeta(\dd),q_0=\zeta(\ven)$~~. 
The relations \eqref{eq:wWMiura} are obtained from \eqref{eq:wVMiura} by substitution for $V_{\aR}$ using \eqref{eq:WS}, 
and by using the expressions for $\chi_{\aR,\pm\dd}$ 
and $\chi_{\aR,\pm\ven}$ in terms of the $\Psi$-function in the form 
$$ \chi_{\aR,\dd}= e^{-\eta_\dd\aR}\frac{p_\aR \Psi_\xi(\aR)}{\Psi_{\wt{\xi}}(\aR)}\quad,\quad  
\wt{\chi}_{\aR,-\dd}=-e^{\eta_\dd\aR} \frac{p_{-\aR} \Psi_{\wt{\xi}}(\aR)}{\Psi_\xi(\aR)}\  , $$ 
and similar expressions with $\dd$ replaced by $\ven$. The latter relations follow directly from the addition formulae
\eqref{eq:Psiids}.  
Finally, setting ~$\aR=\dd\Rightarrow p_{-\aR}=p_{-\dd}=0$~~, the second forms of the expressions in the lemma 
yield:
$$ \wh{w}-\wt{w}=-q_{-\dd} e^{\eta_\ven\dd}\frac{\wh{W}_\dd}{W_\dd}\quad {\rm and}\quad w-\wh{\wt{w}}=q_\dd e^{-\eta_\ven\dd}\frac{W_\dd}{\wh{W}_\dd} $$ 
which by multiplication, using ~$q_\dd q_{-\dd}=\wp(\ven)-\wp(\dd)$~~ gives us the H1 equation for $w$. %$\square$
\end{proof}    

%%

%%%%%%%%%%%%%%%%%%%%%%    Q3   %%%%%%%%%%%%%%%%%%%%%%%%%%%%%%%%%%%%%%%%%%%%%%%

\section{Elliptic $N$-soliton solution of Q3}

\subsection{$N$-soliton formula for Q3} 
We are now ready to formulate the main result of this paper, which comprises a 
generalization of the rational $N$-soliton solution for Q3 as obtained in \cite{AHN,NAH} to the 
elliptic case. We use the following expression to abbreviate the quadrilateral form of the equation: 
\be\label{eq:Q3quad} 
\cQ^\Delta_{\ssp,\ssq}(u,\wt{u},\wh{u},\wh{\wt{u}})\Equiv  P(u\wh{u}+\wt{u}\wh{\wt{u}})-Q(u\wt{u}+\wh{u}\wh{\wt{u}})-(p^2-q^2)\left((\wh{u}\wt{u}+u\wh{\wt{u}})
+~\frac{\Delta}{PQ}~\,\right)
\ee  
in which the lattice parameters $\ssp$, $\ssq$ are points $(p,P)$, $(q,Q)$ on the elliptic curve:
\be\label{eq:ellcurve} 
\Gamma_{a,b}=\{ (x,X)\,|\,X^2=(x^2-a^2)(x^2-b^2)\,\} \  , 
\ee 
with moduli $\pm a$, $\pm b$, which we refer to as the \textit{parameter curve}. 

\begin{theorem}\label{th:theorem1}
The following formula 
\begin{eqnarray} \label{eq:Nsol} 
u^{(N)}&=& A\rho^{1/2}(\aR)\rho^{1/2}(\bb)S_{\aR,\bb}+B\rho^{1/2}(\aR)\rho^{1/2}(-\bb)S_{\aR,-\bb} \nn \\ 
&& +C\rho^{1/2}(-\aR)\rho^{1/2}(\bb)S_{-\aR,\bb}+D\rho^{1/2}(-\aR)\rho^{1/2}(-\bb)S_{-\aR,-\bb}\  , 
\end{eqnarray}  
for each fixed positive integer $N$, and parameters and variables given in the previous section, provide solutions of the 
quadrilateral equation 
\be\label{eq:Q3new} 
 \cQ^\Delta_{\ssp,\ssq}(u^{(N)},\wt{u}^{(N)},\wh{u}^{(N)},\wh{\wt{u}}^{(N)})=0 
\ee 
with arbitrary constant coefficients $A$,$B$,$C$,$D$, where $\Delta$ is given by 
\be\label{eq:Delta} 
\Delta=\wp'(\aR)\,\wp'(\bb)\det\left(\begin{array}{cc} A & B\\ C & D\end{array}\right)\  . \ee  
Furthermore, the associated function to $u^{(N)}$:
\begin{eqnarray} \label{eq:assNsol}  
U^{(N)}&=& A\rho^{1/2}(\aR)\rho^{1/2}(\bb)W_\aR W_\bb+B\rho^{1/2}(\aR)\rho^{1/2}(-\bb)W_{\aR} W_{-\bb} \nn \\ 
&& +C\rho^{1/2}(-\aR)\rho^{1/2}(\bb)W_{-\aR}W_\bb+D\rho^{1/2}(-\aR)\rho^{1/2}(-\bb)W_{-\aR}W_{-\bb}\  , 
\end{eqnarray}
factorises the corresponding biquadratic to the quadrilateral, namely it obeys
\be\label{eq:Hfact} 
\cH_\ssp(u^{(N)},\wt{u}^{(N)})=U^{(N)}\wt{U}^{(N)}\quad,\quad 
\cH_\ssq(u^{(N)},\wh{u}^{(N)})=U^{(N)}\wh{U}^{(N)} \ee 
where 
\bse\label{eq:Hbiquad}\bea 
\cH_\ssp(u,\wt{u})&\Equiv &P(u^2+\wt{u}^2)-(2p^2-a^2-b^2)u\wt{u}+\frac{\Delta}{P} \\ 
\cH_\ssq(u,\wt{u})&=&Q(u^2+\wh{u}^2)-(2q^2-a^2-b^2)u\wh{u}+\frac{\Delta}{Q}\  .  
\eea\ese  
%}
\end{theorem} 
\paragraph{\it Proof:} The proof is obtained through direct computation, using the relations 
between the various objects established earlier. It follows roughly the same lines of 
computation as in the rational case \cite{NAH}, but only differs in a few 
subtleties regarding the elliptic representation. It  breaks down into three steps.

\noindent 
\underline{\it step \# 1:} Using the basic relations established between the objects $W_\aR$, $W_\bb$ 
and $S_{\aR,\bb}$ one can obtain relations between the objects $u^{(N)}$, $U^{(N)}$ 
and the H1 object $w$, as given in the following lemma. These constitute what in effect is a Miura transformation 
betweem Q3 and H1. 

\begin{lemma} The following relations hold between the elliptic $N$-soliton solution $u=u^{(N)}$ of Q3, together with 
its associated object $U=U^{(N)}$, and the elliptic $N$-soliton solution $w$ of H1: 
\bse\label{eq:wuUMiura}\begin{eqnarray}
\wh{w}-\wt{w}&=&\frac{P\wh{u}-Q\wt{u}-(p^2-q^2)\wh{\wt{u}}}{\wh{\wt{U}}}
 =-\frac{P\wt{u}-Q\wh{u}-(p^2-q^2)u}{U} \label{eq:wuUMiuraa} \\ 
w-\wh{\wt{w}}&=&\frac{Pu-Q\wh{\wt{u}}-(p^2-q^2)\wt{u}}{\wt{U}}
 =-\frac{P\wh{\wt{u}}-Qu-(p^2-q^2)\wh{u}}{\wh{U}} \label{eq:wuUMiurab} 
\end{eqnarray}\ese  
\end{lemma} 
\begin{proof}
These relations, which can be obtained by direct computation term by term in the expressions \eqref{eq:Nsol} and \eqref{eq:assNsol} 
(after multiplying out the denominators), namely as follows: 
\begin{eqnarray*}
&& P\wh{u}-Q\wt{u}-(p^2-q^2)\wh{\wt{u}} = \\ 
&& \qquad = A\wh{\wt{\rho}}^{1/2}(\aR)\wh{\wt{\rho}}^{1/2}(\bb)
\left[ P\left(e^{2\wh{\zeta}(\dd)\aR}\frac{p_\aR}{p_{-\aR}}\right)^{1/2}
\left(e^{2\wh{\zeta}(\dd)\bb}\frac{p_\bb}{p_{-\bb}}\right)^{1/2}e^{[-\zeta(\wh{\xi}+\dd)+\zeta(\wh{\xi})](\aR+\bb)}\wh{S}_{\aR,\bb}\right. \\ 
&& \qquad\qquad\qquad 
\left. -Q\left(e^{2\wt{\zeta}(\ven)\aR}\frac{q_\aR}{q_{-\aR}}\right)^{1/2}
\left(e^{2\wt{\zeta}(\ven)\bb}\frac{q_\bb}{q_{-\bb}}\right)^{1/2}e^{[-\zeta(\wt{\xi}+\ven)+\zeta(\wt{\xi})](\aR+\bb)}\wt{S}_{\aR,\bb}
-(p^2-q^2) \wh{\wt{S}}_{\aR,\bb} \right] + \cdots \\ 
&& \qquad = A\wh{\wt{\rho}}^{1/2}(\aR)\wh{\wt{\rho}}^{1/2}(\bb)
\left[ p_\aR p_\bb e^{-\wh{\eta}_\dd(\aR+\bb)} \wh{S}_{\aR,\bb}- q_\aR q_\bb e^{-\wt{\eta}_\ven(\aR+\bb)} \wt{S}_{\aR,\bb} 
-(p^2-q^2) \wh{\wt{S}}_{\aR,\bb}\right] + \cdots \\ 
&& \qquad = A\wh{\wt{\rho}}^{1/2}(\aR)\wh{\wt{\rho}}^{1/2}(\bb)\left[ 
p_\bb e^{-\wh{\eta}_\dd\bb} \left( p_\aR e^{-\wh{\eta}_\dd\aR} \wh{S}_{\aR,\bb}- p_{-\bb}e^{\wh{\eta}_\dd\bb}\wh{\wt{S}}_{\aR,\bb}\right)
-q_\bb e^{-\wt{\eta}_\ven\bb} \left( q_\aR e^{-\wt{\eta}_\ven\aR} \wt{S}_{\aR,\bb}- q_{-\bb}e^{\wt{\eta}_\ven\bb}\wh{\wt{S}}_{\aR,\bb}\right)\right] + \cdots \\ 
&& \qquad = A\wh{\wt{\rho}}^{1/2}(\aR)\wh{\wt{\rho}}^{1/2}(\bb)\left[ 
p_\bb e^{-\wh{\eta}_\dd\bb}\wh{\wt{W}}_\aR \wh{W}_\bb- q_\bb e^{-\wt{\eta}_\ven\bb}\wh{\wt{W}}_\aR \wt{W}_\bb \right] +\cdots \\ 
&& \qquad = A\wh{\wt{\rho}}^{1/2}(\aR)\wh{\wt{\rho}}^{1/2}(\bb)(\wh{w}-\wt{w})\wh{\wt{W}}_\aR\wh{\wt{W}}_\bb +\cdots 
\end{eqnarray*}
by using in the last steps the basic relation \eqref{eq:WW} in combination with \eqref{eq:wWMiura}, and where the $\cdots$ stand for similar 
expressions for the remaining terms with coefficients $B$, $C$, $D$ instead of $A$, and with $(\aR,\bb)$ replaced by $(\aR,-\bb)$, $(-\aR,\bb)$ and 
$(-\aR,-\bb)$ respectively. Thus, identifying the right hand side, up to a common factor $(\wh{w}-\wt{w})$ with the variable $\wh{\wt{U}}$, 
we obtain the first relation in \eqref{eq:wuUMiuraa}. The other relations in \eqref{eq:wuUMiura} follow by similar computations. 
\end{proof}

\noindent
\underline{\it step \# 2:} We next establish the main  
identities between ~$u=u^{(N)}_{n,m}$, and $U=U^{(N)}_{n,m}$. 

\begin{lemma} the following factorisation formulae for the biquadratics associated with the Q3 quadrilateral hold:
\bse\label{eq:QrMiuras}\begin{eqnarray}
U\wt{U}-P(u^2+\wt{u}^2)+(2p^2-a^2-b^2)u\wt{u} &=&\frac{\wp'(\aR)\,\wp'(\bb)}{P}\det(\cA)\  , \label{eq:Q3Miuraa} \\ 
U\wh{U}-Q(u^2+\wh{u}^2)+(2q^2-a^2-b^2)u\wh{u} &=&\frac{\wp'(\aR)\,\wp'(\bb)}{Q}\det(\cA)\  , \label{eq:Q3Miurab} 
\end{eqnarray}\ese 
in which the $2\times 2$ matrix $\cA$ is given by 
\begin{equation}\label{eq:A} 
\cA=\left(\begin{array}{cc} A & B \\ C & D \end{array}\right)\   . 
\end{equation}
\end{lemma}
\begin{proof}
This can asserted by computing the (Lax type) matrices 
\bse\label{eq:LM}\bea 
\bL&\Equiv &\left(\begin{array}{ccc} P\wt{u}-(p^2-b^2)u &,& e^{\eta_\dd\bb} p_{-\bb}U \\ 
e^{-\eta_\dd\bb} p_\bb\wt{U} &,& -Pu+(p^2-b^2)\wt{u}\end{array}\right)\  , \label{eq:L} \\ 
\bM&\Equiv &\left(\begin{array}{ccc} Q\wh{u}-(q^2-b^2)u &,& e^{\eta_\ven\bb} q_{-\bb}U \\ 
e^{-\eta_\ven\bb} q_\bb\wh{U} &,& -Qu+(q^2-b^2)\wh{u}\end{array}\right)\   . \label{eq:M}  
\eea\ese 
in two different ways: one directly, and the other by expressing the entries in terms of our basic objects. 
In fact, introducing the 2-component vectors: 
\be\label{eq:rr} 
\brr^T_\aR=\left( \rho^{1/2}(\aR)\,W_\aR,\rho^{1/2}(-\aR)\,W_{-\aR}\right)\quad,\quad 
\brr_\bb=\left( \begin{array}{c} \rho^{1/2}(\bb)\,W_\bb \\ \rho^{1/2}(-\bb)\,W_{-\bb}\end{array} \right)\  , 
\ee 
we can express $\bL$ as:
\be\label{eq:LL} \bL=\left(\begin{array}{ccc} -\sqrt{p_\bb p_{-\bb}}\,\brr^T_\aR\cA\wt{\brr}_\bb &,& e^{\eta_\dd\bb} p_{-\bb}\brr^T_\aR\cA\brr_\bb\\ 
e^{-\eta_\dd\bb} p_\bb\wt{\brr}^T_\aR\cA\wt{\brr}_\bb &,& -\sqrt{p_\bb p_{-\bb}}\,\wt{\brr}^T_\aR\cA\brr_\bb \end{array}\right)\  , 
\ee 
and a similar expression for $\bM$. In fact, 
\begin{eqnarray*}
&& P\wt{u}-(p^2-b^2)u = \nn \\ 
&=& A\rho^{1/2}(\aR)\rho^{1/2}(\bb)\left[ p_{-\aR}p_{-\bb} e^{\eta_\dd(\aR+\bb)} \wt{S}_{\aR,\bb}-p_\bb p_{-\bb}S_{\aR,\bb}\right]+ \dots \nn \\ 
&=& A\rho^{1/2}(\aR)\rho^{1/2}(\bb) e^{\eta_\dd\bb} p_{-\bb}\left(p_{-\aR} e^{\eta_\dd\aR} \wt{S}_{\aR,\bb}-p_\bb e^{-\eta_\dd\bb} S_{\aR,\bb}\right) + \dots \nn \\ 
&=& -A\rho^{1/2}(\aR)\rho^{1/2}(\bb)e^{\eta_\dd\bb} p_{-\bb}\wt{W}_\bb W_\aR+ \dots \nn \\ 
&=& -\sqrt{p_\bb p_{-\bb}}\left( A\rho^{1/2}(\aR)\wt{\rho}^{1/2}(\bb)W_\aR \wt{W}_\bb +\dots \right) = 
-\sqrt{p_\bb p_{-\bb}}\,\brr^T_\aR\cA\wt{\brr}_\bb\  , 
\end{eqnarray*} 
and similarly, 
$$
Pu-(p^2-b^2)\wt{u}=\sqrt{p_\bb p_{-\bb}}\,\wt{\brr}^T_\aR\cA\brr_\bb\  . 
$$ 
Evaluating the determinant of $\bL$ using the general identity
$$  \det\left(\sum_{j=1}^r\,\mbx_j\mby_j^T\right)=\det\left( (\mby_i^T\cdot\mbx_j)_{i,j=1,\cdots,r}\right)  $$ 
holding for any $2r$ $r$-component vectors $\mbx$, $\mby$, we now get: 
\begin{eqnarray*}
\det(\bL)&=&p_\bb p_{-\bb}\det\left(\cA\wt{\brr}_\bb\,\brr^T_\aR+\cA\brr_\bb\,\wt{\brr}^T_\aR\right) \\ 
&=&(p^2-b^2)\det(\cA) \det\left(\wt{\brr}_\bb\,\brr^T_\aR+\brr_\bb\,\wt{\brr}^T_\aR\right)\\ 
&=& (p^2-b^2)\det(\cA)\det\left\{\left(\wt{\brr}_\bb,\brr_\bb\right)\,\left(\begin{array}{c}\brr^T_\aR\\
 \wt{\brr}^T_\aR\end{array}\right)\right\} \\ 
&=& -(p^2-b^2)\det(\cA)\det\left(\brr_\aR,\wt{\brr}_\aR\right)\,\det\left(\brr_\bb,\wt{\brr}_\bb\right)\  . 
\end{eqnarray*}

It remains to compute the determinant of the matrix $\left(\brr_\aR,\wt{\brr}_\aR\right)$~ whose columns are the 
2-component vectors $\brr_\aR$ and $\wt{\brr}_\aR$. This is done by using  \eqref{eq:WWs} as follows:
\begin{eqnarray*}
\det\left(\brr(a),\wt{\brr}(a)\right)&=& \rho^{1/2}(\aR)\wt{\rho}^{1/2}(-\aR)W_\aR\wt{W}_{-\aR}
-\wt{\rho}^{1/2}(\aR)\rho^{1/2}(-\aR)\wt{W}_\aR\,W_{-\aR} \\ 
&=& \sqrt{\frac{p_\aR}{p_{-\aR}}}\,e^{-\eta_\dd\aR} W_{\aR}\,\wt{W}_{-\aR}-\sqrt{\frac{p_{-\aR}}{p_\aR}}\,e^{\eta_\dd\aR} 
\wt{W}_\aR\,W_{-\aR}\\
%&=& \sqrt{\frac{p_\aR}{p_{-\aR}}p_{-\aR}\left\{ [\zeta(\xi)-\zeta(\wt{\xi})+\zeta(\aR)+\zeta(\dd-\aR)] 
%+[\wp(\wt{\xi})-\wp(\aR)]\wt{U}_{-\aR,\aR}-[\wp(\xi)-\wp(\aR)]U_{-\aR,\aR}\right\}  \\ 
%&& \qquad -\sqrt{\frac{p_{-\aR}}{p_\aR}}p_\aR\left\{ [\zeta(\xi)-\zeta(\wt{\xi})-\zeta(\aR)+\zeta(\dd+\aR)] 
%+[\wp(\wt{\xi})-\wp(\aR)]\wt{U}_{\aR,-\aR}-[\wp(\xi)-\wp(\aR)]U_{\aR,-\aR}\right\} \\ 
&=& \sqrt{p_\aR p_{-\aR}}\,[2\zeta(\aR)+\zeta(\dd-\aR)-\zeta(\dd+\aR)]=\frac{\wp'(\aR)}{\sqrt{p^2-a^2}} \   .  
\end{eqnarray*}

Thus, putting everything together we obtain the result:
$$
\det(\bL)=-(p^2-b^2)\det(\cA)\,\frac{\wp'(\aR)\,\wp'(\bb)}{\sqrt{p^2-a^2}\,\sqrt{p^2-b^2}}\  . 
$$
On the other hand a direct computation of the determinant gives:
\begin{eqnarray*}\label{eq:detLL}
\det(\bL)&=& -[P\wt{u}-(p^2-b^2)u] [Pu-(p^2-b^2)\wt{u}]-(p^2-b^2)U\wt{U}  \\ 
&=& (p^2-b^2)\left[ P(u^2+\wt{u}^2)-(2p^2-a^2-b^2)u\wt{u}-U\wt{U}\right]\  . 
\end{eqnarray*}
which yields the first of \eqref{eq:Hfact}, the second part of which follows by a similar computation of $\det(\bM)$, 
which is obviously obtained from the $\det(\bL)$ by replacing $\dd$ by $\ven$ and $\wt{\phantom{a}}$ by $\wh{\phantom{a}}$. 
\end{proof}

\underline{\it step \# 3:} The final step is to arrive at Q3  by combining the results of step \# 1 and \# 2  for the solution 
$u^{(N)}$ from the Miura relations \eqref{eq:wuUMiura}, together with the relations \eqref{eq:Hfact}. Thus, since $w$ obeys the H1 
equation, we have:  
\begin{eqnarray*}  
p^2-q^2 &=& (w-\wh{\wt{w}})(\wh{w}-\wt{w}) \\ 
&=& \frac{1}{\wt{U}\wh{\wt{U}}}\left[ P\wh{u}-Q\wt{u}-(p^2-q^2)\wh{\wt{u}}\right]\left[ Pu-Q\wh{\wt{u}}-(p^2-q^2)\wt{u}\right] \\ 
\Rightarrow  &&(p^2-q^2)\left[ Q(\wt{u}^2+\wh{\wt{u}}^2)-(2q^2-a^2-b^2)\wt{u}\wh{\wt{u}}-\frac{\Delta}{Q}\right] \\ 
%&=& P^2(u\wh{u})+Q^2(\wt{u}\wh{\wt{u}}) +(p^2-q^2)^2\wt{u}\wh{\wt{u}} -PQ(u\wt{u}+\wh{u}\wh{\wt{u}}) + \\
%&&\qquad -(p^2-q^2)P(\wh{u}\wt{u}+u\wh{\wt{u}}) +(p^2-q^2)Q(\wt{u}^2+\wh{\wt{u}}^2)\  \\ 
&=& P^2(u\wh{u}+\wt{u}\wh{\wt{u}})+(Q^2-P^2)(\wt{u}\wh{\wt{u}}) +(p^2-q^2)^2\wt{u}\wh{\wt{u}} -PQ(u\wt{u}+\wh{u}\wh{\wt{u}}) + \\
&&\qquad -(p^2-q^2)P(\wh{u}\wt{u}+u\wh{\wt{u}}) +(p^2-q^2)Q(\wt{u}^2+\wh{\wt{u}}^2)\  
\end{eqnarray*}
%and thus, 
%\begin{eqnarray*}  
%&&(p^2-q^2)\left[-(2q^2-a^2-b^2)\wt{u}\wh{\wt{u}}+\frac{\Delta}{Q}\right]= \\ 
%&=& P\left[ P(u\wh{u}+\wt{u}\wh{\wt{u}})-Q(u\wt{u}+\wh{u}\wh{\wt{u}})-(p^2-q^2)(\wh{u}\wt{u}+u\wh{\wt{u}})\right] + \\ 
%&& \quad + (p^2-q^2)\left[ (a^2+b^2-p^2-q^2)\wt{u}\wh{\wt{u}}+(p^2-q^2)\wt{u}\wh{\wt{u}} \right]  
%\end{eqnarray*} 
which after some cancelations leads to Q3 in the form \eqref{eq:Q3new} and, hence, concludes the proof of the Theorem. 
\DONE

\subsection{\bf Hirota form of the elliptic $N$-soliton solution} 

The $\tau$-function can be explicitly computed by using the expansion 
%\begin{eqnarray*}
%f=\det\left(\ld\boldsymbol{1}+\boldsymbol{M}\right)&=& \ld^N+\ld^{N-1}\sum_{i=1}^N \left|M_{i,i}\right| 
%+\ld^{N-2}\sum_{i<j}\left|\begin{array}{cc} M_{i,i} & M_{i,j} \\ M_{j,i} & M_{j,j} \end{array}\right| \nn \\ 
%&& +\ld^{N-3}\sum_{i<j<k}\left|\begin{array}{ccc} M_{i,i} & M_{i,j} & M_{i,k}\\ 
%M_{j,i} & M_{j,j} & M_{j,k}\\ M_{k,i} & M_{k,j} & M_{k,k} \end{array}\right| + \cdots + \det(\boldsymbol{M})\  . \nn \\ 
%\end{eqnarray*}
$$
f=\det\left(\boldsymbol{1}+\boldsymbol{M}\right)= 1+ \sum_{i=1}^N \left|M_{i,i}\right| 
+\sum_{i<j}\left|\begin{array}{cc} M_{i,i} & M_{i,j} \\ M_{j,i} & M_{j,j} \end{array}\right| + \cdots + \det(\boldsymbol{M})\  .  
%&& +\ld^{N-3}\sum_{i<j<k}\left|\begin{array}{ccc} M_{i,i} & M_{i,j} & M_{i,k}\\ 
%M_{j,i} & M_{j,j} & M_{j,k}\\ M_{k,i} & M_{k,j} & M_{k,k} \end{array}\right| + \cdots + \det(\boldsymbol{M})\  . \nn \\ 
$$ 
and using the Frobenius formula for the relevant elliptic Cauchy determinants, cf. \cite{Frob}. Thus, from 
\be\label{eq:Frob}
\det\left(\rho_i\Psi_{\kp_i+\kp_j}(\xi)c_j\right)=\left(\prod_i\frac{\rho_i c_i}{\sg(2\kp_i)}\right)
\frac{\sg(\xi+2\sum_i\kp_i)}{\sg(\xi)}\,e^{-2\zeta(\xi)\sum_i \kp_i}\prod_{i<j}\left(\frac{\sg(\kp_i-\kp_j)}{\sg(\kp_i+\kp_j)}\right)^2\  .  
\ee
Introducing the notations 
\[ e^{A_{i,j}}\Equiv  \left(\frac{\sg(\kp_i-\kp_j)}{\sg(\kp_i+\kp_j)}\right)^2\quad,\quad e^{\theta_i}= \frac{\rho_i c_i}{\sg(2\kp_i)}e^{-2\zeta(\xi)\kp_i}\  , \] 
the Hirota formula for the $\tau$-function thus takes the form: 
\begin{eqnarray} 
f&=& 1+\sum_{i=1}^N \frac{\sg(\xi+2\kp_i)}{\sg(\xi)}e^{\theta_j}+\sum_{i<j=1}^N\, \frac{\sg(\xi+2\kp_i+2\kp_j)}{\sg(\xi)}e^{\theta_i+\theta_j+A_{i,j}}+ \\ 
&&+\sum_{i<j<k=1}^N\, \frac{\sg(\xi+2\kp_i+2\kp_j+2\kp_k)}{\sg(\xi)}e^{\theta_i+\theta_j+\theta_k+A_{i,j}+A_{i,k}+A_{j,k}}+\cdots\  . \nn \\
&& \label{eq:Hirform} 
\end{eqnarray}
The expression \eqref{eq:Hirform} for $\tau$-function $f$ enters solution of Q3 through the  main quantities $S_{\aR,\bb}$ and $W_\aR$ in the 
following way. First, we note that $V_\dd$, and hence $W_\dd$ can be expressed in terms of $f$ via \eqref{eq:ff}. Second, by setting $\aR=\ven$, 
$\bb=\dd$ in \eqref{eq:pVV}, using also \eqref{eq:ff}, we obtain 
\be\label{eq:Sddven}
 S_{\dd,\ven}=\Psi_\xi(\dd+\ven)\,\frac{\wh{\wt{f}}}{f}\  .   \ee
Thus for special parameters $\aR,\bb$ associated with the lattice parameters, we have an explicit expression for these quantities in terms of 
the $\tau$-function. To get a similar expression for arbitrary $\aR$, $\bb$ we just need to extend the two-dimensional lattice in terms of the 
variables $n$ and $m$ (associated with lattice parameters $\dd$ and $\ven$ respectively) to a four-dimensional one, containing additional 
lattice directions associated with the parameters $\aR$ and $\bb$. The discrete variables associated with these parameters we will denote by 
$h$ and $l$ respectively, and hence the extension amounts to including in the plane-wave factors $\rho(\kp)$ a dependence on these variables 
as follows:   
\be\label{eq:rhoextend}
 \rho(\kp)= \left(e^{-2\zeta(\dd)\kp}\frac{p_{-\kp}}{p_\kp}\right)^n\left(e^{-2\zeta(\ven)\kp}\frac{q_{-\kp}}{q_\kp}\right)^m
\left(e^{-2\zeta(\aR)\kp}\frac{a_{-\kp}}{a_\kp}\right)^h\left(e^{-2\zeta(\bb)\kp}\frac{b_{-\kp}}{b_\kp}\right)^l e^{2\zeta(\xi)\kp}
\rho_{0,0,0,0}(\kp)\   ,  \ee 
whilst extending at the same time the variable $\xi$ covariantly as ~$\xi=\xi_0+n\dd+m\ven+h\aR+l\bb$~~. In \eqref{eq:rhoextend}  
we have set $a_\kp=\Psi_\aR(\kp)$ and $b_\kp=\Psi_\bb(\kp)$. In this extended lattice the elementary lattice shifts $T_\aR$, $T_\bb$ 
associated with these new lattice directions can be used to express $S_{\aR,\bb}$ and $W_\aR$ in terms of the $\tau$-function, in a way 
similar to \eqref{eq:Sddven} and \eqref{eq:ff}, leading to the expressions 
\be\label{eq:Stautau} 
S_{\aR,\bb}=\Psi_\xi(\aR+\bb) \frac{T_\aR T_\bb f}{f}\quad,\quad W_\aR=\Psi_\xi(\aR) \frac{T_\aR f}{f}\   ,   \ee  
and where the $T_\aR$, $T_\bb$ denote elementary lattice shifts in the directions associated with the parameters 
$\aR$, $\bb$. 

Combining now the expression \eqref{eq:Stautau} with the explicit form of the the solution \eqref{eq:Nsol} and of its 
associated variable \eqref{eq:assNsol} we obtain the following result. 
\begin{theorem}\label{th:theoremf}
The elliptic $N$-soliton solution \eqref{eq:Nsol} of Q3 can be written as ~$u^{(N)}=\left({\mathcal T}_{\cA}f\right)/f$~, 
where the shift operator $\mathcal{T}_{\cA}$ is given by 
\be\label{eq:Q3Nsol} 
%u^{(N)}=\frac{{\mathcal T}_{\cA}f}{f}\quad, \quad 
\mathcal{T}_{\cA}f\Equiv  {\rm Tr}\left\{ \brr^0_\bb\,(\brr^0_\aR)^T 
\left(\begin{array}{ccc} A\Psi_\xi(\aR+\bb)T_\aR T_\bb &,& B\Psi_\xi(\aR-\bb)T_\aR T_{\bb}^{-1} \\ 
C\Psi_\xi(-\aR+\bb)T_\aR^{-1}T_\bb &,& D\Psi_\xi(-\aR-\bb)T_\aR^{-1}T_\bb^{-1} \end{array}\right)f\right\}
\ee 
whereas the associated solution \eqref{eq:assNsol} can be written as ~$U^{(N)}=\left(\mathcal{D}_{\cA}f\cdot f\right)/f^2$~ 
where the bilinear operator $\mathcal{D}_{\cA}$ is given by
\be
%U^{(N)}=\brr_\aR^T\cA\brr_\bb=\frac{\mathcal_{\cA}f}{f}\quad, \quad 
\mathcal{D}_{\cA} f\cdot g \Equiv  {\sf m}\left[{\rm Tr}\left\{ \brr^0_\bb\,(\brr^0_\aR)^T 
\left(\begin{array}{ccc} A\Psi_\xi(\aR)\Psi_\xi(\bb)T_\aR\otimes T_\bb &,& B\Psi_\xi(\aR)\Psi_\xi(-\bb)T_\aR\otimes T_{\bb}^{-1} \\ 
C\Psi_\xi(-\aR)\Psi_\xi(\bb)T_\aR^{-1}\otimes T_\bb &,& D\Psi_\xi(-\aR)\Psi_\xi(-\bb)T_\aR^{-1}\otimes T_\bb^{-1}\end{array} \right) 
\right\}(f\otimes g)\right] 
\ee 
with ${\sf m}$ denoting the pointwise product in the space of functions of the lattice sites,  
making use of the notation \eqref{eq:rr}, \eqref{eq:A} as well as 
\be\label{eq:rr0} 
(\brr^0_\aR)^T=\left( \rho^{1/2}(\aR),\rho^{1/2}(-\aR)\right)\quad,\quad 
\brr_\bb^0=\left( \begin{array}{c} \rho^{1/2}(\bb) \\ \rho^{1/2}(-\bb)\end{array} \right)\  . 
\ee 
\end{theorem}

\subsection{Linear scheme for elliptic soliton solutions} 

We will now show how the corresponding solutions constructed in section 4 arise naturally from a scheme of linear 
equations, which in turn leads to a novel Lax representation for Q3. These also explain the emergence of the 
Lax type matrices that were used in the proof of Theorem \ref{th:theorem1}. 
 
%\subsection{Linear structure}  

The objects \eqref{eq:butbu} form the basis of the linear structures underlying the lattice equations that we have studied in the previous 
sections. In fact, they form the main ingredients, on the basis of which we can derive Lax pairs for those lattice equations\footnote{An alternative 
approach being the one using the multidimensional consistency property, cf. \cite{Nij,BS,NW} to derive Lax pairs from the equations themselves, but 
the resulting Lax pairs often require awkward prefactors containing square roots of the relevant biquadratics.}.   

%\begin{definition}
Associated with the Cauchy matrix $\bM$ we introduce the following parameter-dependent $2N$-component eigenfunctions;
\bse\label{eq:BuTbu}\begin{eqnarray}\label{eq:Bu} 
\bu_{\bK}(\bb) &\Equiv& (\bun+\bM)^{-1} \left(\begin{array}{cc} \Psi_\xi(\bb+\bK) &0\\ 
                                        0 & \Psi_\xi(-\bb+\bK) 
                                       \end{array}\right) \left(\begin{array}{c}
                                                                 \rho^{1/2}(\bb) \\ \rho^{1/2}(-\bb)
                                                                \end{array}\right) \brr^0 \\ 
\tbu_{\bK}(\aR) &\Equiv& (\bs^0)^T \left( \rho^{1/2}(\aR), \rho^{1/2}(-\aR)\right) 
                                  \left(\begin{array}{cc} \Psi_\xi(\aR+\bK) &0\\ 
                                        0 & \Psi_\xi(-\aR+\bK) 
                                       \end{array}\right) (\bun+\bM)^{-1} \nn \\ 
\label{eq:Tbu} 
\end{eqnarray}\ese
where the vectors $\brr^0$ and $\tbs^0$ are given by 
$$ (\brr^0)^T=(\rho_1,\cdots,\rho_N)\quad,\quad (\bs^0)^T=(c_1,\cdots,c_N)\   . $$ 
%\end{definition}

Note that these vectors \eqref{eq:BuTbu} effectively live in the direct product space of the 2-component vectors labelled by the parameters $\pm\aR$, $\pm\bb$ and 
$N$-component vectors associated with the soliton parameters $\kp_i$, ($i=1,\dots,N$). The dynamical properties of these vectors are described by the 
following set of relations:

\begin{prop}
The following set of relations describe the behaviour of the vectors given in \eqref{eq:BuTbu} under the elementary shift in the coordinate $n$ of the multidimensional 
lattice
\bse\label{eq:puapub}\begin{eqnarray}
e^{\eta_\dd\bK}p_{-\bK}\bu_{\bK}(\bb) &=& \sqrt{p_{\bb} p_{-\bb}}\, \wt{\bu}_{\bK}(\bb) + \brr_\bb \wt{\bu}^0_{\bK}\   ,\label{eq:pub} \\ 
e^{-\eta_\dd\bK}p_{\bK}\wt{\bu}_{\bK}(\bb)&=&  \sqrt{p_\bb p_{-\bb}}\, \bu_{\bK}(\bb) -\wt{\brr}_\bb \bu^0_{\bK}\  , \label{eq:upb} \\ 
\, ^{t\!}\wt{\bu}_{\bK}(\aR)p_{-\bK}e^{\eta_\dd\bK} &=& \sqrt{p_\aR p_{-\aR}}\, \tbu_{\bK}(\aR) - \tbu^0_{\bK} \wt{\brr}^T_\aR \  , \label{eq:pua} \\ 
\tbu_{\bK}(\aR) p_{\bK}e^{-\eta_\dd\bK} &=& \sqrt{p_\aR p_{-\aR}}\, ^{t\!}\wt{\bu}_{\bK}(\aR) +  \, ^{t\!}\wt{\bu}^0_{\bK}\brr^T_\aR\   . \label{eq:upa} 
\end{eqnarray}\ese 
where the 2-component vectors $\brr^T_\aR$ and $\brr_\bb$ are given in \eqref{eq:rr}, and where the $N$-component vectors $\bu_{\bK}^0$ and $\tbu_{\bK}^0$ are given by 
\be\label{eq:ukp0}  \bu_{\bK}^0=(\bun+\bM)^{-1}\brr\quad,\quad \tbu^0_{\bK}=\bs^T(\bun+\bM)^{-1} \  , \ee 
recalling the vectors $\brr$ and $\bs$ as defined in \eqref{eq:rs}. 
The shifts in the discrete variable $m$ (and other lattice directions) are described by similar relations like \eqref{eq:puapub}, by replacing 
$\dd$ is replaced by $\ven$ (i.e. $p_{\pm\bK}$ replaced by $q_{\pm\bK}$), and $\wt{\phantom{a}}$ by $\wh{\phantom{a}}$. 
\end{prop}
\begin{proof}
The proof of these relations follow similar derivations as those in section 3, i.e. making use of the basic relations \eqref{eq:pM} (and \eqref{eq:qM}), for 
the elliptic Cauchy matrix, as well as the definitions of the various quantities such the plane-wave factors $\rho$, \eqref{eq:rho} and the objects 
\eqref{eq:V}. For instance, to prove \eqref{eq:pub} we proceed as follows:
\begin{eqnarray*}
&& (\bun+\wt{\bM}) e^{\eta_\dd\bK}p_{-\bK}\bu_{\bK}(\bb)= \left[e^{\eta_\dd\bK}p_{-\bK}(\bun+\bM)-\wt{\brr}\,\bs^T \right]\bu_{\bK}(\bb) \\ 
&& = - (\bs^T\bu_{\bK}(\bb))\wt{\brr} + \,e^{\eta_\dd\bK}p_{-\bK}\left(\begin{array}{cc} \Psi_\xi(\bb+\bK) &0\\ 
                                        0 & \Psi_\xi(-\bb+\bK) 
                                       \end{array}\right) \left(\begin{array}{c}
                                                                 \rho^{1/2}(\bb) \\ \rho^{1/2}(-\bb)
                                                                \end{array}\right) \brr^0  
\end{eqnarray*}
in which the entries in the matrix term on the right hand side are computed as follows 
\begin{eqnarray*}
 && e^{\eta_\dd\bK}p_{-\bK}\Psi_\xi(\pm\bb+\bK) \rho^{1/2}(\pm\bb)\brr^0= e^{\eta_\dd\bK}p_{-\bK}\Psi_\xi(\pm\bb+\bK)\rho^{1/2}(\pm\bb) 
    e^{-2\eta_\dd\bK}\frac{p_{\bK}}{p_{-\bK}}\wt{\brr}^0 \\ 
&& =  e^{-\eta_\dd\bK}p_{\bK}\Psi_\xi(\pm\bb+\bK) \rho^{1/2}(\pm\bb)\wt{\brr}^0 = e^{\eta_\dd\bK}\Psi_\dd(\bK)\Psi_\xi(\pm\bb+\bK)\wt{\rho}^{1/2}(\pm\bb) 
    e^{\mp\eta_\dd\bb}\left(\frac{p_{\pm\bb}}{p_{\mp\bb}}\right)^{1/2}\wt{\brr}^0 \\
&& =  e^{-\eta_\dd\bK}\left[e^{\eta_\dd\bK}\Psi_{\xi+\dd}(\bK)\Psi_\xi(\pm\bb)+e^{\eta_\dd(\pm\bb+\bK)}\Psi_\dd(\mp\bb)\Psi_{\xi+\dd}(\pm\bb+\bK)\right]
     \wt{\rho}^{1/2}(\pm\bb) e^{\mp\eta_\dd\bb}\left(\frac{p_{\pm\bb}}{p_{\mp\bb}}\right)^{1/2}\wt{\brr}^0 \\ 
&& = \Psi_{\xi+\dd}(\bK)\Psi_\xi(\pm\bb)\rho^{1/2}(\pm\bb)\wt{\brr}^0 + \sqrt{p_\bb p_{-\bb}}\,\wt{\rho}^{1/2}(\pm\bb)\Psi_{\xi+\dd}(\pm\bb+\bK) \wt{\brr}^0 
\end{eqnarray*}
where use has been made of the addition formula \eqref{eq:Psiidsa} and the shift relations for the plane wave factors $\rho(\pm\bb)$ and in the components of
the vector $\brr^0$. Inserting these expressions in the former relation, and multiplying both sides by ~$(\bun+\wt{\bM})^{-1}$~, whilst taking note of the fact that 
~$\Psi_{\xi+\dd}(\bK)\wt{\brr}^0=\brr$~ according to the definitions, we obtain \eqref{eq:pub}, by identifying

\begin{eqnarray*}
 \bs^T\,\bu_{\bK}(\bb)&=&\bs^T\,(\bun+\bM)^{-1}\left(\begin{array}{cc}
                                                    \Psi_\xi(\bb)\chi_{\bb,\bK}^{-1} & 0 \\ 
                                                    0 & \Psi_\xi(-\bb)\chi_{-\bb,\bK}^{-1}
                                                   \end{array}\right) \Psi_\xi(\bK)\left(\begin{array}{c}
                                                                 \rho^{1/2}(\bb) \\ \rho^{1/2}(-\bb)
                                                                \end{array}\right) \brr^0 \\ 
&=& \left(\begin{array}{c}\rho^{1/2}(\bb)\Psi_\xi(\bb)(1-V_\bb) \\ \rho^{1/2}(-\bb)\Psi_\xi(-\bb)(1-V_{-\bb}))
                                                                \end{array}\right)   
\end{eqnarray*}
and rewriting these 2-vectors in terms of $W_{\pm\bb}$, using the definitions \eqref{eq:WS} and \eqref{eq:rr}. Similar computations yield the other 
relations in \eqref{eq:BuTbu}, where in the latter two we need to indentify the combination
$$ \tbu_{\bK}(\aR)\,\brr= \left(\rho^{1/2}(\aR)\Psi_\xi(\aR)(1-V_\aR)\,,\,\rho^{1/2}(-\aR)\Psi_\xi(-\aR)(1-V_{-\aR})\right)\  . $$  
\end{proof}

\subsection{Lax representation} 

Starting from the basic relations \eqref{eq:puapub} for the 2N-component vectors \eqref{eq:BuTbu} we will now define first $N$-component vectors by projecting in the 
2-component vector space by making use of the arbitrary coefficient matrix $\cA$ of \eqref{eq:A}.  Thus, multiplying \eqref{eq:pub} and \eqref{eq:upb} from the left 
by the 2-component row vector $\brr_{\aR}^T\cA$ we obtain respectively 
\bse\label{eq:rAu}\begin{eqnarray} 
e^{\eta_\dd\bK}p_{-\bK} \brr^T_{\aR}\cA\bu_{\bK}(\bb) &=& \sqrt{p_\bb p_{-\bb}}\,\brr^T_\aR\cA\wt{u}_{\bK}(\bb) + U^{(N)} \wt{\bu}^0_{\bK}\   \label{eq:rAua} \\ 
e^{-\eta_\dd\bK}p_{\bK} \brr^T_\aR\cA\wt{\bu}_{\bK}(\bb) &=& \sqrt{p_\bb p_{-\bb}}\,\brr^T_\aR\cA\bu_{\bK}(\bb) -
\left[\sqrt{p_\bb p_{-\bb}}\,u^{(N)}-\sqrt{p_{\aR} p_{-\aR}}\,\wt{u}^{(N)}\right]\bu^0_{\bK}\     \label{eq:rAub} 
\end{eqnarray}
where we have used the identities relating the expressions in the entries of the matrix \eqref{eq:L} and \eqref{eq:LL}. Similarly, by multiplying \eqref{eq:pub} and \eqref{eq:upb} 
from the left by the 2-component row vector $\wt{\brr}_\aR^T\cA$ we obtain respectively 
\begin{eqnarray}
e^{\eta_\dd\bK}p_{-\bK} \wt{\brr}^T_\aR\cA\bu_{\bK}(\bb) &=& \sqrt{p_\bb p_{-\bb}}\,\wt{\brr}^T_\aR\cA\wt{u}_{\bK}(\bb) + \left[\sqrt{p_\aR p_{-\aR}}\,u^{(N)}-\sqrt{p_\bb p_{-\bb}}\,\wt{u}^{(N)}\right]
\wt{\bu}^0_{\bK}\   \label{eq:rAuc} \\ 
e^{-\eta_\dd\bK}p_{\bK} \wt{\brr}^T_\aR\cA\wt{\bu}_{\bK}(\bb) &=& \sqrt{p_\bb p_{-\bb}}\,\wt{\brr}^T_{\aR}\cA\bu_{\bK}(\bb) -\wt{U}^{(N)}\bu^0_{\bK}\   .  
\label{eq:rAud} 
\end{eqnarray}\ese
Introducing now the new object 
\be\label{eq:ukab}
\bu_{\bK}(\aR,\bb)=\brr^T_{\aR}\cA\bu_{\bK}(\bb)\   . 
\ee 
which is a $N$-component vector, and eliminating the  quantity ~$\brr^T_\aR\cA\wt{\bu}_{\bK}(\bb)$~ from the pair of relations \eqref{eq:rAua} and \eqref{eq:rAub} we obtain  
\bse\label{eq:uAu}\be\label{uAua}
 (b^2-\bk^2) \bu_{\bK}(\aR,\bb) = p_{\bK}e^{-\eta_\dd\bK}U^{(N)}\wt{\bu}_{\bK}^0 + \left[ P\wt{u}^{(N)}-(p^2-b^2)u^{(N)} \right] \bu_{\bK}^0 
\ee
whilst from \eqref{eq:rAuc} and \eqref{eq:rAud} we get  
\be \label{eq:uAub}
 (b^2-\bk^2) \wt{\bu}_{\bK}(\aR,\bb) = -p_{-\bK}e^{\eta_\dd\bK}\wt{U}^{(N)}\bu_{\bK}^0 + \left[Pu^{(N)}-(p^2-b^2)\wt{u}^{(N)}  \right] \wt{\bu}_{\bK}^0\   .  
\ee\ese 
These relations can be used to constitute one part of a Lax pair, whilst the other part follows by replacing $\dd$ by $\ven$ (i.e. $p_{\pm\bK}$ replaced by $q_{\pm\bK}$), and 
$\wt{\phantom{a}}$ by $\wh{\phantom{a}}$. Thus, in terms of the 2N-component vector  
\be\label{eq:bphi}
\bphi_{\bK}(\aR,\bb)=\left( \begin{array}{c}
                            \bu_{\bK}^0 \\ \bu_{\bK}(\aR,\bb) 
                            \end{array}\right) 
\ee 
the following new Lax representation for Q3 is obtained
\bse\label{eq:newLax}\bea
p_{\bK}e^{-\eta_\dd\bK}\wt{\bphi}_{\bK}(\aR,\bb)&=&\frac{1}{U^{(N)}}\,\left( 
\begin{array}{ccc}
-P\wt{u}^{(N)}+(p^2-b^2)u^{(N)} &,& b^2-\bk^2  \\ 
\frac{b^2-p^2}{b^2-\bk^2} \frac{\Delta}{P} - U^{(N)}\wt{U}^{(N)} &,& Pu^{(N)}-(p^2-b^2) \wt{u}^{(N)} 
\end{array}\right)\, \bphi_{\bK}(\aR,\bb) \nn \\ 
&& \label{eq:newL}  \\
q_{\bK}e^{-\eta_\ven\bK}\wh{\bphi}_{\bK}(\aR,\bb)&=&\frac{1}{U^{(N)}}\,\left( 
\begin{array}{ccc}
-Q\wh{u}^{(N)}+(q^2-b^2)u^{(N)} &,& b^2-\bk^2  \\ 
\frac{b^2-q^2}{b^2-\bk^2} \frac{\Delta}{Q} - U^{(N)}\wh{U}^{(N)} &,& Qu^{(N)}-(q^2-b^2) \wh{u}^{(N)} 
\end{array}\right)\, \bphi_{\bK}(\aR,\bb) \nn \\ 
&& \label{eq:newM}\  . 
\eea\ese 

In a similar way we can use the relationns \eqref{eq:pua} and \eqref{eq:upa} to derive an ``adjoint'' Lax pair for the row vectors defined by 
\be\label{eq:tbphi} 
\tbphi_{\bK}(\aR,\bb)=\left(\tbu^0_{\bK}\,,\,\tbu_{\bK}(\aR,\bb)\right)\quad,\quad \tbu_{\bK}(\aR,\bb)=\tbu_{\bK}(\aR)\cA\brr_\bb\  , 
\ee  
for the entries of which one can derive the following coupled set of equations:
\bse\begin{eqnarray}
\tbu_{\bK}(\aR,\bb)(a^2-\bk^2) &=&\,^{t\!}\wt{\bu}^0_{\bK}\, p_{-\bK}e^{\eta_\dd\bK} U^{(N)} +\tbu^0_{\bK}\,\left[ P\wt{u}^{(N)}-(p^2-a^2)u^{(N)}\right]\  , \\ 
 \,^{t\!}\wt{\bu}_{\bK}(\aR,\bb)(a^2-\bk^2) &=& \,^{t\!}\wt{\bu}^0_{\bK}\,\left[ Pu^{(N)}-(p^2-a^2)\wt{u}^{(N)}\right]-\tbu^0_{\bK}\, p_{\bK} e^{-\eta_\dd\bK}
\wt{U}^{(N)}\   . 
\end{eqnarray}\ese
In terms of \eqref{eq:bphi} one can derive a similar Lax pair, namely
\bse\label{eq:newLaxadj}\bea
\,^{t\!}\wt{\bphi}_{\bK}(\aR,\bb)\,p_{-\bK} e^{\eta_\dd\bK} &=& \tbphi_{\bK}\,\frac{1}{U^{(N)}}\,\left( 
\begin{array}{ccc}
-P\wt{u}^{(N)}+(p^2-a^2)u^{(N)} &,& \frac{a^2-p^2}{a^2-\bk^2} \frac{\Delta}{P} - U^{(N)}\wt{U}^{(N)} \\ 
a^2-\bk^2 &,& Pu^{(N)}-(p^2-a^2) \wt{u}^{(N)} 
\end{array}\right) \nn \\ 
&& \label{eq:newLadj}  \\
\,^{t\!}\wh{\bphi}_{\bK}(\aR,\bb)\,q_{-\bK} e^{\eta_\ven\bK} &=& \tbphi_{\bK}\,\frac{1}{U^{(N)}}\,\left( 
\begin{array}{ccc}
-Q\wh{u}^{(N)}+(q^2-a^2)u^{(N)} &,& \frac{a^2-q^2}{a^2-\bk^2} \frac{\Delta}{Q} - U^{(N)}\wh{U}^{(N)} \\ 
a^2-\bk^2 &,& Qu^{(N)}-(q^2-a^2) \wh{u}^{(N)} 
\end{array}\right) \nn \\ 
&& \label{eq:newMadj}\  . 
\eea\ese

In \eqref{eq:newLax} the structure derived from the elliptic $N$-soliton solutions, following the computations in this section, leads to  2N$\times$2N  Lax 
matrices with an N$\times$N block structure, each block of which is labelled by the soliton parameters $\kp_i$. However, having derived the Lax pair, one can 
now consider the variables $\bk$ (which were diagonal matrices) as general spectral parameters, and assess the compatibility of the system \eqref{eq:newLax} 
on its own merit. This leads to the following statement:

\begin{theorem}
The compatibility conditions for the two linear equations \eqref{eq:newL} and \eqref{eq:newM} are satisfied if and only if 
$u^{(N)}$ obeys Q3 in the form \eqref{eq:Q3new}, making use of the biquadratic identities \eqref{eq:Hfact} together with 
\eqref{eq:Hbiquad}. The same holds true for the two linear equations \eqref{eq:newLadj} and \eqref{eq:newMadj}. 
\end{theorem}
\begin{proof}
The proof is by direct computation, where the (1,1) and (2,2) entries yield the equalities at the left-hand sides of \eqref{eq:wuUMiura}, 
whilst the (1,2) and (2,1) entries yield the equation inserting the explicit expressions for the biquadratics $\cH_\ssp$ and $\cH_\ssq$  
of \eqref{eq:Hbiquad}. 
\end{proof}

\section{B\"acklund scheme for the elliptic soliton solutions} 

In \cite{NAH} we showed by explicit computation that the recursive structure of the $N$-soliton solutions for Q3 coincides with the action of the 
B\"acklund transformation. Although this is to be expected, the mechanism of the proof reveals the role of the relations between 
the various quantities obeying different equations in the ABS list, in particular the Miura type transformations between them. 
In Theorem \ref{th:theorem2} we make the statement as in \cite{NAH}, but now for the elliptic soliton solutions, and uncover the various aspects of the proof 
which have to be adapted for the elliptic case. 

\subsection{B\"acklund transformation from the $N$- to the $(N+1)$-soliton solution} 

The following theorem relates the elliptic $(N+1)$-soliton solution of the form given in \eqref{eq:Nsol} to the corresponding 
elliptic $N$-soliton solution with one less soliton parameter, thus establishing the recursive structure between them. In the proof of this theorem an important role
is played by the {\it covariant extension} (introduced in \cite{AHN}) of the $N$-soliton solution into a new lattice direction associated with lattice parameter $\ssl=\ssk_{N+1}$. Thus, as in \cite{NAH}, the shift $\rho_i\rightarrow\ol{\rho}_i$ denotes in what follows the covariant 
extension of the lattice variables into a direction given by this new parameter, thereby exploiting the multidimensional consistency of Q3. 
Equivalently this can be described by the introduction of a new 
integer lattice variable $h$ associated with the parameter $k_{N+1}$, on which the plane wave-factors \eqref{eq:rho} depend 
through the intial values, in the usual fashion, i.e. including a new factors with teh lattice parameter $\dd$ replaced by $\kp_{N+1}$ and 
the lattice variable $n$ replaced by $h$. 
Having this in mind the following theorem describes the recursive structure between the elliptic $N$- and $(N+1)$-soliton solutions: 

\begin{theorem}\label{th:theorem2} 
Let $u^{(N)}$ be as defined in (\ref{eq:Nsol}) and let $\wb{u}^{(N+1)}$ be equal to $u^{(N+1)}$ as defined in (\ref{eq:Nsol}), 
depending on additional parameters $\kp_{N+1}$, $c_{N+1}$ and additional plane-wave factor $\rho_{N+1}$, and where all but the latter  
plane-wave factors $\rho_i$, $i=1,\dots,N$, as well as the discrete exponentials $\rho(\aR)$, $\rho(\bb)$ are replaced by  
\be\label{eq:extendrho} 
\ol{\rho}_i=e^{2\eta_{\kp_{N+1}}\kp_i}\frac{\sg(\kp_i-\kp_{N+1})}{\sg(\kp_i+\kp_{N+1})}\,\rho_i\quad,\quad (i\neq N+1)\quad,\quad 
\wb{\rho}(\kp)=e^{2\eta_{\kp_{N+1}}\kp}\frac{\sg(\kp-\kp_{N+1})}{\sg(\kp+\kp_{N+1})}\,\rho(\kp)\  ,  \ee
with $\kp=\pm\aR,\pm\bb$ respectively. Then $u^{(N)}$ is related to $\ol{u}^{(N+1)}$ by the B\"acklund transformation with 
B\"acklund parameter $\ssk_{N+1}\in\Gamma$,
that is the following equations hold
\begin{equation}\label{eq:Q3solBT} 
\cQ_{\ssp,\ssk_{N+1}}^\Delta(u^{(N)},\wt{u}^{(N)},\ol{u}^{(N+1)},\wt{\ol{u}}^{(N+1)})=0,\qquad
\cQ_{\ssq,\ssk_{N+1}}^\Delta(u^{(N)},\wh{u}^{(N)},\ol{u}^{(N+1)},\wh{\ol{u}}^{(N+1)})=0,
\end{equation}
in which $\cQ_{\ssp,\ssk_{N+1}}$, $\cQ_{\ssq,\ssk_{N+1}}$ are the quadrilateral expressions of the form given in 
\eqref{eq:Q3quad}
\end{theorem} 
\paragraph{\it Proof:} The proof of Theorem \ref{th:theorem2} follows a similar pattern as the corresponding proof in \cite{NAH}, apart from some refinements which involve the 
non-autonomous elliptic Cauchy kernel and the structure of the elliptic identities involved in the computations. In the first step we break down the $(N+1)$-soliton 
expression into components associated with the $N$-soliton solution. In the second step we apply the BT to the $N$-soliton solution and in the final 
step we compare the expressions obtained in the two previous steps. 

\paragraph{}{\it Step \# 1.} We first establish a recursive structure between the basic objects like 
$S_{\aR,\bb}$ and $V_{\aR}$ between the $N$- and $(N+1)$-soliton solutions. This uses the breakdown of 
the Cauchy matrix as it occurs as the kernel $(\bun+\bM)^{-1}$ in the various objects.  First, for the sake 
of self-containedness, we repeat the following lemma: 

\begin{lemma}
The following identity holds for the inverse of a $(N+1)\times(N+1)$ block-matrix 
\be\label{eq:matinv} \left(\begin{array}{c|c} \bA & \bbb \\ \hline  
\bc^T & d\end{array} \right)^{-1} = \left(\begin{array}{c|c} \bA^{-1}(\bun+\frac{1}{s}\bbb\,\bc^T\bA^{-1}) & -\frac{1}{s}\bA^{-1}\bbb \\ 
\hline  -\frac{1}{s}\bc^T\bA^{-1} & \frac{1}{s}\end{array} \right)\  ,   \ee 
in which $\bA$ is an invertible $N\times N$ matrix $\bbb$ and $\bc^T$ are $N$-component vector column- and row-vector respectively , 
and $d$ is a nonzero scalar, where the scalar quantity $s$, given by
\be\label{eq:sinv} 
 s=d-\bc^T\bA^{-1}\bbb=\frac{1}{\det(\bA)}\,\left|\begin{array}{c|c} \bA & \bbb \\ \hline  
\bc^T & d\end{array} \right|\  , 
\ee 
is assumed to be nonzero.  
\end{lemma}
\begin{proof}
By direct multiplication, noting that the matrix is invertible if $s$ is nonzero. 
\end{proof}  

Let $\bM^{(N+1)}$ be the $(N+1)\times(N+1)$ Cauchy matrix with parameters $\kp_1,\dots,\kp_{N+1}$ as defined in a similar way as 
in \eqref{eq:Cauchy}. Applying the Lemma to compute the inverse of the matrix ~$\bun+\bM^{(N+1)}$~, which can be decomposed 
as above by setting
\begin{eqnarray*}
&& \bA=\bun+\bM^{(N)}\quad,\quad \bbb=c_{N+1}\Psi_\xi(\kp_{N+1})\chi_{\kp_{N+1},\bK}^{-1}\brr\quad ,  \\ 
&& \bc^T=\bs^T\chi_{\bK,\kp_{N+1}}^{-1}\Psi_\xi(\kp_{N+1})\rho_{N+1}\quad,\quad d=1+\rho_{N+1}c_{N+1}\Psi_\xi(2\kp_{N+1})\   ,  
\end{eqnarray*}  
we have 
\begin{eqnarray*}
&& (\bun+\bM^{(N+1)})^{-1}=    \\   
&& =\left(\begin{array}{c|c} 
(\bun+\bM^{(N)})^{-1}+s^{-1}c_{N+1}\rho_{N+1}\Psi_\xi^2(\kp_{N+1})\bu_{\kp_{N+1}}\,\tbu_{\kp_{N+1}} & 
-s^{-1}c_{N+1}\Psi_\xi(\kp_{N+1})\bu_{\kp_{N+1}} \\ 
\hline  -s^{-1}\rho_{N+1}\Psi_\xi(\kp_{N+1})\tbu_{\kp_{N+1}} & s^{-1}\end{array} \right) 
\end{eqnarray*} 
where $\bu(\cdot)$ and $\tbu(\cdot)$ are given in \eqref{eq:butbu}. Applying the explicit form of $\bun+\bM^{(N+1)})^{-1}$ to the definitions 
\eqref{eq:UV} to compute the elliptic $(N+1)$-soliton formulae $U^{(N+1)}_{\aR,\bb}$ and $V^{(N+1)}_\aR$ in terms of the $N$-soliton ones, we obtain 
the expressions
\bse\label{eq:UVs}\bea
U^{(N+1)}_{\aR,\bb}&=& U^{(N)}_{\aR,\bb}+\frac{1}{s}\rho_{N+1}c_{N+1}\Psi_\xi^2(\kp_{N+1}) 
\left( U^{(N)}_{\aR,\kp_{N+1}}-\chi_{\aR,\kp_{N+1}}^{-1}\right) \left( U^{(N)}_{\kp_{N+1},\bb}-\chi_{\kp_{N+1},\bb}^{-1}\right)\  , \\  
V^{(N+1)}_\aR &=& V^{(N)}_\aR +\frac{1}{s}\rho_{N+1}c_{N+1}\Psi_\xi^2(\kp_{N+1}) 
\left( U^{(N)}_{\aR,\kp_{N+1}}-\chi_{\aR,\kp_{N+1}}^{-1}\right) V^{(N)}_{\kp_{N+1}}\  ,
\eea\ese 
and subsequently using the definitions \eqref{eq:WS} we find the following recurrence relations between the $N$ and $N+1$ 
elliptic soliton objects:
\bse\label{eq:UVsolrecurs}\begin{eqnarray}
S^{(N+1)}_{\aR,\bb} &=& S^{(N)}_{\aR,\bb} - \frac{1}{s} c_{N+1}\rho_{N+1}\,S^{(N)}_{\aR,\kp_{N+1}} S^{(N)}_{\kp_{N+1},\bb}\   ,  \label{eq:UVsola} \\
W^{(N+1)}_{\aR} &=& W^{(N)}_{\aR}-\frac{1}{s}c_{N+1}\rho_{N+1}\,S^{(N)}_{\aR,\kp_{N+1}} W^{(N)}_{\kp_{N+1}}\   , \label{eq:UVsolb}  
\end{eqnarray}\ese 
where we exclude values for $\aR$ and $\bb$ such that $\aR+\kp_{N+1}=0$ or $\bb+\kp_{N+1}=0$ (modulo the period lattice of the Weierstrass 
elliptic functions). Furthermore, the quantity $s$ can be directly computed using \eqref{eq:sinv} with the identifications of the entries given earlier
and using the definitions \eqref{eq:UV} and \eqref{eq:WS} again, leading to: 
\be\label{eq:s}
 s=1+\rho_{N+1}c_{N+1}S^{(N)}_{\kp_{N+1},\kp_{N+1}}=\frac{f^{(N+1)}}{f^{(N)}} \  . \ee
Here the latter equality follows from the second equality in \eqref{eq:sinv}, which expresses $s$ the ratio of  
the $\tau$-function ~$f^{(N+1)}=\det(\bun+\bM^{(N+1)})$~ of the elliptic $(N+1)$-soliton solution and the $\tau$-function 
~$f^{(N)}=\det(\bun+\bM^{(N)})$~ of the embedded elliptic $N$-soliton solution. 
Setting $\bb=\kp_{N+1}$ in \eqref{eq:UVsola} and using \eqref{eq:s} we subsequently obtain the following identification for the 
quantity $s$
\begin{equation}\label{eq:ss} 
s=\frac{S^{(N)}_{\aR,\kp_{N+1}}}{S^{(N+1)}_{\aR,\kp_{N+1}}}=\frac{W^{(N)}_{k_{N+1}}}{W^{(N+1)}_{k_{N+1}}}\   , 
\end{equation} 
which in particular implies that the ratio of $S_{\aR,\kp_{N+1}}$ between its $N$- and $(N+1)$-elliptic soliton value, is independent 
of the parameter $\aR$.  

\paragraph{}{\it Step \# 2.}  Next we apply the BT with the (so far arbitrary) B\"acklund parameter $\ssl=(l,L)\in\Gamma_{a,b}$ to the 
elliptic $N$-soliton solution defined in \eqref{eq:Nsol}, i.e. we want to solve the system of discrete Riccati equations 
for a new variable $v$ 
\be\label{eq:Q3BT}
\cQ^\Delta_{\ssp,\ssl}(u^{(N)},\wt{u}^{(N)},v,\wt{v})=0\quad, \quad 
 \cQ^\Delta_{\ssq,\ssl}(u^{(N)},\wh{u}^{(N)},v,\wh{v})=0\  ,  
\ee 
relying on the multidimensional consistency of Q3, and the fact that the shifts on the lattice can be interpreted as B\"acklund 
transformations. To solve the system \eqref{eq:Q3BT} we reduce the problem by identifying two particular solutions, which are 
obtained by the forward- and backward shifts, in the covariant extension, of the known elliptic $N$-soliton solutions $u^{(N)}$. 
This implies we need to extend all the plane-wave factors by including a lattice direction associated with the parameter $\ssl$, i.e. we set 
\bse\label{eq:covext}  
\be
 \rho_{0,0}(\kp)=\left( e^{-2\zeta(\ld)\kp} \frac{\ell_{-\kp}}{\ell_\kp}\right)^h\rho_{0,0,0}\quad\Rightarrow\quad 
\ol{\rho}(\kp)=e^{2\eta_\ld\kp}\frac{\ell_{-\kp}}{\ell_{\kp}}\,\rho(\kp)\quad,\quad \underline{\rho}(\kp)=e^{-2\eta_\ld\kp} 
\frac{\ell_{\kp}}{\ell_{-\kp}}\,\rho_i(\kp)\    , \ee
in all relevant definitions, where ~$\ell_\kp=\Psi_\ld(\kp)$ ($\ld$ being the uniformising variable associated with the point $\ssl$ on the 
curve $\Gamma_{a,b}$ and on the Weierstrass curve).  Similarly we have 
\be\label{eq:digextend}
\ol{\rho}_i=e^{2\eta_\ld\kp_i}\frac{\ell_{-\kp_i}}{\ell_{\kp_i}}\,\rho_i\quad,\quad \underline{\rho}_i=e^{-2\eta_\ld\kp_i}\frac{\ell_{\kp_i}}{\ell_{-\kp_i}}\,\rho_i\   .    
\ee \ese 
Here the elementary discrete shift in the discrete variable $h$ is indicated by a $\ol{\phantom{a}}$, in terms of which the elliptic $N$-soliton 
solutions depending on three variables, $u^{(N)}_{n,m,h}$, we have 
$$ 
\ol{u}^{(N)}_{n,m,h}=u^{(N)}_{n,m,h+1}\quad,\quad \underline{u}^{(N)}_{n,m,h}=u^{(N)}_{n,m,h-1}\   . 
$$ 
It is an immediate consequence of the construction behind Theorem \ref{th:theorem1}, that the following equations are satisfied:
\bse\label{eq:BT0sol}\begin{eqnarray}\label{eq:BT0sola}
\cQ^\Delta_{\ssp,\ssl}(u^{(N)},\wt{u}^{(N)},\ol{u}^{(N)},\wt{\ol{u}}^{(N)})=0\quad&,&\quad 
\cQ^\Delta_{\ssq,\ssl}(u^{(N)},\wh{u}^{(N)},\ol{u}^{(N)},\wh{\ol{u}}^{(N)})=0\   , \\  
\cQ^\Delta_{\ssp,\ssl}(u^{(N)},\wt{u}^{(N)},\underline{u}^{(N)},\wt{\underline{u}}^{(N)})=0  
\quad&,&\quad  
\cQ^\Delta_{\ssq,\ssl}(u^{(N)},\wh{u}^{(N)},\underline{u}^{(N)},\wh{\underline{u}}^{(N)})=0\   , \label{eq:BT0solb} 
\end{eqnarray}\ese    
where \eqref{eq:BT0solb} holds because of the symmetry of the quadrilaterals. 
Having established the two solutions $\ol{u}^{(N)}$ and $\underline{u}^{(N)}$ of \eqref{eq:Q3BT} we can now find the 
general solution of that system in the interpolating form:  
\be\label{eq:interpol}
 v=\frac{\ol{u}^{(N)}+\veta\underline{u}^{(N)}}{1+\veta}\   , \ee  
in which $\veta$ is some function to be determined from the following coupled system of homogeneous linear equations:  
\bse\label{eq:etalineqs}\bea
\frac{\wt{\veta}}{\veta} &=& -\frac{\cQ^\Delta_{\ssp,\ssl}(u^{(N)},\wt{u}^{(N)},\underline{u}^{(N)},\wt{\ol{u}}^{(N)})}
{\cQ^\Delta_{\ssp,\ssl}(u^{(N)},\wt{u}^{(N)},\ol{u}^{(N)},\wt{\underline{u}}^{(N)})}\   , \\ 
\frac{\wh{\veta}}{\veta} &=& -\frac{\cQ^\Delta_{\ssq,\ssl}(u^{(N)},\wh{u}^{(N)},\underline{u}^{(N)},\wh{\ol{u}}^{(N)})}
{\cQ^\Delta_{\ssq,\ssl}(u^{(N)},\wh{u}^{(N)},\ol{u}^{(N)},\wh{\underline{u}}^{(N)})}\   . 
\eea\ese 
The compatibility of this system is equivalent to the compatibility of the BT \eqref{eq:Q3BT} as a coupled system of discrete Riccati 
equations, which in turn is a consequence of the multidimensional consistency of the lattice equation.  

Using the explicit expressions for the quadrilaterals \eqref{eq:Q3quad}, eqs. \eqref{eq:etalineqs} reduce to
\bse\label{eq:etalineqs2}\bea
\frac{\wt{\veta}}{\veta} &=& \frac{P\wt{u}^{(N)}-(p^2-l^2)u^{(N)}-L\underline{u}^{(N)}}{P\wt{u}^{(N)}-(p^2-l^2)u^{(N)}-L\ol{u}^{(N)}}
=\frac{\underline{w}-\wt{w}}{\ol{w}-\wt{w}}\   , \\ 
\frac{\wh{\veta}}{\veta} &=& \frac{Q\wh{u}^{(N)}-(q^2-l^2)u^{(N)}-L\underline{u}^{(N)}}{Q\wh{u}^{(N)}-(q^2-l^2)u^{(N)}-L\ol{u}^{(N)}}
=\frac{\underline{w}-\wh{w}}{\ol{w}-\wh{w}}\   , 
\eea\ese 
where in the last step we have made use of eqs. \eqref{eq:wuUMiura}, replacing $\ven$, respectively $\dd$, by $\ld$. 
Finally, using eqs. \eqref{eq:wWMiura} with the same replacements, leading to 
\be\label{eq:wwWW} 
 \ol{w}-\wt{w}=\frac{p_\aR e^{-\ol{\eta}_\dd\aR}\ol{W}_\aR-\ell_\aR e^{-\wt{\eta}_\ld\aR}\wt{W}_\aR}{\wt{\ol{W}}_\aR}\quad, \quad 
\underline{w}-\wt{w}=\frac{p_{-\aR}e^{\eta_\dd\aR}\wt{W}_\aR+\ell_\aR e^{-\underline{\eta}_\ld\aR} \underline{W}_\aR}{W_\aR}\  , \ee  
which by setting subsequently ~$\aR=-\ld~\Rightarrow~\ell_\aR=0$~ yields the expressions 
$$  \frac{\wt{\veta}}{\veta}= e^{-(\eta_\dd+\ol{\eta}_\dd)\ld} \left(\frac{p_\ld}{p_{-\ld}}\right)\,
\frac{\wt{W}_{-\ld}\wt{\ol{W}}_{-\ld}}{W_{-\ld}\ol{W}_{-\ld}}\quad, \quad 
\frac{\wh{\veta}}{\veta}= e^{-(\eta_\ven+\ol{\eta}_\ven)\ld}\left(\frac{q_\ld}{q_{-\ld}}\right)\,
\frac{\wh{W}_{-\ld}\wh{\ol{W}}_{-\ld}}{W_{-\ld}\ol{W}_{-\ld}}\   . $$ 
These relations can be simultaneously integrated, yielding the following expression for the function $\veta$:
\be\label{eq:eta} 
\veta=\veta_{n,m}=\veta_0\,e^{-[\zeta(\xi)+\zeta(\ol{\xi})]\ld} \left( e^{2\zeta(\dd)\ld} \frac{p_\ld}{p_{-\ld}}\right)^n 
\left( e^{2\zeta(\ven)\ld} \frac{q_\ld}{q_{-\ld}}\right)^m\,W_{-\ld}\ol{W}_{-\ld}\  . 
\ee 
Using the relation 
$$ \ol{W}_\aR\,W_\bb=\ell_\aR e^{-\eta_\ld\aR} S_{\aR,\bb}-\ell_{-\bb} e^{\eta_\ld\bb} \ol{S}_{\aR,\bb}
\quad\Rightarrow\quad \ol{W}_\ld W_\ld=\ell_\ld e^{-\eta_\ld\ld} S_{\ld,\ld}\   , $$ 
obtained by replacing $\dd$ by $\ld$ and $\wt{\phantom{a}}$ by $\ol{\phantom{a}}$ in \eqref{eq:WW}, and using also \eqref{eq:WWinv} with 
the same replacements, we can obtain from \eqref{eq:eta} the following result:  
\be\label{eq:feta}
\frac{1}{\ol{\veta}}=\frac{\sg^2(\ld)\,\sg(2\ld)}{\ol{\veta}_0}\frac{\rho(\ld)}{\rho_{0,0}(\ld)}\,S_{\ld,\ld}\  .  
\ee 

\paragraph{Step \# 3.} We now compare the expressions for the solution obtained from 
the recurrence of step \# 1 and the one from the BT of step \# 2. On the one hand, the solution from the BT \eqref{eq:interpol} 
can be recast into the following form 
\be\label{eq:vsol}
v=\underline{u}^{(N)}+\frac{1}{1+\veta}(\ol{u}^{(N)}-\underline{u}^{(N)})\quad\Rightarrow\quad \ol{v}=
u^{(N)}-\frac{1}{1+\ol{\veta}}\,\frac{1}{L}(w-\ol{\ol{w}})\ol{U}^{(N)}\   ,  
\ee 
where we have used the relation ~$w-\ol{\ol{w}}=L(u^{(N)}-\ol{\ol{u}}^{(N)})/\ol{U}^{(N)}$~, which follows from \eqref{eq:wuUMiurab}, 
setting $\dd=\ven=\ld$. Setting now in \eqref{eq:feta} the B\"acklund parameter $\ld=\kp_{N+1}$, i.e. $\ssl=\ssk_{N+1}=(k_{N+1},K_{N+1})$, which 
implies that the $\ol{\phantom{a}}$-shift from now on is the lattice shift associated with the lattice parameter $\ssk_{N+1}$, and we choose the integration 
constant $\veta_0$, such that we identify \eqref{eq:feta} as follows:  
\be\label{eq:etaident}
\frac{1}{\ol{\veta}}=c_{N+1}\rho_{N+1}S^{(N)}_{\kp_{N+1},\kp_{N+1}}\quad\Rightarrow\quad s=1+\frac{1}{\ol{\veta}}\   , 
\ee 
comparing the result with \eqref{eq:s}. This identification requires the choice
\be\label{eq:t0id}
\ol{\veta}_0=\sg(2\kp_{N+1})\sg^2(\kp_{N+1})/c_{N+1}\rho_{0.0}(\kp_{N+1})\  . \ee   

On the other hand, 
inserting the recursion relation \eqref{eq:UVsola} into the expression for the elliptic $(N+1)$-soliton solution we find   
\begin{eqnarray*}
u^{(N+1)} &=& A\rho^{1/2}(\aR)\,\rho^{1/2}(\bb) S^{(N+1)}_{\aR,\bb} +\cdots  \\ 
&=& A\rho^{1/2}(\aR)\,\rho^{1/2}(\bb)\left( S^{(N)}_{\aR,\bb} -\frac{1}{s}c_{N+1}\rho_{N+1} S^{(N)}_{\aR,\kp_{N+1}}\,S^{(N)}_{\kp_{N+1},\bb}\right) + \cdots \\
&=& u^{(N)} - \frac{1}{s}c_{N+1}\rho_{N+1}\left[ A \rho^{1/2}(\aR)\,\rho^{1/2}(\bb) S^{(N)}_{\aR,\kp_{N+1}}\,S^{(N)}_{\kp_{N+1},\bb} +\cdots \right]\  ,  \\ 
&=& u^{(N)} - \frac{1}{s}c_{N+1}\rho_{N+1}\left[ A \rho^{1/2}(\aR)\,\rho^{1/2}(\bb) e^{\eta_{\kp_{N+1}}(\aR+\bb)}  
\frac{\ol{W}^{(N)}_{\aR}\ol{W}^{(N)}_{\bb}(W^{(N)}_{\kp_{N+1}})^2}{\Psi_{\kp_{N+1}}(\aR)\,\Psi_{\kp_{N+1}}(\bb)} +\cdots \right]\  ,
\end{eqnarray*}  
where the dots stand for the remaining terms  with coefficients $B$, $C$, 
$D$ and with $(\aR,\bb)$ replaced by $(\aR,-\bb)$, $(-\aR,\bb)$, $(-\aR,-\bb)$ respectively, and where in the last step we have made use of the relations 
$$  \ol{W}_\aR W_\ld=\ell_\aR e^{-\eta_\ld\aR} S_{\aR,\ld} \quad,\quad \ol{W}_\bb W_\ld=\ell_\bb e^{-\eta_\ld\bb} S_{\bb,\ld}\  ,  $$
which follows from the earlier relation setting $\bb=\ld$, and the same with $\aR$ and $\bb$ interchanged, and subsequently taking $\ld=\kp_{N+1}$. 
Taking into account the covariant extension of the plane-wave factors $\rho(\aR)$, $\rho(\bb)$, applying a relation of the form of \eqref{eq:rhoshifts}, we can 
conclude from this computation that 
\be \label{eq:usol}
u^{(N+1)}=u^{(N)} -\frac{c_{N+1}\rho_{N+1}}{sK_{N+1}}(W^{(N)}_{\kp_{N+1}})^2 \ol{U}^{(N)} \   , 
\ee 
using the expression \eqref{eq:assNsol} and identifying the factor $L$ with $K_{N+1}$ through the relation  
$L=(\ell_{-\aR}\ell_\aR\ell_{-\bb}\ell_\bb)^{1/2}$ for $\ld=\kp_{N+1}$. 

Finally, noting that 
\be
 \underline{w}-\ol{w}=\frac{\ell_\aR e^{-\underline{\eta}_\ld\aR} \underline{W}_\aR+\ell_{-\aR} e^{\eta_\dd\aR} \ol{W}_\aR}{W_\aR}=
\ell_\ld e^{-\underline{\eta}_\ld\ld} \frac{\underline{W}_\ld}{W_\ld}\quad\Rightarrow\quad w-\ol{\ol{w}}=\frac{W_\ld^2}{S_{\ld,\ld}}\  ,  
\ee
we can compare the expressions for $\ol{v}$ in \eqref{eq:vsol} and $u^{(N+1})$ in \eqref{eq:usol} and conclude that they are the same provided 
that 
$$ \frac{1}{1+\ol{t}}=\frac{1}{s}c_{N+1}\rho_{N+1} S_{\kp_{N+1},\kp_{N+1}}\   , $$ 
and which holds as a consequence of \eqref{eq:etaident}. 
%\be\label{eq:etas} 
%1+\ol{\veta}= e^{-[\zeta(\ol{\xi})-\zeta(\xi)]\ld} \frac{\sg(2\kp_{N+1})}{\sg^2(\kp_{N+1})}\left(c_{N+1}\rho_{N+1}\ol{W}_{\kp_{N+1}}W_{\kp_{N+1}}\right)^{-1} s\   , 
%\ee  
%which is consistent with \eqref{eq:s} and \eqref{eq:eta} setting the integration constant ~$\veta_0=\sg^2(\kp_{N+1})/c_{N+1}\rho_{0.0}(\kp_{N+1})$~. 
Thus, we have established that with the above choice \eqref{eq:t0id} of the integration constant $\veta_0$ in \eqref{eq:eta}  
the B\"acklund transformed solution coincides with the corresponding elliptic $(N+1)$-soliton solution given by the Cauchy-matrix objects. 
This then  completes the proof of Theorem \ref{th:theorem2}. 
\DONE 

Theorem \ref{th:theorem2} establishes the precise connection between the structure of the $N$-soliton solution as given by the 
Cauchy matrix approach, and the way to generate a soliton hierarchy through B\"acklund transforms. What we conclude 
is that these two approaches coincide up to a subtle identification of the relevant constants in the solution. 
Since, as was remarked in the Corollary of section 3 that the soliton solutions of Q3 really live in an extended 
four-dimensional lattice, the precise identification of those constants is of interest, since they contain possibly 
the additional lattice directions. In fact, in \cite{AHN,AN} we established the first soliton type solutions for Q4 
through the B\"acklund approach, and it is of interest to see how that approach connects to a (yet unknown) 
representation of multi-soliton solutions in terms of a scheme similar to the one set up in this paper for Q3.   

\section{Elliptic soliton solutions for degenerate subcases of Q3} 

In the same way as in the rational soliton case of \cite{NAH}, one can obtain from the elliptic $N$-soliton solution for Q3 
the corresponding elliptic $N$-soliton solutions for all "lower" ABS equations through limits on the parameters $\aR$, $\bb$ 
(adjusting, where necessary, the coefficients $A$,$B$,$C$,$D$ 
in an appropriate manner). This follows the coalescence scheme of Figure \ref{Hdegen}.
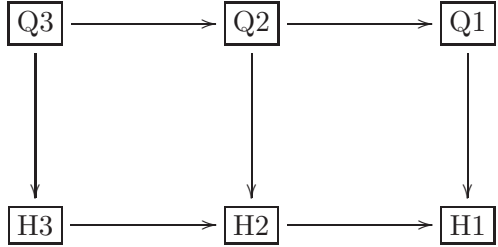
\begin{figure}[h]
\xymatrix{
&&\framebox{Q3} \ar[dd] \ar[rr] && \framebox{Q2} \ar[dd] \ar[rr] && \framebox{Q1} \ar[dd] \\ \\
&&\framebox{H3} \ar[rr] && \framebox{H2} \ar[rr] && \framebox{H1}
}
\caption{Coalescence scheme employed to construct elliptic $N$-soliton-type solutions for the degenerate sub cases of equation Q3.}
\label{Hdegen}
\end{figure}
The upper horizontal sequence in this scheme, involving the degenerations of the Q-equations, are obtained 
from performing careful limits of the type $\bb\rightarrow\aR$, whilst the vertical limit from Q- to H-equations is 
obtained from the limits $\aR~{\rm or}~\bb\rightarrow 0$. We observe that this degeneration scheme roughly follows the pattern of the limits of 
the NQC equation that were studied in \cite{NQC} leading to the various degenerations of KdV type lattice equations. 

\subsection{Q-degenerations} 

We first make the following observation which can be independently validated: 
\begin{lemma}\label{th:lemma4}
The quantity 
\be\label{eq:uaa}
z=\digamma_{\aR}+\left(\wp(\aR)-\wp(\xi)\right) U_{\aR,-\aR}\  , 
\ee
where $U_{\aR,-\aR}$ is defined as before, and where we have introduced
\be\label{eq:digamma}
\digamma_\aR=\digamma_{n,m}(\xi;\aR)\Equiv \zeta(\xi)+ \frac{n}{2}\left(\zeta(\aR-\dd)-\zeta(\aR+\dd)\right)+ 
\frac{m}{2}\left(\zeta(\aR-\ven)-\zeta(\aR+\ven)\right)\  , 
\ee
obeys the Q1 equation in the following form:
\be\label{eq:zeq}
(p^2-a^2)(z-\wt{z})(\wh{z}-\wh{\wt{z}})-(q^2-a^2)(z-\wh{z})(\wt{z}-\wh{\wt{z}})=
\frac{1}{4}(\wp'(\aR))^2\left(\frac{1}{p^2-a^2}-\frac{1}{q^2-a^2} \right)\   . 
\ee
\end{lemma} 
\begin{proof}
To prove this lemma, let us consider the quantity $z$ as defined in \eqref{eq:uaa}, with $\digamma_\aR$ given in \eqref{eq:digamma}. 
From the definition we have, using \eqref{eq:WWs},  
\begin{eqnarray*}
\wt{z}-z &=& \zeta(\xi+\dd)-\zeta(\xi)+\frac{1}{2}\left( \zeta(\aR-\dd)-\zeta(\aR+\dd)\right) 
+\left( \wp(\aR)-\wp(\wt{\xi})\right) \wt{U}_{\aR,-\aR}-\left( \wp(\aR)-\wp(\xi)\right) U_{\aR,-\aR}  \\ 
&=& - \frac{1}{p_{\aR}}e^{\eta_\dd\aR}\,\wt{W}_{\aR}W_{-\aR}+\frac{1}{2}\left(\zeta(\aR+\dd)+\zeta(\aR-\dd)-2\zeta(\aR)\right) \\ 
&=& - \frac{1}{p_{\aR}}e^{\eta_\dd\aR}\,\wt{W}_{\aR}W_{-\aR}-\frac{1}{2} \frac{\wp'(\aR)}{\wp(\dd)-\wp(\aR)}\   , 
\end{eqnarray*}
where the latter form arises by using \eqref{eq:zetawp}. Since $z$ remains unchanged under flipping the sign of $\aR$, because 
$\digamma_{\aR}=\digamma_{-\aR}$ and $U_{\aR,\bb}$ is symmetric with respect to the interchange of $\aR$ and $\bb$, we also 
have
$$ \wt{z}-z= - \frac{1}{p_{-\aR}}e^{-\eta_\dd\aR}\,\wt{W}_{-\aR}W_{\aR}+\frac{1}{2} \frac{\wp'(\aR)}{\wp(\dd)-\wp(\aR)}\  .  $$
Hence, we have 
$$ (p^2-a^2)\left(\wt{z}-z+\frac{1}{2}\frac{\wp'(\aR)}{\wp(\dd)-\wp(\aR)}\right)\left(\wh{\wt{z}}-\wh{z}-\frac{1}{2}\frac{\wp'(\aR)}{\wp(\dd)-\wp(\aR)}\right)
=(p^2-a^2)\frac{e^{\eta_\dd\aR}e^{-\wh{\eta}_\dd\aR}}{p_{\aR}p_{-\aR}}\,\wt{W}_{\aR} W_{-\aR} \wh{\wt{W}}_{-\aR}\wh{W}_{\aR}\   , $$ 
where the expression on the right-hand side is invariant under interchanges ~$\dd\leftrightarrow\ven$~,~$\wh{\phantom{a}}\leftrightarrow\wt{\phantom{a}}$~, 
taking into account that ~$p_{\aR}p_{-\aR}=\wp(\dd)-\wp(\aR)=p^2-a^2$~. and that ~$\eta_\dd-\wh{\eta}_\dd=\eta_\ven-\wt{\eta}_\ven$~. Thus, we have 
that 
\begin{eqnarray*}
 && (p^2-a^2)\left(\wt{z}-z+\frac{1}{2}\frac{\wp'(\aR)}{\wp(\dd)-\wp(\aR)}\right)\left(\wh{\wt{z}}-\wh{z}-\frac{1}{2}\frac{\wp'(\aR)}{\wp(\dd)-\wp(\aR)}\right) = \\ 
 && \quad = (q^2-a^2)\left(\wh{z}-z+\frac{1}{2}\frac{\wp'(\aR)}{\wp(\ven)-\wp(\aR)}\right)\left(\wh{\wt{z}}-\wt{z}-\frac{1}{2}\frac{\wp'(\aR)}{\wp(\ven)-\wp(\aR)}\right)
\end{eqnarray*}
which after some cancellations leads to 
$$ (p^2-a^2)(\wt{z}-z)(\wh{\wt{z}}-\wh{z})-(q^2-a^2)(\wh{z}-z)(\wh{\wt{z}}-\wt{z})=\frac{1}{4}(\wp'(\aR))^2\left( \frac{1}{p^2-a^2}-\frac{1}{q^2-a^2}\right)\  , $$ 
which is eq. \eqref{eq:zeq}. 
\end{proof}

The identification with the equation Q1 \eqref{eq:Qeqsa} can be made by identifying the lattice parameters  
\be\label{eq:Q2idens}
\po=\frac{1}{2} \frac{\wp'(\aR)}{p^2-a^2}\quad,\quad \qo=\frac{1}{2} \frac{\wp'(\aR)}{q^2-a^2}\  ,    
\ee 
making it evident that $z$ defined in \eqref{eq:uaa} provides a solution of Q1 in that form with the fixed parameter $\dd^2=1$. 
The elliptic soliton solutions for the more general parameter case of Q1 are obtained as an immediate corollary of Theorem \ref{th:theorem1}, 
by considering the limit $\bb\rightarrow \aR$. This leads to the following result: 
\begin{theorem}\label{th:cor1} The following formula 
\begin{equation}\label{eq:Q1Nsol} 
u^{(N)} = A\rho(\aR)S_{\aR,\aR}+D\rho(-\aR)S_{-\aR,-\aR} + 2Bz+ C \nn 
\end{equation} 
% \begin{eqnarray} \label{eq:Q1Nsol} 
% && u^{(N)} = A\rho(\aR)S_{\aR,\aR}+D\rho(-\aR)S_{-\aR,-\aR} + C \nn \\
% && +B\left\{ 2\zeta(\xi)+n\left[(\zeta(\aR-\dd)-\zeta(\aR+\dd)\right]+ m\left[\zeta(\aR-\ven)-\zeta(\aR+\ven)\right]
% +2[\wp(\aR)-\wp(\xi)]U_{\aR,-\aR} \right\} \  , \nn \\ 
% \end{eqnarray}
in which $S_{\pm\aR,\pm\aR}$, $\rho(\pm\aR)$ are defined as before, where $z$ is given as in Lemma \ref{th:lemma4},  
and where the coefficients are subject to the relation 
$$ AD-B^2=\frac{\Delta}{\wp'(\aR)^2} \  , $$ 
constitutes the elliptic $N$-soliton solution of the $({\rm Q}_1)_\Delta$ equation, given in the form
\be\label{eq:Q1Delta}
(p^2-a^2)(u-\wt{u})(\wh{u}-\wh{\wt{u}})-(q^2-a^2)(u-\wh{u})(\wt{u}-\wh{\wt{u}})+\Delta\left(\frac{1}{p^2-a^2}-\frac{1}{q^2-a^2}\right)=0\  . 
\ee 
Hence, \eqref{eq:Q1Nsol} leads to a solution of the equation Q1 in the form \eqref{eq:Qeqsa} using the identifications of lattice parameters 
given by \eqref{eq:Q2idens}, for generic value of the fixed parameter. 
\end{theorem}
\begin{proof}
Taking the limit $\bb=\aR+\gm$, and letting $\gm\rightarrow 0$ in \eqref{eq:Nsol} the terms with coefficients $A$ and $D$ are regular under 
this limit, whilst the terms with coefficients $B$ and $C$ possess a singularity due to the prefactor $\Psi_\xi(\pm\aR\mp\bb)$ in the quantities 
$S_{\pm\aR,\mp\bb}$. Thus we have the limits 
\begin{eqnarray*}
% \rho(\bb) &\rightarrow& \rho(\aR)\left\{ 1+ 2\gm\digamma_\aR + \gm^2 \left[ 2\digamma_\aR^2 
%+ \frac{n}{2}\left(\wp(\aR+\dd)-\wp(\aR-\dd)\right)+ \frac{m}{2}\left(\wp(\aR+\ven)-\wp(\aR-\ven)\right)\right] 
%+\mathcal{O}(\gm^3)\right\}  \\ 
 \rho^{1/2}(\pm\bb) &\rightarrow & \rho^{1/2}(\pm\aR)\left\{ \phantom{\frac{a}{b}} 1 \pm \gm\digamma_\aR + \right. \\ 
%\gm \left[ \zeta(\xi)+ \frac{n}{2}\left(\zeta(\aR-\dd)-\zeta(\aR+\dd)\right)+ \frac{m}{2}\left(\zeta(\aR-\ven)-\zeta(\aR+\ven)\right)\right]+\cdots \right.   \\ 
&& \quad \left. +\frac{1}{2}\gm^2 \left[ \digamma_\aR^2 
%\left( \zeta(\xi)+ \frac{n}{2}\left(\zeta(\aR-\dd)-\zeta(\aR+\dd)\right)+ \frac{m}{2}\left(\zeta(\aR-\ven)-\zeta(\aR+\ven)\right) \right)^2 
\pm \frac{n}{2}\left(\wp(\aR+\dd)-\wp(\aR-\dd)\right)\pm \frac{m}{2}\left(\wp(\aR+\ven)-\wp(\aR-\ven)\right)  \right]
+\mathcal{O}(\gm^3)\right\}  \\ 
 \Psi_\xi(\pm\aR\mp\bb) &\rightarrow& \mp\frac{1}{\gm} \left[ 1-\frac{1}{2}\gm^2\wp(\xi)+\mathcal{O}(\gm^3) \right] \\ 
 \Psi_\xi(\pm\aR\pm\bb) &\rightarrow& \Psi_\xi(\pm 2\aR)\left[ 1\pm\gm\eta_{\pm 2\aR}+\frac{1}{2}\gm^2\left(\wp(2\aR)-\wp(\xi\pm 2\aR)+\eta_{\pm 2\aR}^2 \right) + \mathcal{O}(\gm^3) \right] \\
\chi_{\pm\aR,\mp\bb} &\rightarrow& \pm\gm\left(\wp(\aR)-\wp(\xi) \right) + \frac{1}{2}\gm^2\left(\wp'(\pm\aR)+\wp'(\xi) \right)
+\mathcal{O}(\gm^3)\  , \\ 
\chi_{\pm\aR,\pm\bb} &\rightarrow& \chi_{\pm\aR,\pm\aR}+ %\mp\frac{1}{2}\frac{\wp''(\aR)}{\wp'(\aR)}-\eta_{\pm 2\aR} +
\gm\left(\wp(\xi\pm 2\aR)-\wp(\aR) \right) + \frac{1}{2}\left(\wp'(\xi\pm 2\aR)\mp\wp'(\aR) \right)+ \mathcal{O}(\gm^3)\  , 
\end{eqnarray*}
%where we have introduced
%\be\label{eq:digamma}
%\digamma_\aR=\digamma_{n,m}(\xi;\aR)\Equiv \zeta(\xi)+ \frac{n}{2}\left(\zeta(\aR-\dd)-\zeta(\aR+\dd)\right)+ \frac{m}{2}\left(\zeta(\aR-\ven)-\zeta(\aR+\ven)\right)
%\ee
and introducing
\be\label{eq:T}
T_{\aR,\bb}=\bs^T \chi^{-1}_{\aR,\bK} (\bun+\bM)^{-1} \chi_{\bb,\bK}^{-2}\left(\wp(\bb)-\wp(\xi+\bb+\bK)\right)\,\brr\  . 
%\quad {\rm with} 
%\quad \eta_{\bK}=\zeta(\xi+\bK)-\zeta(\xi)-\zeta(\bK)\   . 
\ee 
we have also
\begin{eqnarray*}
U_{\aR,-\bb} &\rightarrow& U_{\aR,-\aR} - \gm T_{\aR,-\aR} + \mathcal{O}(\gm^2)\   , \\ 
\Rightarrow\quad S_{\aR,-\bb} &\rightarrow& -\frac{1}{\gm}+\left(\wp(\aR)-\wp(\xi) \right)U_{\aR,-\aR} 
+\gm \left[\frac{1}{2}\wp(\xi)-\left(\wp(\aR)-\wp(\xi)\left(\eta_{-\aR}\,U_{\aR,-\aR}+T_{\aR,-\aR}\right) \right)   \right]\   , 
\end{eqnarray*}
and 
\begin{eqnarray*}
U_{\aR,\bb} &\rightarrow& U_{\aR,\aR} + \gm T_{\aR,\aR} + \mathcal{O}(\gm^2)\   , \\ 
\Rightarrow\quad S_{\aR,\bb} &\rightarrow& S_{\aR,\aR}+\gm\Psi_\xi(2\aR)\left[ \eta_{2\aR}\left(1-\chi_{\aR,\aR}U_{\aR,\aR}\right) 
-\left(\wp(\xi+2\aR)-\wp(\aR)\right)U_{\aR,\aR}-\chi_{\aR,\aR}T_{\aR,\aR}\right] \   . 
\end{eqnarray*}
Assuming the coefficients $B$ and $C$ to behave as:
$$ B~\rightarrow~B_0+B_1\gm+\cdots\quad,\quad C~\rightarrow~C_0+C_1\gm+\cdots \  , $$ 
where the singularity can be avoided if we choose $C_0=B_0$.     
\end{proof}

We now consider the limit to the case of the equation Q2. An extension of Lemma \ref{th:lemma4} is the following statement:  
\begin{lemma}\label{th:lemma5}
The following quantity obtained from the expansion of the Q3 elliptic soliton formula, namely 
\begin{eqnarray}\label{eq:zpm}
z'_\pm&=&\frac{1}{2}\left[\wp(\xi)-\digamma_{\aR}^2\pm\frac{n}{2}\left(\wp(\aR+\dd)-\wp(\aR-\dd)\right)+\frac{m}{2}\left(\wp(\aR+\ven)
-\wp(\aR-\ven)\right)\right] \nn \\ 
&& -\left(\wp(\aR)-\wp(\xi)\right)\left[\left(\eta_{\mp\aR}+\digamma_\aR\right)\,U_{\aR,-\aR}+T_{\pm\aR,\mp\aR}\right] 
\end{eqnarray}
obeys the following relations
\begin{eqnarray}\label{eq:Q2rels}
&& (p^2-a^2)\left[(z-\wt{z})(\wh{z}'_{\pm}-\wh{\wt{z}}'_{\pm})+(z'_{\pm}-\wt{z}'_{\pm})(\wh{z}-\wh{\wt{z}})\right] \nn \\ 
&& -(q^2-a^2)\left[(z-\wh{z})(\wt{z}'_{\pm}-\wh{\wt{z}}'_{\pm})+(z'_{\pm}-\wh{z}'_{\pm})(\wt{z}-\wh{\wt{z}})\right]= \nn \\ 
&& = \pm\frac{1}{2}\wp'(\aR) (z-\wh{\wt{z}})(\wh{z}-\wt{z}) + \frac{1}{4}\left(\frac{1}{p^2-a^2}-\frac{1}{q^2-a^2}\right)\bigg[ \pm\wp'(\aR)\wp''(\aR) \nn \\ 
&& \qquad -(\wp'(\aR))^2(z+\wt{z}+\wh{z}+\wh{\wt{z}}) \pm \frac{1}{2}(\wp'(\aR))^3\left(\frac{1}{p^2-a^2}+\frac{1}{q^2-a^2}\right) \bigg] \  . 
\end{eqnarray}
\end{lemma}
\begin{proof}
The proof is by direct expansion of the Q3 solution using the expansions ~$u^{(N)}=-\frac{1}{\gm}+z+\gm z_+$~ in the case that we take the 
coefficients $A=D=C=0$, or the expansion ~$u^{(N)}=\frac{1}{\gm}+z-\gm z_-$~ in the case we take the coefficients $A=D=B=0$, and using the expansion 
of the proof of Theorem \ref{th:cor1} to obtain the explicit form of the quantities $z_\pm$.  
\end{proof}

A combination of the quantities $z_\pm$ of Lemma \ref{th:lemma5} yields the core of the general elliptic soliton solution of the equation Q2, as is 
expressed in the following statement. 

\begin{theorem}\label{th:cor2}
The following function
\begin{eqnarray}\label{eq:Q2Nsol}
u^{(N)} &=& \left(\wp(\aR)-\wp(\xi)\right)(T_{\aR,-\aR}+T_{-\aR,\aR})-\wp'(\xi)U_{\aR,-\aR} -\wp(\xi)
-\frac{1}{6}\left(\frac{\wp'''(\aR)}{\wp'(\aR)}-\frac{3}{2} \frac{(\wp''(\aR))^2}{(\wp'(\aR))^2} \right) \nn \\ 
&& + \left( \frac{B_1+C_1}{2B_0}-\digamma_{\aR}\right) \left[ \frac{B_1+C_1}{2B_0}-\digamma_{\aR}-2\left(\wp(\aR)-\wp(\xi)\right)U_{\aR,-\aR}\right] 
-\frac{A_1D_1}{B_0^2}-\frac{A_1}{B_0}\rho(\aR)S_{\aR,\aR} \nn \\ 
&& -\frac{D_1}{B_0}\rho(-\aR)S_{-\aR,-\aR}\  ,  
\end{eqnarray}
where $A_1$, $D_1$, $B_0$, $B_1$, $C_1$ are constants\footnote{We note that the solution \eqref{eq:Q2Nsol} of Q2 presented in Theorem \ref{th:cor3} depends effectively on three  coefficients, 
$A_1/B_0$, $D_1/B_0$ and $(B_1+C_1)/B_0$. We could have replaced these coefficients by single symbols, but for the clarity of the proof we have 
abstained from doing that.} and  
with the notations as introduced earlier, i.e. \eqref{eq:U}, \eqref{eq:WS}, \eqref{eq:T} and \eqref{eq:digamma}, and the identifications \eqref{eq:Q2idens}, 
obeys the equation Q2 in the form \eqref{eq:Qeqsb}. 
% with the identifications 
%\be\label{eq:Q2idens}
%\po=\frac{1}{2} \frac{\wp'(\aR)}{p^2-a^2}\quad,\quad \qo=\frac{1}{2} \frac{\wp'(\aR)}{q^2-a^2}\  .    
%\ee 
\end{theorem}
\begin{proof}
The result is obtained by systematic expansion with respect to the parameter $\gm$ as introduced in the proof of Corrollary \ref{th:cor1}. 
The solution \eqref{eq:Nsol} of the Q3 equation under the expansions given earlier takes the form 
\begin{equation}\label{eq:uexpand}
u^{(N)}= \frac{1}{\gm}u_{-1}+u_0+\gm u_1+\mathcal{O}(\gm^2)\  ,
\end{equation}
in which the coefficients $u_{-1}$, $u_0$ and $u_1$ are obtained from the expansions of the various ingredients of the elliptic $N$-soliton 
solution as given earlier. Here we still have the freedom to impose a dependence of the coefficients $A$, $B$, $C$, $D$ on $\gm$ as we want, 
and we shall chose these such that the coefficients $u_{-1}$, $u_0$ are actually \textit{constant}. This can be achieved by chosing  
$$ A=\gm A_1\quad,\quad D=\gm D_1\quad,\quad B=B_0+\gm B_1+\gm^2 B_2\quad,\quad C=-B_0+\gm C_1+\gm^2 C_2\  , $$ 
in which $A_1$, $D_1$, $B_0$, $B_1$, $B_2$, $C_1$, $C_2$ are all assumed to be of order 1 in powers of $\gm$, and which will be specified 
further lateron. Thus, we obtain ~$ u_{-1}= -2B_0$~, ~$u_0=C_1-B_1$~ and $u_1$ given by
$$ u_1=(C_2-B_2)+(C_1+B_1)z + B_0(z'_+ + z'_-)+A_1 \rho(\aR)S_{\aR,\aR}+D_1\rho(-\aR)S_{-\aR,-\aR}\  , $$
where $z$ is given in \eqref{eq:uaa} and $z'_\pm$ in \eqref{eq:zpm}. 
Next we employ the same expansions on the Q3 quadrilateral \eqref{eq:Q3quad}, where for arbitrary coefficients $A$, $B$, $C$, $D$, 
$\Delta$ is given by \eqref{eq:Delta}. Thus, on the equation we can compute the expansion by noting that up to fourth in powers of $\gm$ we have: 
\begin{eqnarray*}
P&=&(p^2-a^2)\left[  1-\frac{1}{2}\gm\frac{\wp'(\aR)}{p^2-a^2} -\frac{1}{4}\gm^2 \left( \frac{\wp''(\aR)}{p^2-a^2} +
\frac{1}{2}\frac{(\wp'(\aR))^2}{(p^2-a^2)^2} \right) \right. \\ 
&& \left. \qquad\qquad\quad -\frac{1}{12}\gm^3 \left(\frac{\wp'''(\aR)}{p^2-a^2} +
\frac{3}{2}\frac{\wp'(\aR)\wp''(\aR)}{(p^2-a^2)^2}  + \frac{3}{4} \frac{(\wp'(\aR))^3}{(p^2-a^2)^3}\right) \right. \\ 
&& \left. \qquad\qquad\quad -\frac{1}{48}\gm^4 \left( \frac{\wp''''(\aR)}{p^2-a^2} +
\frac{3}{2}\frac{(\wp''(\aR))^2}{(p^2-a^2)^2}  + 2\frac{\wp'(\aR)\wp'''(\aR)}{(p^2-a^2)^2}\right. \right. \\ 
&& \left.\left. \qquad\qquad\qquad\qquad\quad +\frac{9}{2} \frac{(\wp'(\aR))^2\wp''(\aR)}{(p^2-a^2)^3}
+ \frac{15}{8} \frac{(\wp'(\aR))^4}{(p^2-a^2)^4}\right)+\mathcal{O}(\gm^5)\,\right]\  ,  \\ 
&& 
\end{eqnarray*}
and similarly for $Q$. Inserting this expansion into the Q3 quadrilateral we obtain up to the required order in powers of $\gm$ 
\begin{eqnarray*}
&& \cQ_{\ssp,\ssq}(u,\wt{u},\wh{u},\wh{\wt{u}})= \\ 
&& (p^2-a^2)(u-\wt{u})(\wh{u}-\wh{\wt{u}})- (q^2-a^2)(u-\wh{u})(\wt{u}-\wh{\wt{u}}) 
- \left(\frac{1}{2}\gm\wp'(\aR)+\frac{1}{4}\gm^2\wp''(\aR)\right)(u-\wh{\wt{u}})(\wh{u}-\wt{u}) \\ 
&& -\frac{1}{8}\left(\gm^2(\wp'(\aR))^2+\gm^3\wp'(\aR)\wp''(\aR)+\frac{1}{4}\gm^4(\wp''(\aR))^2+\frac{1}{3}\gm^4\wp'(\aR)\wp'''(\aR)
\right)\,\left[\frac{u\wh{u}+\wt{u}\wh{\wt{u}}}{p^2-a^2}-\frac{u\wt{u}+\wh{u}\wh{\wt{u}}}{q^2-a^2}\right] \\ 
&& -\frac{1}{16}\left(\gm^3(\wp'(\aR))^3+\frac{3}{2}\gm^4(\wp'(\aR))^2\wp''(\aR)\right)
\,\left[\frac{u\wh{u}+\wt{u}\wh{\wt{u}}}{(p^2-a^2)^2}-\frac{u\wt{u}+\wh{u}\wh{\wt{u}}}{(q^2-a^2)^2}\right] \\ 
&& -\frac{5}{128}\gm^4(\wp'(\aR))^4
\,\left[\frac{u\wh{u}+\wt{u}\wh{\wt{u}}}{(p^2-a^2)^3}-\frac{u\wt{u}+\wh{u}\wh{\wt{u}}}{(q^2-a^2)^3}\right] \\
&& + \det(\cA)(\wp'(\aR))^2\left(\frac{1}{p^2-a^2}-\frac{1}{q^2-a^2}\right) \left\{ 1 +\gm\left[\frac{\wp''(\aR)}{\wp'(\aR)}+
\frac{1}{2}\wp'(\aR)\left(\frac{1}{p^2-a^2}+\frac{1}{q^2-a^2}\right)  \right] \right. \\ 
&& ~~ +\frac{1}{2}\gm^2\left[ \frac{1}{4}(\wp'(\aR))^2\left(\frac{3}{(p^2-a^2)^2}+\frac{3}{(q^2-a^2)^2}+\frac{2}{(p^2-a^2)(q^2-a^2)}\right)
+ \frac{\wp'''(\aR)}{\wp'(\aR)}\right. \\  
&& \left.\left.\qquad\qquad  +\frac{3}{2}\wp''(\aR)\left(\frac{1}{p^2-a^2}+\frac{1}{q^2-a^2}\right)
\right] \right\} %\\ 
%&& \left. \quad + \mathcal{O}(\gm^3)\right\} 
\end{eqnarray*}
Inserting the expansion \eqref{eq:uexpand}, we find that terms of of order 0 and 1 in powers of $\gamma$ (there are no 
negative powers) cancel, provided that 
$$ (\det(\cA))_0=\frac{1}{4}u_{-1}^2\quad,\quad (\det(\cA))_1=\frac{1}{2}u_{-1}u_0\  , $$ 
where $(\det(\cA))_i$ indicates the i$^{\rm th}$ order of the expansion of the determinant, and these hold identically true for
the choices of $u_{-1}$ and $u_0$ given above. The dominant order is then quadratic in $\gm$ and yields the following relation 
\begin{eqnarray}
&& \cQ_{\ssp,\ssq}(u,\wt{u},\wh{u},\wh{\wt{u}})= \nn \\ 
&& = \gm^2\bigg\{ (p^2-a^2)(u_1-\wt{u}_1)(\wh{u}_1-\wh{\wt{u}}_1)- (q^2-a^2)(u_1-\wh{u}_1)(\wt{u}_1-\wh{\wt{u}}_1)  \nn \\ 
&& \qquad + \left(\frac{1}{p^2-a^2}-\frac{1}{q^2-a^2}\right)(\wp'(\aR))^2\left[ (\det(\cA))_2-\frac{1}{8}
\left(2u_0^2+u_{-1}(u_1+\wt{u}_1+\wh{u}_1+\wh{\wt{u}}_1)\right) \right. \nn \\ 
&& \qquad +\frac{1}{24} u_{-1}^2 \left(\frac{\wp'''(\aR)}{\wp'(\aR)}-\frac{3}{2} \frac{(\wp''(\aR))^2}{(\wp'(\aR))^2} \right)  \nn \\  
% && \qquad + \frac{1}{64} u_{-1}^2 (\wp'(\aR))^2 \left(\frac{1}{(p^2-a^2)^3}-\frac{2}{(p^2-a^2)^2(q^2-a^2)}
% +\frac{2}{(p^2-a^2)(q^2-a^2)^2}-\frac{1}{(q^2-a^2)^3}\right) \nn \\ 
% && \qquad -\frac{1}{8}(\wp'(\aR))^2 \left[ 2u_0^2+u_{-1}(u_1+\wt{u}_1+\wh{u}_1+\wh{\wt{u}}_1)\right]\left(\frac{1}{p^2-a^2}-\frac{1}{q^2-a^2}\right) \nn \\ 
% && \qquad +\frac{1}{24}\left[\wp'(\aR)\wp'''(\aR)-\frac{3}{2}(\wp''(\aR))^2 \right]u_{-1}^2\left(\frac{1}{p^2-a^2}-\frac{1}{q^2-a^2}\right) \nn \\ 
&& \qquad + \left.\frac{1}{64} u_{-1}^2 (\wp'(\aR))^2 \left(\frac{1}{(p^2-a^2)^2}-\frac{1}{(p^2-a^2)(q^2-a^2)}
+\frac{1}{(q^2-a^2)^2}\right) \right]\bigg\} + \mathcal{O}(\gm^3)\  . \nn \\ 
%  && \quad \left. + (\det(\cA))_2 (\wp'(\aR))^2 \left(\frac{1}{p^2-a^2}-\frac{1}{q^2-a^2}\right)\right\} + \mathcal{O}(\gm^3)\  .  
\end{eqnarray}
Thus, if $\cQ_{\ssp,\ssq}(u,\wt{u},\wh{u},\wh{\wt{u}})=0$ as for the Q3 solution we started out with, then the dominant coefficient in the expansion in 
powers of $\gm$ must vanish as well, leading to the conclusion that with the identifications \eqref{eq:Q2idens}, as well as setting 
\be\label{eq:idens} 
%&& \po=\frac{1}{2} \frac{\wp'(\aR)}{p^2-a^2}\quad,\quad \qo=\frac{1}{2} \frac{\wp'(\aR)}{q^2-a^2}  \\ 
v=2\frac{u_1}{u_{-1}}+\frac{u_0^2-4(\det(\cA))_2}{u_{-1}^2}-\frac{1}{6}\left(\frac{\wp'''(\aR)}{\wp'(\aR)}-\frac{3}{2} \frac{(\wp''(\aR))^2}{(\wp'(\aR))^2} \right)
\ee
the function $v$ obeys the Q2 equation in the form \eqref{eq:Qeqsb} with the identifications \eqref{eq:Q2idens}. 
\end{proof}

\subsection{H-degenerations} 

The degenerate limits from the solutions of Q3 to those of the H-equations are obtained through the limit $\bb\rightarrow 0$. 
As an immediate corollary of theorem \ref{th:theorem1} we have the following result: 

\begin{theorem}\label{th:cor3} The following formula 
 \begin{equation} \label{eq:H3Nsol} 
v^{(N)} = i^{n+m}\left[\left(A+B(-1)^{n+m}\right)\rho^{1/2}(\aR)W_{\aR}+\left(C+D(-1)^{n+m}\right)\rho^{1/2}(-\aR)W_{-\aR}\right]\  , 
\end{equation} 
in which $W_{\pm\aR}$ and $\rho(\pm\aR)$ are defined as before, %and where the coefficients are subject to the relation 
%$$ A^2-B^2=\frac{\Delta}{\wp'(\aR)^2} \  , $$ 
constitutes the elliptic $N$-soliton solution of the $({\rm H}_3)_\Delta$ equation, given in the form
\be\label{eq:H3Delta}
P_\aR(v\wh{v}+\wt{v}\wh{\wt{v}})-Q_\aR(v\wt{v}+\wh{v}\wh{\wt{v}})+2(p^2-q^2)\frac{\Delta_\aR}{P_\aR Q_\aR}=0\  . 
\ee 
where ~$P_\aR=\sqrt{a^2-p^2}$~,~$Q_\aR=\sqrt{a^2-q^2}$~, and  where 
$$\Delta_\aR=\wp'(\aR)\det(\cA)\   . $$
Furthermore, the associated solution 
\be\label{eq:H3assNsol} 
V^{(N)}=i^{n+m}\left[\left(A-B(-1)^{n+m}\right)\rho^{1/2}(\aR)W_{\aR}+\left(C-D(-1)^{n+m}\right)\rho^{1/2}(-\aR)W_{-\aR}\right]\  ,  
\ee 
together with $v^{(N)}$ are Miura related to the elliptic $N$-soliton solution $w$ of H1 via the relations:
\bse\label{eq:wvMiura}\begin{eqnarray}
\wh{w}-\wt{w}&=&\frac{P_\aR\wh{v}-Q_\aR\wt{v}}{\wh{\wt{V}}}
 =-\frac{P_\aR\wt{v}-Q_\aR\wh{v}}{V} \label{eq:wvMiuraa} \\ 
w-\wh{\wt{w}}&=&\frac{P_\aR v-Q_\aR\wh{\wt{v}}}{\wt{V}}
 =-\frac{P_\aR\wh{\wt{v}}-Q_\aR v}{\wh{V}} \label{eq:wvMiurab} 
\end{eqnarray}\ese  
\end{theorem}
\begin{proof}
The proof is by simple expansions of the elliptic functions in the solution \eqref{eq:Nsol}, namely by observing that as 
$\bb\rightarrow 0$, we have 
\begin{equation} 
%&& \rho(\bb)\sim\left(\frac{\sg(\bb-\dd)}{\sg(\bb+\dd)}\right)^n\left(\frac{\sg(\bb-\ven)}{\sg(\bb+\ven)}\right)^m\, e^{2\zeta(\xi)\bb}~ \\ 
%&& ~ \rightarrow~ (-1)^{n+m}\left[ 1+2\bb\left( \zeta(\xi)-n\zeta(\dd)-m\zeta(\ven) \right)+2\bb^2\left( \zeta(\xi)-n\zeta(\dd)-m\zeta(\ven)\right)^2  
%+ \mathcal{O}(\bb^3)\right] \\   
%&& ~ \Rightarrow~ 
\rho^{1/2}(\pm\bb)\sim (\pm i)^{n+m}\left[ 1\pm\bb\left( \zeta(\xi)-n\zeta(\dd)-m\zeta(\ven) \right)+\frac{1}{2}\bb^2\left( \zeta(\xi)-n\zeta(\dd)-m\zeta(\ven)\right)^2  
+ \mathcal{O}(\bb^3)\right]\   ,     
\end{equation} 
choosing an appropriate branch of the square root of $\rho(\bb)$, 
as well as 
\begin{eqnarray*}
&& \chi_{\aR,\bb}~\rightarrow~\frac{1}{\bb} -\eta_{\aR} + \bb\wp(\xi+\aR)+\mathcal{O}(\bb^2)
\quad,\quad \Psi_\xi(\aR+\bb)~\rightarrow~\Psi_\xi(\aR)\left( 1+\eta_{\aR}\bb +\mathcal{O}(\bb^2)\right)\  , \\ 
&& \chi_{\bb,\bK}~\rightarrow~\frac{1}{\bb} -\eta_{\bK} +\bb\wp(\xi+\bK)+ \mathcal{O}(\bb^2)
\quad \Rightarrow\quad U_{\aR,\bb}~\rightarrow~\bb(1-V_\aR)+
\bb^2 Z_{\aR}+\mathcal{O}(\bb^3)\  , \\  
\end{eqnarray*}
in which 
\be\label{eq:Z}
Z_{\aR}=\bs^T \eta_{\bK} (\bun+\bM)^{-1} \chi_{\aR,\bK}^{-1}\,\brr\quad {\rm with}\quad \eta_{\bK}=\zeta(\xi+\bK)-\zeta(\xi)-\zeta(\bK)\   . 
\ee 
Hence, 
\be\label{eq:Slim} 
S_{\aR,\bb}=\Psi_\xi(\aR+\bb)\left[1-\chi_{\aR,\bb}U_{\aR,\bb}\right]~\rightarrow~W_{\aR}+\bb\Psi_\xi(\aR)(\eta_{\aR}-Z_{\aR})+\mathcal{O}(\bb^2) 
\ee 
and, thus, we get the limiting behaviour 
\bse\bea\label{eq:Nsollim} 
&& u^{(N)}\rightarrow v^{(N)}+\bb v' + \mathcal{O}(\bb^2)  \\ 
&& v'\Equiv  i^{n+m}\left\{ \left(A-B(-1)^{n+m}\right)\rho^{1/2}(\aR)\left[ \Psi_\xi(\aR)(\eta_{\aR}-Z_{\aR})+
\left(\zeta(\xi)-n\zeta(\dd)-m\zeta(\ven)\right)W_{\aR}\right] \right. \nn \\ 
&& \qquad \left. +\left(C-D(-1)^{n+m}\right)\rho^{1/2}(-\aR)\left[\Psi_\xi(-\aR)(\eta_{-\aR}-Z_{-\aR})+
\left(\zeta(\xi)-n\zeta(\dd)-m\zeta(\ven)\right)W_{-\aR}\right]\right\}\   , \nn \\ 
&&  \label{eq:Nsollimb} 
\eea\ese 
as given in \eqref{eq:H3Nsol}. In the quadrilateral \eqref{eq:Q3quad} we have the following limits: 
$$ P\rightarrow \frac{1}{\bb}P_\aR+\mathcal{O}(\bb)\quad,\quad Q\rightarrow \frac{1}{\bb}Q_\aR+\mathcal{O}(\bb)\quad {\rm and}\quad 
\Delta=\wp'(\aR)\wp'(\bb)\det(\cA)\rightarrow \frac{-2}{\bb^3}\Delta_\aR\  , $$ 
so that the dominant behaviour of the quadrilateral \eqref{eq:Q3quad} reduces to the left-hand side of \eqref{eq:H3Delta}. Finally, with 
~$\Psi_\xi(\pm\bb)\rightarrow\pm\bb^{-1}\left(1-\frac{1}{2}\wp(\xi)\bb^2+\mathcal{O}(\bb^3)\right)$~ we get the expansions 
\begin{eqnarray*} 
&& W_\bb=\Psi_\xi(\bb)(1-V_\bb)~\rightarrow \frac{1}{\bb}\left[1 -\bb (\bs^T\,\bu^0)-\bb^2\left((\bs^T\,\eta_{\bK}\,\bu^0)+\frac{1}{2}\wp(\xi)\right)
+\mathcal{O}(\bb^3)\right] \\ 
&& \quad\Rightarrow\quad \rho^{1/2}(\pm\bb)W_{\pm\bb}~\rightarrow~\pm\frac{1}{\bb}(\pm i)^{n+m}\left( 1\pm\bb w + \frac{1}{2}\bb^2 w' +\mathcal{O}(\bb^3) \right)\  ,  
\end{eqnarray*} 
where $w$ is, once again, the solution \eqref{eq:wnew} of the equation H1, and where 
\be\label{eq:w'}
w'\Equiv  \left(\zeta(\xi)-n\zeta(\dd)-m\zeta(\ven)-\bs^T\,\bu^0\right)^2-\wp(\xi)-2(\bs^T\,\eta_{\bK}\,\bu^0)-(\bs^T\,\bu^0)^2\   .  
\ee 
This implies, in particular, that  
\be\label{eq:Ulim} 
U^{(N)}~\rightarrow~(\bb^{-1}+\frac{1}{2}\bb w')V^{(N)}+ wv^{(N)}+\mathcal{O}(\bb^2)\   ,  
\ee
from which in turn follow that the relations \eqref{eq:wuUMiura} through these expansion reduce to the relations given in \eqref{eq:wvMiura} as $\bb\rightarrow 0$.  
\end{proof}

The quantity $v'$ that is obtained in the higher orders of \eqref{eq:Nsollim} obeys also some interesting relations, which can be easily obtained from 
\eqref{eq:wuUMiura} by pushing the expansions one step further. Thus, as a corollary, we obtain 

\begin{cor}\label{th:cor11}
The quantity $v'$ defined in \eqref{eq:Nsollimb} obeys the following Miura type relations 
\bse\label{eq:v'wv}\begin{eqnarray}
&& P_{\aR}\wh{v}'-Q_{\aR}\wt{v}'=(\wh{w}-\wt{w})w\wh{\wt{v}}\quad ,\quad  P_{\aR}v'-Q_{\aR}\wh{\wt{v}}'=(w-\wh{\wt{w}})\wh{w}\wt{v}\   , \\ 
&& P_{\aR}\wt{v}'-Q_{\aR}\wh{v}'=(\wt{w}-\wh{w})v\wh{\wt{w}}\quad ,\quad  P_{\aR}\wh{\wt{v}}'-Q_{\aR}v'=-(w-\wh{\wt{w}})\wt{w}\wh{v}\   ,
\end{eqnarray}\ese 
where $v$ is the solution of the equation H3 as given in \eqref{eq:H3Nsol} and where $w$ is the solution \eqref{eq:wnew} of the equation H1 and where 
the parameters $P_\aR$, $Q_\aR$ are given as in Theorem \ref{th:cor3}.    
\end{cor}

It is noteworthy that the quantity $w'$, i.e. \eqref{eq:w'}, that appears in the higher orders of the expansions is, in fact, the object that solves the H2 equation. 
Thus, we can state the following:

\begin{theorem}\label{th:H2theorem} The quantity $w'$ as defined in \eqref{eq:w'}, i.e. 
$$ w'=w^2-\wp(\xi)-2(\bs^T\,\eta_{\bK}\,\bu^0) -(\bs^T\,\bu^0)^2\  , $$ 
is a solution of the H2 equation, i.e. it solves the set of relations  
\be\label{eq:H2BT}
w'+\wt{w}'-p^2=2w\wt{w}\quad,\quad  w'+\wh{w}'-q^2=2w\wt{w}\   , 
\ee
which is a Miura relation (non-auto-BT) between H1 and H2, where $w$ is a solution of H1, and where $w'$ is a solution of 
\be\label{eq:H2}
(w'+\wt{w}'-p^2)(\wh{w}'+\wh{\wt{w}}'-p^2)=(w'+\wh{w}'-q^2)(\wt{w}'+\wh{\wt{w}}'-q^2)\   ,
\ee
which coincides with \eqref{eq:Heqsb}.
\end{theorem}
\begin{proof}
To prove Theorem \ref{th:H2theorem} we proceed in s similar fashion as in section 4. From \eqref{eq:bu0}, using 
\eqref{eq:rho}, we have using the definitions: 
\begin{eqnarray*}
&& p_{\bK}e^{-\eta_\dd\bK} (\bun+\wt{\bM}) \wt{\bu}^0 = p_{\bK} e^{-\eta_\dd\bK}\,\wt{\brr} \\ 
&& \qquad = e^{-\eta_\dd\bK} e^{-2\zeta(\dd)\bK}p_{-\bK} e^{2[\zeta(\wt{\xi})-\zeta(\xi)]\bK}\frac{\Psi_{\wt{\xi}}(\bK)}{\Psi_\xi(\bK)}\,\brr \\ 
&& \qquad = e^{\eta_\dd\bK} \Psi_\dd(-\bK)\frac{\Psi_{\wt{\xi}}(\bK)}{\Psi_\xi(\bK)}\,\brr= 
-\left[\zeta(\wt{\xi})-\zeta(\dd)+\zeta(\bK)-\zeta(\xi+\bK)\right]\,\brr  \\ 
&& \qquad = (\eta_{\bK}-\eta_\dd)\,\brr = \left[ (\bun+\bM)p_{\bK}e^{-\eta_\dd\bK} - \brr\,\wt{\bs}^T\right] \wt{\bu}^0\  , 
\end{eqnarray*}
where the last step follows by applying \eqref{eq:pMb}. Multiplying both sides from the left by ~$(\bun+\bM)^{-1}$~ we obtain the relation 
\bse\label{eq:bu00}\be\label{eq:bu00a}
p_{\bK} e^{-\eta_\dd\bK} \wt{\bu}^0 = (\bun+\bM)^{-1}\,\eta_{\bK}\,\brr -\left(\eta_\dd-(\wt{\bs}^T\,\wt{\bu}^0)\right)\,\bu^0\   . 
\ee 
In a similar way, using \eqref{eq:pMa}, we obtain the relation
\be\label{eq:bu00b}
p_{-\bK} e^{\eta_\dd\bK} \bu^0 = -(\bun+\wt{\bM})^{-1}\,\wt{\eta}_{\bK}\,\wt{\brr} -\left(\eta_\dd+(\bs^T\,\bu^0)\right)\,\wt{\bu}^0\   , 
\ee \ese 
taking into account that from the definitions we have ~$\wt{\eta}_{-\dd}=-\eta_\dd$~. Multiplying \eqref{eq:bu00a} from the left by $\bs^T$ 
we now compute 
\begin{eqnarray*}
\bs^T p_{\bK} e^{-\eta_\dd\bK} \wt{\bu}^0 &=& \wt{\bs}^T \frac{\Psi_\xi(\bK)}{\Psi_{\wt{\xi}}(\bK)}\Psi_\dd(\bK) e^{-\eta_\dd\bK} \wt{\bu}^0 \\ 
&=& \wt{\bs}^T \left[ \zeta(\xi)+\zeta(\bK)+\zeta(\dd)-\zeta(\wt{\xi}+\bK)\right] \wt{\bu}^0=-\wt{\bs}^t (\eta_\dd+\wt{\eta}_{\bK}) \wt{\bu}^0 \\ 
&=& \bs^T (\bun+\bM)^{-1} \eta_{\bK}\,\brr -\left(\eta_\dd -(\wt{\bs}^T\,\wt{\bu}^0) \right)\bs^T\bu^0 \  ,  
\end{eqnarray*}
and hence we obtain
\be\label{eq:bsu00} 
-\eta_\dd (\wt{\bs}^T\,\wt{\bu}^0)-\wt{\bs}^T\,\wt{\eta}_{\bK}\,\wt{\bu}^0=\tbu^0\,\eta_{\bK}\,\brr -\left(\eta_\dd -(\wt{\bs}^T\,\wt{\bu}^0) \right)(\bs^T\,\bu^0)\   .  
\ee 
The same equation is obtained by multiplying \eqref{eq:bu00b} from the left by $\wt{\bs}^T$ and performing a similar calculation. 
Using the abbreviations ~$w_0\Equiv (\bs^T\,\bu^0)$~, ~$w_1\Equiv (\bs^T\,\eta_{\bK}\bu^0)$~ we can rewrite \eqref{eq:bsu00} more conveniently in 
the form: 
$$ w_1+\wt{w}_1=\eta_\dd w_0 -\eta_\dd\wt{w}_0-w_0 \wt{w}_0\quad\Rightarrow\quad 2(w_1+\wt{w}_1)+w_0^2+\wt{w}_0^2+\eta_\dd^2=(w_0-\wt{w}_0+\eta_\dd)^2\  .   $$
Using the addition formula for the Weierstrass elliptic function
$$ \eta_\dd^2=\left(\zeta(\xi+\dd)-\zeta(\xi)-\zeta(\dd)\right)^2= \wp(\xi+\dd)+\wp(\xi)+\wp(\dd)  $$ 
and noting that from the definition \eqref{eq:wnew} we have that ~$\wt{w}-w=\eta_\dd+w_0-\wt{w}_0$~, we finally obtain;
\be\label{eq:bsu00alt}
\left(\wp(\xi)+w_0^2+2w_1-w^2\right) + \left(\wp(\wt{\xi})+\wt{w}_0^2+2\wt{w}_1-\wt{w}^2\right)+\wp(\dd)=-2w\,\wt{w}
\ee 
which after identifying the expressions between brackets with the object $w'$ in \eqref{eq:w'} leads to the first part of the Miura transform \eqref{eq:H2BT} 
setting $\wp(\dd)=p^2$. The other part of the Miura transform follows by the usual replacements, and subsequently the H2 equation for $w'$ is derived by 
applying the trivial identity ~$(2w\wt{w})(2w\wt{w})\!\wh{\phantom{a}}=(2w\wh{w})(2w\wh{w})\!\wt{\phantom{a}}$~.  
\end{proof}

An alternative proof of Theorem \ref{th:H2theorem}, using the expansion \eqref{eq:Ulim}, can be readily given, but we believe it is instructive 
to provide a direct verification based on the form of the solution. We note that the Miura relations \eqref{eq:H2BT} were first given in 
\cite{Atkinson}. 

\section*{Acknowledgements} 
The authors are grateful for the hospitality of the Isaac Newton Institute for Mathematical Sciences, Cambridge, 
where the present work was commenced during the programme Discrete Integrable Systems (DIS). 
JA was supported by the Australian Research Council Discovery Grant DP 0985615.

\end{document}